\PassOptionsToPackage{capitalise}{cleveref}
\PassOptionsToPackage{draft}{todonotes}
\PassOptionsToPackage{svgnames,dvipsnames}{xcolor}
\documentclass[a4paper,UKenglish,cleveref,autoref, thm-restate, orivec,numberwithinsect]{lipics-v2021}

\Crefname{section}{\S}{\S\S}
\Crefname{definition}{Def.}{Defs.}
\Crefname{example}{Ex.}{Exs.}
\Crefname{theorem}{Thm.}{Thms.}
\Crefname{proposition}{Prop.}{Props.}
\Crefname{figure}{Fig.}{Figs.}

\newif\iftr\trtrue    %%% Comment next line for the TR
\iftr\nolinenumbers
\fi

\title{
  Design-by-Contract for \emph{Flexible} Multiparty Session Protocols
  \iftr --- Extended Version \fi
}
\titlerunning{
  Design-by-Contract for \emph{Flexible} Multiparty Session Protocols%A Design-by-Contract Approach for Choreographies
}

\author{Lorenzo Gheri}{
  Imperial College London, UK \and
  \url{https://sites.google.com/view/lorgheri/home} }{l.gheri@imperial.ac.uk
}{https://orcid.org/0000-0002-3191-7722}{}

\author{Ivan Lanese}{
  Focus Team, University of Bologna/INRIA (Italy)
  \and \url{https://www.unibo.it/sitoweb/ivan.lanese/}
}{
  ivan.lanese@gmail.com
}{
  https://orcid.org/0000-0003-2527-9995
}{}
\author{Neil Sayers}{
  Imperial College London, UK \and Coveo Solutions Inc., Canada
}{
  sayers.neil@gmail.com
}{
  https://orcid.org/0000-0003-4718-7290
}{}

 \author{Emilio Tuosto}{
  Gran Sasso Science Institute, Italy \and
  \url{https://cs.gssi.it/emilio.tuosto}
}{
  emilio.tuosto@gssi.it
}{
  https://orcid.org/0000-0002-7032-3281
}{}
\author{Nobuko Yoshida}{
  Imperial College London, UK \and
  \url{https://www.imperial.ac.uk/people/n.yoshida}
}{
  n.yoshida@imperial.ac.uk
}{
   https://orcid.org/0000-0002-3925-8557
}{}%% {
\authorrunning{L. Gheri, I. Lanese, N. Sayers, E. Tuosto, and N. Yoshida}

\Copyright{Lorenzo Gheri, Ivan Lanese, Neil Sayers, Emilio Tuosto, Nobuko Yoshida} %TODO mandatory, please use full first names. LIPIcs license is "CC-BY";  http://creativecommons.org/licenses/by/3.0/

 %TODO mandatory: Please choose ACM 2012 classifications from https://dl.acm.org/ccs/ccs_flat.cfm

 %TODO mandatory; please add comma-separated list of keywords

\supplement{Software: ECOOP 2022 Artifact Evaluation approved artifact, also available at \url{https://github.com/Tooni/CAScript-Artifact}
}%optional, e.g. related research data, source code, ... hosted on a repository like zenodo, figshare, GitHub, ...
\funding{\tnxbehapi \tnxitmatters  Lanese and Tuosto are partially supported by INdAM as members of GNCS (Gruppo Nazionale per il Calcolo Scientifico).
 The work is partially supported by EPSRC EP/T006544/1,
EP/K011715/1, EP/K034413/1, EP/L00058X/1,
EP/N027833/1, EP/N028201/1, EP/T014709/1
and EP/V000462/1, and NCSS/EPSRC
VeTSS.}%optional, to capture a funding statement, which applies to all authors. Please enter author specific funding statements as fifth argument of the \author macro.

\acknowledgements{We thank the anonymous reviewers for their useful comments and suggestions. We thank Fangyi Zhou for their help with building our artifact on top of their software, $\nu$\textsf{Scr}.}%optional

\iftr \else
\EventEditors{Karim Ali and Jan Vitek}
\EventNoEds{2}
\EventLongTitle{36th European Conference on Object-Oriented Programming (ECOOP 2022)}
\EventShortTitle{ECOOP 2022}
\EventAcronym{ECOOP}
\EventYear{2022}
\EventDate{June 6--10, 2022}
\EventLocation{Berlin, Germany}
\EventLogo{}
\SeriesVolume{222}
\ArticleNo{8}
\fi
\usepackage{hyperref}
\usepackage{xspace}
\usepackage{listings}
\usepackage{multicol}
\usepackage{stmaryrd}
\usepackage{nicefrac}
\usepackage{proof}
\usepackage{mathrsfs}
\usepackage{caption}
\usepackage{subcaption}
\usepackage{amsmath}
\usepackage{xargs}
\usepackage{xparse}
\usepackage[appendix=append]{apxproof}
\usepackage{todonotes}
\usepackage{tikz}                   % Graphs
\usetikzlibrary{shadows,arrows,shapes,automata,positioning,decorations.markings}
\tikzset{
  cnode/.style={
    shape=circle,
    minimum size = 0mm,
    inner sep = 1pt,
    font=\tiny,
    draw
  },
  carrow/.style={
    ->,
    shorten >=1pt,
    >=stealth',
    auto,
    font=\scriptsize,
    draw,
    sloped
  }
}

\ifdefined\NOCOLOR
   \newcommand{\withcolor}[2]{#2} % no color selsected
   \else
   \newcommand{\withcolor}[2]{\colorlet{currbkp}{.}\color{#1}{#2}\color{currbkp}}
   \fi

\newcommand{\code}[1]{\texttt{\upshape #1}} % typeset with mono spaced font

\newcommand{\bnfas}{\mathrel{::=}}
\newcommand{\bnfalt}{\mathrel{\mid}}

\newcommand{\colorse}{DarkGreen} % colour for session vars
\newcommand{\colorgp}{violet} % colour for global types

\newcommand{\dgt}[1]{\withcolor{\colorgp}{#1}}

\newcommand{\gtvar}[1]{\dgt{\mathbf{#1}}}
\newcommand{\gtrecur}[2]{\dgt{\mu\mathbf{#1}.#2}}
\newcommand{\gtend}{\dgt{\code{end}}}

\newcommand{\gtsigma}[1]{\dgt{\mathbf{\Sigma}}_{#1}}

\newcommandx{\gint}[3][1=p, 2=q, 3=m, usedefault=@]{{\ptp[#1] \to \ptp[#2]: \amsg[{#3}]}}
\newcommandx{\knw}[3][1 = {\chora}, 2={p}, 3=t, usedefault=@]{\mathsf{knw}_{#1}({\ptp[#2],{#3}})}

\newcommand{\arro}[2][{}]{\xrightarrow[{#1}]{\ {#2}\ }}
\newcommand{\cfsmtrans}[3]{#1 \arro{#2} #3}
\newcommand{\aint}[1][a]{\lcgreek{#1}}
\newcommand{\chora}[1][A]{\mathsf{C{#1}}}
\newcommand{\reccall}[2]{\gtvar{#1} \cdot {#2}}

\newcommand{\lint}{\lset_{\text{int}}}
\newcommand{\lact}{\lset_{\text{act}}}

\usepackage{xifthen}        % for conditional commands
\newcommand{\ifempty}[3]{%
  \ifthenelse{\isempty{#1}}{#2}{#3}%
}

\def\colorFun{\color{black}}
\newcommand{\mkfun}[4][\colorFun]{
  \newcommandx{#2}[2][1=#4,2={}]{
    {\textsf{#1#3}}_{##2}
    \ifempty{##1}{}{
      ({##1})}
  }
}

\mkfun{\ptpof}{ptp}{{\sigma}}

\mkfun{\lang}{L}{{\sigma}}

\mkfun{\pref}{pref}{{\sigma}}

\mkfun{\ca}{ca}{{\dgt{G},q}}

\mkfun{\cat}{ca}{{\dgt{G}}}

\mkfun{\catpass}{caPass}{{\dgt{G}}}
\mkfun{\catr}{catr}{{\dgt{G}}}

\mkfun{\catrpass}{catrPass}{{\dgt{G}}}

\mkfun{\traceof}{trace}{\arun}

\mkfun{\openof}{$\mathbb{O}$}{\chora}

\mkfun{\cycleof}{$\mathbb{C}$}{\chora}

\mkfun{\partof}{$\mathbb{P}$}{\arun}

\mkfun{\fvar}{var}{}
\mkfun{\bvar}{bvar}{}

\def\colorMsg{\color{BrickRed}}
\newcommand{\amsg}[1][m]{\mathsf{\colorMsg{#1}}}
\newcommand{\proj}[2]{#1\!\downarrow_{{\ptp[#2]}}}

\newcommand{\set}[1]{\{#1\}}

\newcommand{\autom}[4]{(#1,#2,#3,#4)}

\newcommand{\trans}[3]{(#1,#2,#3)}

\newcommand{\transset}[2]{\textit{out}(#1,#2)}

\newcommand{\lset}{\mathcal{L}}
\newcommand{\tset}{\mathcal{T}}
\newcommand{\msgset}{\mathcal{M}}
\newcommand{\assertionset}{\mathcal{A}}

\newcommand{\sst}{\;\big|\;}
\newcommand{\qst}{\;\colon\;} %such that

\newcommand{\aCM}{M}
\newcommand{\aM}{\aCM}
\newcommand{\aCS}[1][S]{\mathsf{#1}}

\def\colorPtp{\color{blue}}
\newcommand{\ptpset}{\mathcal{\colorPtp{P}}}
\newcommand{\ptpuniv}{\mathfrak{P}}

\newcommandx{\aout}[3][1={\p},2={\q},3=m,usedefault=@]{\ptp[#1] \ptp[#2] ! \amsg[{#3}]}
\newcommandx{\ain}[3][1={\p},2={\q},3=m,usedefault=@]{\ptp[#1] \ptp[#2] ? \amsg[{#3}]}

\newcommand{\ssem}[1]{{\llbracket #1 \rrbracket}}

\newcommand{\aConf}{s}
\newcommand{\pstate}[2]{\aConf_{#1}(#2)}

\newcommand{\quo}[1]{\lq\lq {#1}\rq\rq}
\newcommand{\qqand}[1][and]{\qquad\text{#1}\qquad}
\newcommand{\qand}[1][and]{\quad\text{#1}\quad}

\newcommand{\apred}[1][A]{\mathsf{\textcolor{orange}{#1}}}
\newcommand{\predB}{\apred[B]}
\newcommand{\ENTAILS}{\supset}
\newcommand{\EXISTS}[2]{\exists {#1}  \qst {#2}}

\newcommand{\asort}[1][s]{\textcolor{IndianRed}{\textsf{#1}}}
\newcommand{\varid}[1][v]{\textcolor{MediumBlue}{\textsf{#1}}}

\newcommand{\truek}{\top}
\newcommand{\falsek}{\bot}
\newcommand{\arun}{\pi}
\newcommand{\conf}[1]{\ensuremath{\left\langle {#1} \right\rangle}}
\newcommand{\confel}{\aConf}
\newcommand{\sset}{\mathbb{S}}
\newcommand{\Set}[1]{\left\{\,#1\,\right\}}
\newcommand{\rcal}{\mathcal{R}}
\newcommand{\aword}{w}

\ExplSyntaxOn
\NewDocumentCommand{\ucgreek}{m}
 {
  \str_case:nn { #1 }
   {
    {A}{\mathrm{A}}
    {B}{\mathrm{B}}
    {C}{\Sigma}
    {D}{\Delta}
    {E}{\mathrm{E}}
    {F}{\Phi}
    {G}{\Gamma}
    {H}{\mathrm{H}}
    {I}{\mathrm{I}}
    {J}{\Theta}
    {K}{\mathrm{K}}
    {L}{\Lambda}
    {M}{\mathrm{M}}
    {N}{\mathrm{N}}
    {O}{\mathrm{O}}
    {P}{\Pi}
    {Q}{\mathrm{X}}
    {R}{\mathrm{P}}
    {S}{\Sigma}
    {T}{\mathrm{T}}
    {U}{\Upsilon}
    {W}{\Omega}
    {X}{\Xi}
    {Y}{\Psi}
    {Z}{\mathrm{Z}}
   }
 }
\NewDocumentCommand{\lcgreek}{m}
 {
  \str_case:nn { #1 }
   {
    {a}{\alpha}
    {b}{\beta}
    {c}{\varsigma}
    {d}{\delta}
    {e}{\varepsilon}
    {f}{\varphi}
    {g}{\gamma}
    {h}{\eta}
    {i}{\iota}
    {j}{\vartheta}
    {k}{\kappa}
    {l}{\lambda}
    {m}{\mu}
    {n}{\nu}
    {o}{o}
    {p}{\pi}
    {q}{\chi}
    {r}{\rho}
    {s}{\sigma}
    {t}{\tau}
    {u}{\upsilon}
    {w}{\omega}
    {x}{\xi}
    {y}{\psi}
    {z}{\zeta}
   }
 }
\ExplSyntaxOff

\newcommand{\tnxbehapi}[1][partly]{
  Research {#1} supported by the EU H2020 RISE programme under the
  Marie Sk{\l}odowska-Curie grant agreement No 778233.
}

\newcommand{\tnxitmatters}[1][Work partially funded]{
   #1 by MIUR project PRIN 2017FTXR7S \emph{IT MATTERS}
  (Methods and Tools for Trustworthy Smart Systems).
}

\newcommand{\eqdef}{\triangleq}
\newcommand{\mmdef}{\eqdef}

\newcommand{\emptyword}{\varepsilon}

\newcommand{\ptp}[1][p]{\withcolor{\colorse}{\ensuremath{{\mathsf{\lowercase{#1}}}}}} % participants
\newcommand{\p}{{\ptp}\xspace}
\newcommand{\q}{{\ptp[q]}\xspace}

\makeatletter
\newcommand{\dolist}[2]{%
  \def\nextitem{\def\nextitem{#1}}%
  \@for \el:=#2\do{\nextitem\el}%
}
\makeatother

\newcommand{\tuplevar}[1][V]{\mathtt{#1}}
\newcommand{\atag}[1][t]{\mathtt{\colorMsg \lcgreek{#1}}}
\newcommand{\atuple}[2][t]{
  \atag[#1]\conf{
	 \dolist{#2}{,}
  }
}

\newcommandx{\fexists}[2][
  1=A, 2=X, usedefault=@
]{\exists_{\setminus\varid[{#2}]}{\apred[{#1}]}}
\newcommand{\recctx}{\rho}

\newcommand{\pcirc}{\circ}
\newcommandx{\predcompose}[2][
  1=A, 2=B, usedefault=@
]{{\apred[{#1}]} \pcirc {\apred[{#2}]}}

\mkfun{\assertionof}{\mbox{$\mathbb{A}$}}{{\arun}}
\mkfun{\assertionnd}{\mbox{$\Delta$}}{{X,Y}}
\mkfun{\assertionofcm}{\mbox{$\mathbb{A}$}}{{\arun}}
\mkfun{\preof}{\mbox{$\mathbb{P}$}}{q}
\mkfun{\enablingof}{\mbox{$\mathbb{E}$}}{q}
\mkfun{\eclos}{\mbox{$\emptyword$-clos}}{q}
\mkfun{\cmderiv}{\mbox{$\partial$}}{q}
\mkfun{\pnf}{\text{pnf}}{\apred}
\mkfun{\simpleruns}{\text{SPath}}{}

\def\finex{{\unskip\nobreak\hfil
\penalty50\hskip1em\null\nobreak\hfil$\diamond$
\parfillskip=0pt\finalhyphendemerits=0\endgraf}}

\usepackage{listings}
\usepackage{mfirstuc}
\usepackage{xstring}

\newcommand{\gsubs}[2]{^{#1} / _{#2}}
\newcommandx{\gsubst}[3][1=\aM,2=q,3=q',usedefault=@]{
  \left \{\gsubs{#3}{#2} \right \}#1
}
\newcommandx{\gsubsts}[5][1=\aM,2=q,3=q',4=q,5=q',usedefault=@]{
  \left \{\gsubs{#3}{#2}, \gsubs{#5}{#4} \right \}#1
}

\newcommandx{\wellpar}[2][1={\aG},2={\aG'},usedefault=@]{\mathit{wf}({#1}, {#2})}

\newif\ifcp
\cpfalse
\newcommand{\gname}[1][i]{\ifcp{\colorNode{\scriptstyle\textsf{#1}}}\else{}\fi}

\newif\ifguard
\guardfalse
\newcommand{\aguard}{\ifguard{\colorGuard \phi}\else{}\fi}

\def\colorGuard{\color{cyan}}
\def\colorPtp{\color{blue}}
\def\colorFun{\color{Navy}}

\def\colorOp{\color{OliveGreen}}
\def\colorNode{\color{LightCoral}}
\def\colorR{\color{OliveGreen}}
\def\colorE{\color{orange}}

\def\colorMsg{\color{BrickRed}}

\newcommand{\fillcolor}{orange!5}

\newcommand{\msg}[1][m]{\mathsf{\colorMsg{#1}}}

\newcommand{\sndint}[1][{\gint[]}]{\mathrm{snd}(#1)}
\newcommand{\rcvint}[1][{\gint[]}]{\mathrm{rcv}(#1)}
\newcommandx{\ggcommon}[3][1=\ptp,2={\aH},3={\aH'},usedefault=@]{f_{#1}}
\newcommandx{\opair}[2][1={\ae},2={\ae'},usedefault=@]{\conf{{#1},{#2}}}
\newcommandx{\hopair}[2][1={\aE},2={\aE'},usedefault=@]{\llparenthesis\, {#1},{#2}\, \rrparenthesis}
\newcommandx{\wf}[2][1={\aG},2={\aG'},usedefault=@]{wf({#1}, {#2})}
\newcommandx{\wb}[2][1={\aG},2={\aG'},usedefault=@]{wb({#1}, {#2})}
\newcommandx{\ws}[2][1={\aG},2={\aG'},usedefault=@]{ws({#1}, {#2})}

\tikzset{
  component/.style={
    draw,
    fill = white,
    minimum width = 1.5cm,
    minimum height = .5cm,
    drop shadow
  }
}
\tikzset{
  file/.style={
    thin,
    fill = blue!5,
    font = \tt\scriptsize,
    text width = .8cm,
    minimum width = 1.0cm,
    minimum height = .5cm,
    drop shadow
  }
}

\tikzset{
  dataflow/.style={
    thick,
    draw, ->, >=latex,
    dashed,
    OliveGreen
  }
}

\tikzset{
  pipeline/.style={
    thick,
    draw, ->, >=latex,
    double,
    red
  }
}

\tikzset{
  pomsetcloud/.style={
    cloud,
	 cloud puffs=20,
	 cloud ignores aspect,
	 minimum height=.1cm,
	 minimum width=2cm,
	 fill=blue!10,
	 opacity=.5,
	 draw
  }
}

\newcommand{\apom}{r}

\newcommand{\aR}[1][R]{{\colorR{#1}}}

\newcommandx{\detM}[1][1=\aCM,usedefault=@]{\Delta({#1})}
\tikzset{
  cnode/.style={
    shape=circle,
    minimum size = 0mm,
    inner sep = 1pt,
    font=\tiny,
    draw
  },
  carrow/.style={
    ->,
    shorten >=1pt,
    >=stealth',
    auto,
    font=\scriptsize,
    draw,
    sloped
  }
}

\newcommandx{\choranimation}[4][1=.8,2=.4,3=notempty,4=.35,usedefault=@]{
  \begin{overlayarea}{#1\linewidth}{#2\textheight}
    \begin{tikzpicture}[
      node distance=1cm and 2cm,
      scale=#4,
      every node/.style={transform shape},
      font=\large
      ]
      \node [choreo, align=center] (global){Choreography $\aG$ \\ global viewpoint};
      \node [below= of global] (fake) {};
      \node [choreo, align=center, local, below =of fake] (typei) {$\aCM_1$ \\ Local viewpoint$_1$};
      \node [choreo, align=center, local, left =of typei] (type1) {$\aCM_i$ \\ Local viewpoint$_1$};
      \node [choreo, align=center, local, right=of typei] (typen) {$\aCM_n$ \\ Local viewpoint$_n$};
      \node<+-> (synctxt) at (9,0)  {\textcolor{Navy}{\bf Synchrony}};
      \node<.-> (asynctxt) at (9,-4) {\textcolor{Navy}{\bf Asynchrony}};
      \node<.-> [below=of type1] (fake1) {};
      \node<.-> [below=of typei] (fakei) {};
      \node<.-> [below=of typen] (faken) {};
      \path<+-> [bigar] (global) edge[sloped,above] node {\color{OliveGreen}Project} (type1);
      \path<.-> [bigar] (global) edge[sloped,above] node {\color{OliveGreen}Project} (typei);
      \path<.-> [bigar] (global) edge[sloped,above] node {\color{OliveGreen}Project} (typen);
      \path<.-> [elli] (type1) -- (typei);
      \path<.-> [elli] (typei) -- (typen);
      \node<+-> [process, below=of fake1] (proc1) {Component$_1$};
      \node<.-> [process, below=of fakei] (proci) {Component$_i$};
      \node<.-> [process, below=of faken] (procn) {Component$_n$};
      \path<.-> [bigar,->,dashed,gray] (type1) edge[sloped,above] node {Validate} (proc1);
      \path<.-> [bigar,->,dashed,gray] (typei) edge[sloped,above] node {Validate} (proci);
      \path<.-> [bigar,->,dashed,gray] (typen) edge[sloped,above] node {Validate} (procn);
      \path<.-> [elli] (proc1) -- (proci);
      \path<.-> [elli] (proci) -- (procn);
		\ifempty{#3}{}{
      \node<+-> [process, right=of procn,xshift=4cm] (evolve1) {Component'$_1$};
      \node<.-> [process, right=of evolve1] (evolvei) {Component'$_i$};
      \node<.-> [process, right=of evolvei] (evolven) {Component'$_n$};
      \path<.-> [bigar,blue,dotted] (procn) edge node [above] {evolve/replace/compose} (evolve1);      
      \path<.-> [elli] (evolve1) -- (evolvei);
      \path<.-> [elli] (evolvei) -- (evolven);
      \node<+-> [choreo, align=center, local, above=of evolve1, yshift=1.5cm] (t11) {New $\aCM'_1$ \\ Local viewpoint$_1$};
      \node<.-> [choreo, align=center, local, above=of evolvei, yshift=1.5cm] (t1i) {New $\aCM'_i$ \\ Local viewpoint$_i$};
      \node<.-> [choreo, align=center, local, above=of evolven, yshift=1.5cm] (t1n) {New $\aCM'_n$ \\ Local viewpoint$_n$};
      \path<.-> [elli] (t11) -- (t1i);
      \path<.-> [elli] (t1i) -- (t1n);
      \path<.-> [bigar,->,dashed,gray] (evolve1) edge[sloped,above] node {Extract} (t11);
      \path<.-> [bigar,->,dashed,gray] (evolvei) edge[sloped,above] node {Extract} (t1i);
      \path<.-> [bigar,->,dashed,gray] (evolven) edge[sloped,above] node {Extract} (t1n);
      \node<+> [above=of t1i, yshift=1.5cm] (qm) {\Huge \textcolor{red}{ ??? }};
      \node<.-> [above=of t1i] (dummy) {};
      \node<+-> [choreo,align=center,above=of dummy,scale=.85] (global') {New choreography $\aG'$ \\ global viewpoint};
      \path<.-> [bigar,-] (t11) -- (dummy);
      \path<.-> [bigar,->] (t1i) edge node[right,xshift=1em] {\color{OliveGreen}Synthesise} (global');
      \path<.-> [bigar,-] (t1n) -- (dummy);
		}
    \end{tikzpicture}
  \end{overlayarea}
}%%%end choranimation

\newcommandx{\cm}[2][1=\ptp, 2=\aM]{{#2}_{#1}}
\newcommandx{\achan}[2][1=A,2=B,usedefault=@]{{\ptp[#1]}{\,}{\ptp[#2]}}

\newcommand{\oact}{\outop[]}
\newcommand{\iact}{\inop[]}

\newcommandx{\acfsmout}[3][1=A,2=B,3=m,usedefault=@]{\achan[{#1}][{#2}] \oact {\msg[{#3}]}}
\newcommandx{\acfsmin}[3][1=A,2=B,3=m,usedefault=@]{\achan[{#1}][{#2}] \iact {\msg[{#3}]}}
\newcommandx{\fsaout}[2][1={\p},2={},usedefault=@]{
  \ptp[#1] \ \outop[]\ \msg[{#2}]
}
\newcommandx{\fsain}[2][1={\p},2={},usedefault=@]{
  \ptp[#1] \ \inop[]\ \msg[{#2}]
}

\makeatletter
\newcommand{\linenumfontsize}{\@setfontsize{\linenumfontsize}{3pt}{3pt}}
\makeatother
\lstdefinelanguage{sys}{
	commentstyle=\color{Gray},
	morecomment=[s]{[}{]},
	keywords=[0]{system,of,do,end},	keywordstyle=\color{orange}\bfseries,
}

\lstdefinelanguage{sgg}{
  commentstyle=\color{Gray},
  morecomment=[l]{..},
  morecomment=[s]{[}{]},
  keywords=[0]{repeat,branch,sel},
  keywordstyle=\color{orange}\bfseries,
  morekeywords=[1]{*,\+,|,->},
  literate={->}{$\colorOp \xrightarrow$}1 {|}{$\gparop$}1 {;}{$\gseqop$}1 {+}{$\gchoop$}1 {\{}{{\textcolor{Navy}{\{}}}1 {\}}{{\textcolor{Navy}{\}}}}1
}

\newcommand{\aG}{\mathsf{G}}

\newcommand{\gseqop}{{\colorOp ;}\,}
\newcommand{\gparop}{{\colorOp \ |\ }}
\newcommand{\gchoop}{{\colorOp \ +\ }}
\newcommand{\grecop}{{\colorOp *}}
\newcommand{\grecopp}{{\colorOp{@}}}
\newcommandx{\nmerge}[2][1={i},2={},usedefault=@]{
  \ifempty{#2}{
    \ifempty{#1}{\mu}{\gname[-{#1}]}
  }{-{#2}}
}

\mkfun{\esbj}{sbj}{\ae}

\makeatletter%
\@ifclassloaded{exam-paper}%
  {}%
  {\makeatletter%
    \@ifclassloaded{test}%
    {}%
    \makeatother%
  }
\makeatother%

\newcommandx{\gnode}[2][1=i,2=\gint,usedefault=@]{
  \ifcp{
    \ifempty{#1}{#2}{\gname[#1].\big({#2}\big)}
  }
  \else
  {#2}
  \fi
}

\newcommandx{\refgint}[3][1=A,2=\msg,3=B,usedefault=@]{
  \ptp[#1] {\colorOp \xdashrightarrow[{}]{\msg[{#2}]}}{
    \renewcommand{\do}[1]{\ptp[##1]}
	 \docsvlist{#3}
  }
}
\newcommandx{\gout}[4][1=\gname,2=\ptp,3=m,4={\ptp[C]},usedefault=@]{
  \achan[{#2}][{#4}] {\colorOp {\colorOp{!}}} {\msg[{#3}]}
}
\newcommandx{\gin}[4][1=\gname,2=\ptp,3=m,4={\ptp[C]},usedefault=@]{
  \achan[{#2}][{#4}] {\colorOp {\colorOp{?}}} {\msg[{#3}]}
}
\newcommandx{\gseq}[3][1=\gname,2={\aG},3={\aG'},usedefault=@]{
  \def\ggraph{{#2} \gseqop {#3}}
  \ggraph
}
\newcommand{\ginfix}[4]{
  \def\ggraph{{#2} {#4} {#3}}
  \gnode[#1][\ggraph]
}
\newcommandx{\gpar}[3][1=i,2={\aG},3={\aG'},usedefault=@]{
  \ginfix{#1}{#2}{#3}{\gparop}
}
\newcommandx{\gcho}[3][1=i,2={\aG},3={\aG'},usedefault=@]{
  \ginfix{#1}{#2}{#3}{\gchoop}
}
\newcommandx{\gchov}[3][1=\gname,2={\aG},3={\aG'},usedefault=@]{
  \def\ggraph{\left(
  \begin{array}l
    \ifempty{#1}{{#2} \\ \gchoop \\ {#3}}{\!\!{#2} \\ \gchoop \\ {#3}}
  \end{array}\right)}
  \ifcp\gnode[{$#1$}][\ggraph] \else \ggraph \fi
}
\newcommandx{\grec}[3][1=i,2={\aG},3={\p},usedefault=@]{
  \def\ggraph{\grecop {#2} \ifempty{#3}{}{\grecopp {#3}}}
  \ifempty{#1}{\ggraph}{\gname[{$#1$}][\ggraph]}
}	

\newcommand{\getcentroid}[2]{
    \coordinate (tmpgatecoord) at (0,0);
    \foreach \n [count=\i] in {#1}{
      \path (\n);
      \coordinate (tmpgatecoord) at ($(tmpgatecoord) + (\n)$);
      \coordinate (#2) at ($1/\i*(tmpgatecoord)$);
    }
}

\tikzset{
  hgsem/.style={
    draw,
    node distance=2cm and 1cm,
    transform shape,
    smooth,
    every node/.style = {font=\sffamily\bfseries}
  }
}

\tikzset{
  hgstyle/.style={
    src color={#1},
    tgt color={#1},
    centroid color={#1},
    centroid label={#1},
    centroid name={#1},
    centroid radius={#1},
    centroid ratio={#1},
    xoffset={#1},
    yoffset={#1},
    xsrcoffset={#1},
    ysrcoffset={#1},
    xtgtoffset={#1},
    ytgtoffset={#1},
    font={#1},
    centroid angle={#1},
    centroid tolerance={#1}
  },
  src color/.store in = \hgsrccol,
  tgt color/.store in = \hgtgtcol,
  centroid color/.store in =\hgfillcolor,
  centroid label/.store in =\hglabel,
  centroid name/.store in =\hgname,
  centroid radius/.store in = \hgradius,
  centroid ratio/.store in = \hgratio,
  xoffset/.store in =\hgxoffset,
  yoffset/.store in =\hgyoffset,
  xsrcoffset/.store in =\hgxsrcoffset,
  ysrcoffset/.store in =\hgysrcoffset,
  xtgtoffset/.store in =\hgxtgtoffset,
  ytgtoffset/.store in =\hgytgtoffset,
  centroid angle/.store in =\hgangle,
  centroid tolerance/.store in =\hgtolerance,
  src color = black,
  tgt color = black,
  centroid color = orange!40,
  centroid label={},
  centroid name={dummycentroid},
  centroid radius = .7pt,
  centroid ratio = .35,
  xoffset = 0,
  yoffset = 0,
  xsrcoffset = 0,
  ysrcoffset = 0,
  xtgtoffset = 0,
  ytgtoffset = 0,
  font=\sffamily\scriptsize,
  centroid angle=0,
  centroid tolerance=10pt
}

\newcommandx{\mkhg}[5][1={},4={},5={},usedefault=@]{
  \begingroup
  \tikzset{#1}
  \StrCount{#2,}{,}[\l] % from package xstring
  \StrCount{#3,}{,}[\m] % from package xstring
  \ifthenelse{\l = 1 \AND \m = 1}{
    \ifempty{#4}{
      \ifempty{#5}{
        \path[hgsem, ->, >=stealth', shorten >=1pt] (#2) -- (#3);
      }{
        \path[hgsem, ->, >=stealth', shorten >=1pt] (#2) #5 (#3);
      }
    }{
      \ifempty{#5}{
        \path[hgsem, ->, >=stealth', shorten >=1pt, #4] (#2) -- (#3);
      }{
        \path[hgsem, ->, >=stealth', shorten >=1pt, #4] (#2) #5 (#3);
      }
    }
  }{
    \coordinate (srcoffset) at (\hgxsrcoffset,\hgysrcoffset);
    \coordinate (tgtoffset) at (\hgxtgtoffset,\hgytgtoffset);
    \getcentroid{#2}{srccentroid};
    \getcentroid{#3}{tgtcentroid};
    \node[label={left:\hglabel}] (\hgname) at ($(srccentroid)!{1-\hgratio}!\hgangle:(tgtcentroid) + (\hgxoffset,\hgyoffset)$) {};
    \pgfgetlastxy \xc \yc;
    \pgfmathtruncatemacro{\xcontrol}{\xc};
    \pgfmathtruncatemacro{\ycontrol}{\yc};
    \foreach \n in {#2}{
      \path (\n);
      \pgfgetlastxy \xntmp \yntmp;
      \pgfmathtruncatemacro{\xn}{\xntmp};
      \pgfmathtruncatemacro{\yn}{\yntmp};
      \pgfmathsetmacro\xtmpdiff{abs(\xn - \xcontrol + \hgxsrcoffset)};
      \pgfmathsetmacro\ytmpdiff{abs(\yn - \ycontrol + \hgytgtoffset)};
      \ifdim \xtmpdiff pt > \hgtolerance
      \ifempty{#4}{
        \path[hgsem, \hgsrccol] (\n) .. controls ($(srccentroid.center) + (srcoffset)$) .. (\hgname.center);
      }{
        \path[hgsem, \hgsrccol] (\n) .. controls ($(srccentroid.center) + (srcoffset)$) .. (\hgname.center);
      }
      \else
      \ifempty{#4}{
        \path[hgsem, \hgsrccol] (\n) -- (\hgname.center);
      }{
        \path[hgsem, \hgsrccol, #4] (\n) -- (\hgname.center);
      }
      \fi
    }
    \foreach \n in {#3}{
      \path (\n);
      \pgfgetlastxy \xntmp \yntmp;
      \pgfmathtruncatemacro{\xn}{\xntmp};
      \pgfmathtruncatemacro{\yn}{\yntmp};
      \pgfmathsetmacro\xtmpdiff{abs(\xn - \xcontrol)};
      \pgfmathsetmacro\ytmpdiff{abs(\yn - \ycontrol)};
      \ifdim \xtmpdiff pt > \hgtolerance
      \ifempty{#4}{
        \path[hgsem, ->, >=stealth', shorten >=1pt, \hgtgtcol] (\hgname.center) .. controls (tgtcentroid.center) and ($(tgtcentroid.center) + (tgtoffset)$) .. (\n);
      }{
        \path[hgsem, ->, >=stealth', shorten >=1pt, \hgtgtcol,#4] (\hgname.center) .. controls (tgtcentroid.center) and ($(tgtcentroid.center) + (tgtoffset)$) .. (\n);
      }
      \else
      \ifempty{#4}{
        \path[hgsem, ->, >=stealth', shorten >=1pt, \hgtgtcol] (\hgname.center) --  (\n);
      }{
        \path[hgsem, ->, >=stealth', shorten >=1pt, \hgtgtcol] (\hgname.center) --  (\n);
      }
      \fi
    }
    \fill[\hgfillcolor] (\hgname) circle [radius=\hgradius];
  }
  \endgroup
}

\newcommandx{\hgordeq}[1][1={\aH},usedefault=@]{\sqsubseteq_{#1}}
\newcommandx{\gintsem}[4][4=.5]{
  \tikz[hgsem,scale=#4,every node/.style={font=\scriptsize}]{
    \node (out) {$\aout[{#1}][{#2}][][{#3}]$};
    \node[below = 20pt of out] (in) {$\ain[{#1}][{#2}][][{#3}]$};
    \mkhg{out}{in};
  }
}

\newcommandx{\gsem}[2][1={\aG},2={},usedefault=@]{\left\llbracket {#1} \right\rrbracket_{#2}}

\newcommandx{\rbot}{\text{undef}}
\newcommandx{\rtrs}[1][1={\aH},usedefault=@]{{#1}^{\star}}
\newcommandx{\gord}[1][1={\aG},usedefault=@]{\leq_{#1}}
\newcommandx{\gordeq}[1][1={\aG},usedefault=@]{\leq_{#1}}
\mkfun{\cause}{cs}{}
\mkfun{\effect}{ef}{}
\newcommandx{\aW}{w}

\newcommand{\gfun}[1]{\ensuremath{\mathsf{\colorFun #1}}}
\mkfun{\eact}{\gfun{act}}{}
\mkfun{\enode}{\gfun{cp}}{}

\newcommandx{\rseq}[2][1=\aG,2={\aG'},usedefault=@]{\gfun{seq}({#1},{#2})}
\newcommandx{\rpar}[2][1=\aG,2={\aG'},usedefault=@]{\gfun{par}({#1},{#2})}

\newcommandx{\gproj}[2][1=\aG,2=\ptp]{{#1}\downarrow_{#2}}
\newcommandx{\cinit}[1][1={\aQzero},usedefault=@]{{#1}}
\newcommandx{\cfinal}[1][1={q_e},usedefault=@]{{#1}}

\newcommandx{\geproj}[4][1=\aG,2=\ptp,3=\cinit,4=\cfinal,usedefault=@]{
  {#1}\downarrow_{#2}^{{#3},{#4}}
}

\newcommand*{\StrikeThruDistance}{0.15cm}%
\tikzset{strike thru arrow/.style={
    decoration={markings, mark=at position 0.5 with {
        \draw [blue, thick,-] 
            ++ (-\StrikeThruDistance,-\StrikeThruDistance) 
            -- ( \StrikeThruDistance, \StrikeThruDistance);}
    },
    postaction={decorate},
}}

\newcommandx{\ich}[1][1={\aG},usedefault=@]{{#1}^{\oplus}}
\newcommandx{\ichedges}[2][1={\aG},2={\gname},usedefault=@]{{#1}^{\oplus}({#2})}
\newcommandx{\parts}[1]{2^{#1}}
\newcommandx{\actch}{c}
\newcommandx{\soundactch}[2][1={\aG},2={\actch},usedefault=@]{{#1} \,\circledR\, {#2}}
\newcommandx{\rOnActch}[2][1={\aG},2={\actch},usedefault=@]{{#1} \setminus {#2}}
\newcommandx{\rOnActchClean}[2][1={\aG},2={\actch},usedefault=@]{{#1} \circledR {#2}}
\newcommandx{\rAllEvents}[1][1={\aG},usedefault=@]{\mathit{dom}(#1)}

\newcommand{\AV}{\mathcal{V}}
\newcommand{\aH}{H}

\newcommandx{\hgvertex}[2][1=\al,2=\gname,usedefault=@]{{#1}_{\textcolor{red}{[{#2}]}}}
\newcommand{\aE}{{\colorE E}}
\renewcommand{\ae}[1][e]{{\colorE{#1}}}

\newcommand{\al}[1][l]{{\colorE{#1}}}
\newcommandx{\hyedge}[1]{\{#1\}}

\newcommandx{\rdiv}[2][1=\gcho,2=\ptp,usedefault=@]{
  \gfun{div}_{#2}(#1)
}

\newcommandx{\rrdiv}[5][1={\aG},2={\aG'},3={\AV},4={,\AV'},5=\ptp,usedefault=@]{
  \gfun{div}^{#3#4}_{#5}(#1,#2)
}
\newcommandx{\pdiv}[3][1={\apom_1},2={\apom_2},3={\apom},usedefault=@]{
  \gfun{div}_{#3}(#1,#2)
}
\newcommandx{\pfork}[3][1={\apom_1},2={\apom_2},3={\apom},usedefault=@]{
  \gfun{fork}_{#3}(#1,#2)
}

\tikzset{
  pomset/.style={
    node distance = .6cm and .6cm,
    scale = .7,
    transform shape,
    smooth
  }
}

\newcommandx{\mkint}[6][3=a,4=\msg,5=b,usedefault=@]{
  \node[bblock, #1] (#2) {$\gint[#3][#4][#5]$};
}

\newcommand{\mkseq}[2]{\path[line] (#1) -- (#2);}

\newcommand{\mknseq}[1]{
  \StrCount{#1}{,}[\l] % from package xxstring
  \StrBefore{#1}{,}[\myhead]
  \StrBehind{#1}{,}[\mytail]
  \StrBefore{\mytail}{,}[\sndel]
  \ifnum \l > 1 {
    \mkseq{\myhead}{\sndel};
    \mknseq{\mytail}
  }
  \else{\ifnum \l > 0{
      \mkseq{\myhead}{\mytail};
    }
    \else{}
    \fi
  }
  \fi
}

\newcommandx{\mkgateblock}[6][6=yellow!10,usedefault=@]{
  \path(#2);
  \pgfgetlastxy{\xgate}{\ygate};
  \pgfmathtruncatemacro{\xgateround}{\xgate};
  \StrCount{#3,}{,}[\l] % from package xstring
  \ifnum \l < 2 {\errmessage{#3 argument should be a comma-separated list of lenght >= 2}}
  \else{
    \foreach \n in {#3}{
      \path (\n);
      \pgfgetlastxy{\xnode}{\ynode};
      \pgfmathtruncatemacro{\xnround}{\xnode};
      \pgfmathsetmacro\tmpdiff{abs(\xnround - \xgateround)}
      \ifdim \tmpdiff pt > 1 pt \path[line] (#2) -| (\n);
      \else
        \path[line] (#2) -- (\n);
      \fi
    }
  }
  \fi
  \StrCount{#4,}{,}[\l] % from package xstring
  \ifnum \l < 2 {\errmessage{#4 argument should be a comma-separated list of lenght >= 2}}
  \else{
    \foreach \n in {#4}{
      \path (\n);
      \pgfgetlastxy{\xnode}{\ynode};
      \pgfmathtruncatemacro{\xnround}{\xnode};
      \pgfmathsetmacro\tmpdiff{abs(\xnround - \xgateround)}
      \ifdim \tmpdiff pt > 1 pt \path[line] (\n) |- (#5);
      \else
        \path[line] (\n) -- (#5);
      \fi
    }
  }
  \fi
  \node[#1] at (#2) {};
  \node[#1] at (#5) {};
  {
    \begin{pgfonlayer}{background}
      \path[fill=#6,rounded corners]
      (current bounding box.south west) rectangle
      (current bounding box.north east);
    \end{pgfonlayer}
  }
}

\newcommandx{\mkbranchblock}[5][5=@]{
  \mkgateblock{ogate}{#1}{#2}{#3}{#4}[#5]
}

\newcommandx{\mkforkblock}[5][5=@]{
  \mkgateblock{agate}{#1}{#2}{#3}{#4}[#5]
}

\newcommandx{\mkgraph}[3][1=.5cm, usedefault=@]{
  \node[source,above = #1 of {#2}] (src#2) {};
  \node[sink,below  = #1 of {#3}] (sink#3) {};
  \path[line] (src#2) -- (#2);
  \path[line] (#3) -- (sink#3);
}

\newcommandx{\mkloop}[5][1=.5, 2=1.5, 5=\aguard, usedefault=@]{
  \node[lgate,above = #1 of {#3}] (entry#3) {};
  \pgfgetlastxy \xentry \yentry;
  \pgfmathtruncatemacro{\xentryrounded}{\xentry};
  \node[below = #1 of {#4}, label = {above right:{$#5$}},yshift=-1em] (dummy) {};
  \node[lgate,below  = #1 of {#4}] (exit#4) {};
  \pgfgetlastxy \xexit \yexit;
  \pgfmathtruncatemacro{\xexitrounded}{\xexit};
  \path[line] (entry#3) -- (#3);
  \path[line] (#4) -- (exit#4);
  \pgfmathsetmacro\tmpdiff{abs(\xentryrounded - \xexitrounded)}
  \path[line, color=teal] (exit#4) -| ($(exit#4)+(\tmpdiff,0)+(#2,0)$) |- (entry#3);
}

\newcommandx{\mkfork}[4][2=gatenode,3=i,4=.6,usedefault=@]{
  \mkgatebegin{#1}[{\gname[{#3}]}][agate][#4]{#2}
}

\newcommandx{\mkbranch}[4][2=gatenode,3=i,4=.6,usedefault=@]{
  \mkgatebegin{#1}[{\gname[{#3}]}][ogate][#4]{#2}
}

\newcommandx{\mkgatebegin}[5][2={},3=ogate,4=.5,usedefault=@]{
  \coordinate (gatecord) at (0,0);
  \coordinate (xmax) at (0,0);
  \coordinate (xmin) at (0,0);
  \pgfgetlastxy \xmin \xmax;
  \foreach \n [count=\i] in {#1}{
    \pgfgetlastxy \xc \yc;
    \path (\n);
    \pgfgetlastxy \xn \yn;
    \ifnum \i = 1
      \coordinate (xmin) at (\xn,0);
      \coordinate (xmax) at (\xn,0);
      \coordinate (max) at (0,\yn);
    \else
      \ifdim \xn < \xmin
        \coordinate (xmin) at (\xn,0);
      \fi
      \ifdim \xn > \xmax
        \coordinate (xmax) at (\xn,0);
      \fi
      \ifdim \yn < \yc
        \coordinate (max) at (0,\yc);
      \else
        \coordinate (max) at (0,\yn);
      \fi
    \fi
  }
  \coordinate (gatecord) at ($(xmin)!.5!(xmax) + (max) + (0,#4) + (max)$);
  \node[#3,label={below:$#2$}] (#5) at (gatecord) {};
  \pgfgetlastxy{\xgate}{\ygate};
  \pgfmathtruncatemacro{\xgateround}{\xgate};
  \StrCount{#1,}{,}[\l] % from package xxstring
  \ifnum \l < 2 {\errmessage{#1 argument should be a comma-separated list of lenght >= 2}}
  \else{
    \foreach \n in {#1}{
      \path (\n);
      \pgfgetlastxy{\xnode}{\ynode};
      \pgfmathtruncatemacro{\xnround}{\xnode};
      \pgfmathsetmacro\tmpdiff{abs(\xnround - \xgateround)}
      \ifdim \tmpdiff pt > 1 pt \path[line] (#5) -| (\n);
      \else
        \path[line] (#5) -- (\n);
      \fi
    }
  }
  \fi
}

\newcommandx{\mkgatebeginold}[5][2={},3=ogate,4=.5,usedefault=@]{
  \coordinate (gatecord) at (0,0);
  \foreach \n [count=\i] in {#1}{
    \pgfgetlastxy \xc \yc;
    \path (\n);
    \pgfgetlastxy \xn \yn;
    \coordinate (gatecord) at ($(gatecord) + (\xn,0)$);
    \coordinate (gatecord) at ($1/\i*(gatecord)$);
    \ifdim \yn < \yc
    \node (max) at (0,\yc) {};
    \else
    \node (max) at (0,\yn) {};
    \fi
  }
  \coordinate (gatecord) at ($(gatecord) + (0,#4) + (max)$);
  \node[#3,label={below:$#2$}] (#5) at (gatecord) {};
  \pgfgetlastxy{\xgate}{\ygate};
  \pgfmathtruncatemacro{\xgateround}{\xgate};
  \StrCount{#1,}{,}[\l] % from package xxstring
  \ifnum \l < 2 {\errmessage{#1 argument should be a comma-separated list of lenght >= 2}}
  \else{
    \foreach \n in {#1}{
      \path (\n);
      \pgfgetlastxy{\xnode}{\ynode};
      \pgfmathtruncatemacro{\xnround}{\xnode};
      \pgfmathsetmacro\tmpdiff{abs(\xnround - \xgateround)}
      \ifdim \tmpdiff pt > 1 pt \path[line] (#5) -| (\n);
      \else
        \path[line] (#5) -- (\n);
      \fi
    }
  }
  \fi
}

\newcommandx{\mkmerge}[4][2=gatenode,3=i,4=.5,usedefault=@]{
  \mkgateend{#1}[{\ifempty{#3}{}{\nmerge[#3]}}][ogate][#4]{#2}
}

\newcommandx{\mkjoin}[4][2=gatenode,3=i,4=.5,usedefault=@]{
  \mkgateend{#1}[{\ifempty{#3}{}{\nmerge[#3]}}][agate][#4]{#2}
}

\newcommandx{\mkgateend}[5][2={},3=ogate,4=.5,usedefault=@]{
  \coordinate (gatecord) at (0,0);
  \coordinate (xmax) at (0,0);
  \coordinate (xmin) at (0,0);
  \pgfgetlastxy \xmin \xmax;
  \foreach \n [count=\i] in {#1}{
    \pgfgetlastxy \xc \yc;
    \path (\n);
    \pgfgetlastxy \xn \yn;
    \ifnum \i = 1
      \coordinate (xmin) at (\xn,0);
      \coordinate (xmax) at (\xn,0);
      \coordinate (ymin) at (0,\yn);
    \else
      \ifdim \xn < \xmin
        \coordinate (xmin) at (\xn,0);
      \fi
      \ifdim \xn > \xmax
        \coordinate (xmax) at (\xn,0);
      \fi
      \ifdim \yn > \yc
        \coordinate (ymin) at (0,\yc);
      \else
        \coordinate (ymin) at (0,\yn);
      \fi
    \fi
  }
  \coordinate (gatecoord) at ($(xmin)!.5!(xmax) + (ymin)$);
  \node[#3,label={above:$#2$}] (#5) at ($(gatecoord) - (0,{#4})$) {};
  \pgfgetlastxy{\xgate}{\ygate};
  \pgfmathtruncatemacro{\xgateround}{\xgate};
  \StrCount{#1,}{,}[\l] % from package xstring
  \ifnum \l < 2 {\errmessage{#1 argument should be a comma-separated list of lenght >= 2}}
  \else{
    \foreach \n in {#1}{
      \path (\n);
      \pgfgetlastxy{\xnode}{\ynode};
      \pgfmathtruncatemacro{\xnround}{\xnode};
      \pgfmathsetmacro\tmpdiff{abs(\xnround - \xgateround)}
      \ifdim \tmpdiff pt > 1 pt \path[line] (\n) |- (#5);
      \else
        \path[line] (\n) -- (#5);
      \fi
    }
  }
  \fi
}

\newcommandx{\mkgateendold}[5][2={},3=ogate,4=.5,usedefault=@]{
  \coordinate (gatecord) at (0,0);
  \coordinate (xmax) at (0,0);
  \coordinate (xmin) at (0,0);
  \pgfgetlastxy \xmin \xmax;
  \foreach \n [count=\i] in {#1}{
    \pgfgetlastxy \xc \yc;
    \path (\n);
    \pgfgetlastxy \xn \yn;
    \ifdim \xn < \xmin
    \coordinate (xmin) at (\xn,0);
    \fi
    \ifdim \xn > \xmax
    \coordinate (xmax) at (\xn,0);
    \fi
    \ifdim \yn > \yc
    \coordinate (ymin) at (0,\yc);
    \else
    \coordinate (ymin) at (0,\yn);
    \fi
    \coordinate (gatecord) at ($(xmin)!.5!(xmax) + (ymin)$);
  }
  \node[#3,label={above:$#2$}] (#5) at ($(gatecord) - (0,{#4})$) {};
  \pgfgetlastxy{\xgate}{\ygate};
  \pgfmathtruncatemacro{\xgateround}{\xgate};
  \StrCount{#1,}{,}[\l] % from package xstring
  \ifnum \l < 2 {\errmessage{#1 argument should be a comma-separated list of lenght >= 2}}
  \else{
    \foreach \n in {#1}{
      \path (\n);
      \pgfgetlastxy{\xnode}{\ynode};
      \pgfmathtruncatemacro{\xnround}{\xnode};
      \pgfmathsetmacro\tmpdiff{abs(\xnround - \xgateround)}
      \ifdim \tmpdiff pt > 1 pt \path[line] (\n) |- (#5);
      \else
        \path[line] (\n) - (#5);
      \fi
    }
  }
  \fi
}

\newcommand{\gatedistancein}{3pt}
\newcommand{\gatedistanceinand}{2pt}

\usetikzlibrary{
  arrows,
  backgrounds,
  chains,
  calc,
  decorations.markings,decorations.pathreplacing,
  fadings,
  fit,
  patterns,
  petri,
  positioning,
  shadows,
  shapes,automata,shapes.callouts
}

\tikzset{
  src/.style={draw,circle,fill=white,
    minimum size=2mm,
    inner sep=0pt
  },
  sink/.style={draw,circle,double,fill=white,
    minimum size=1.5mm,
    inner sep=0pt
  },
  node/.style={draw,circle,fill=black,
    minimum size=2mm,
    inner sep=0pt
  },
  source/.style={draw,circle,fill=white,
    minimum size=3mm,
    inner sep=0pt
  },
  sink/.style={draw,circle,double,fill=white,
    minimum size=3mm,
    inner sep=0pt
  },
  block/.style = {rectangle, draw=gray, align=center, fill=orange!25, rounded corners=0.1cm,
    minimum size=5mm, inner sep=2pt},
  prenode/.style = {minimum size=9pt,inner sep=2pt, font=\Large},
  bblock/.style = {rectangle, draw=blue!50, opacity=.7, line width=.5pt, align=center, fill=white, rounded corners=0.1cm,
    minimum size=4mm, inner sep=1pt},
  prenode/.style = {minimum size=9pt,inner sep=2pt, font=\Large},
  agate/.style={draw, rectangle,
    minimum size=3mm,
    inner sep=0pt,
    fill=orange!25,
    label={[red]center:$\mid$}
  },
  ogate/.style = {
    diamond, draw, fill=orange!25,
    minimum size=4mm,
    inner sep=0pt,
    label={[Navy]center:{\bf\textsf{X}}}
  },
  lgate/.style = {
    diamond, draw, fill=orange!25,
    minimum size=4mm,
    inner sep=0pt,
    label={[red]center:$\circlearrowleft$}
    },
  altogate/.style = {
    diamond, draw,
    minimum size=4mm,
    inner sep=0pt,
    postaction={path picture={% 
        \draw
        ([yshift=\gatedistancein]path picture bounding box.south) -- ([yshift=-\gatedistancein]path picture bounding box.north)
        ([xshift=-\gatedistancein]path picture bounding box.east) -- ([xshift=\gatedistancein]path picture bounding box.west)
        ;}}},
  altgate/.style={draw, rectangle,
    minimum size=3mm,
    inner sep=0pt,
    postaction={path picture={% 
        \draw
        ([yshift=\gatedistanceinand]path picture bounding box.south) --
        ([yshift=-\gatedistanceinand]path picture bounding box.north) ;}}},
  anygate/.style = {circle, draw, fill=white,
    minimum size=4mm,
    inner sep=0pt,
    postaction={path picture={% 
        \draw[black]
        ([xshift=-\gatedistancein,yshift=\gatedistancein]path picture bounding box.south east) --
        ([xshift=\gatedistancein,yshift=-\gatedistancein]path picture bounding box.north west)
        ([xshift=-\gatedistancein,yshift=-\gatedistancein]path picture bounding box.north east) --
        ([xshift=\gatedistancein,yshift=\gatedistancein]path picture bounding box.south west)
        ;}}
  },
  smallglobal/.style={
        node distance=1cm and 0.8cm, semithick, scale=0.8, every node/.style={transform shape}
  },
  elli/.style = {draw,densely dotted,-},
  line/.style = {draw,->, rounded corners=0.07cm,>=latex},
  nline/.style = {draw,semithick, ->},
  pline/.style = {draw,->,>=latex},
  node distance=1cm and 0.7cm,
  baseline=(current  bounding  box.center),
  local/.style={rectangle, draw, fill=\fillcolor, drop shadow,
    text centered, rounded corners, minimum height=5em
  },
  bigar/.style={
    draw,very thick, ->
  },
  process/.style={rectangle, draw=gray, fill=\fillcolor, drop shadow,
    text centered, minimum height=5em,text=gray
  },
  choreo/.style={rectangle, draw, fill=\fillcolor, drop shadow,
    text centered, rounded corners, minimum height=5em
  },
  mycfsm/.style={
        font=\footnotesize,
        initial where=above,
        ->,>=stealth,auto, node distance=1cm and 1cm,
        scale=1, every node/.style={transform shape},
        every state/.style=inner sep=2pt,
        baseline=(current  bounding  box.center),
        initial text={}
  },
  machinecloud/.style={
    cloud, cloud puffs=10, cloud ignores aspect, minimum height=.1cm, minimum width=2cm, draw
  },
  fitting node/.style={
    inner sep=0pt,
    fill=none,
    draw=none,
    reset transform,
    fit={(\pgf@pathminx,\pgf@pathminy) (\pgf@pathmaxx,\pgf@pathmaxy)}
  },
  mypetri/.style={
    font=\footnotesize,
    baseline=(current  bounding  box.center)
  },
  silentrans/.style = {rectangle, draw=black, align=center, fill=black,
    minimum height=1pt,
    minimum width=15pt,
    inner sep=1.5pt
  },
  reset transform/.code={\pgftransformreset},
  tmtape/.style={draw,minimum size=1.2cm}
}

\newcommand{\gunlessop}{\mbox{\colorOp\tiny\tt unless}}

\newcommandx{\gtry}[5][1=\gname,2={\aG_1 \gchoop \cdots \gchoop \aG_n},3=\gin,4=\gout,5={j},usedefault=@]{
  \def\foo{\gtryop\ {#2} \ \gcatchop\ {#3} {\colorOp \Rightarrow} {#4} {\colorOp \bullet} {\gname[{#5}]}}
  \gnode[{#1}][{\ifempty{#1} {\foo } {(\foo)}}]
}

\newcommandx{\gtrycatch}[4][1=\gname,2={\aG},3=\gin,4={\aG'},usedefault=@]{
  \def\foo{\gtryop\ {#2} \ \gcatchop\ {#3} \gdoop\ {#4}}
  \gnode[{#1}][{\ifempty{#1} {\foo} {(\foo)}}]
}

\newcommandx{\agG}[2][1={\aG},2=\aguard]{{#1} \ifempty{#2}{}{\ \gunlessop\ {#2}}}

\newcommandx{\grcho}[5][1=\gname,2={\agG},3={\agG[\aG'][\aguard']},4={\cdots},5=A,usedefault=@]{
  \def\foo{{#2} {\ \ifempty{#4}{\gchoop}{\gchoop \ifempty{#4}{}{\ {#4}\  \gchoop}}\ } {#3}}
  \ifempty{#1}{\ifempty{#5}{\foo}{\gselop\ \cpt[{#1}][{\ptp[#5]}]\big\{ \foo \big\}}}{\gselop\ \cpt[{#1}][{\ptp[#5]}]\big\{ \foo \big\}}
}

\newcommandx{\ggprefix}[3][1=\ptp,2={\aR},3={\aR'},usedefault=@]{f_{#1}} % it was \newcommandx{\common}{...}
\newcommand{\aconfigfn}{\chi}
\newcommand{\aconfig}{\ell}

\newcommand{\lstates}{\statemap}
\newcommandx{\sysconfig}[3][1=\lstates,2=\aconfigfn,3={},usedefault=@]{
  \conf{ {#1},{#2} \ifempty{#3}{}{, #3} }
}
\newcommand{\sysctxfn}[1][]{\gamma_{#1}}
\newcommandx{\sysctx}[2][1=\aQ,2={},usedefault=@]{({#1},\sysctxfn[{#2}])}

\newcommandx{\alog}[4][1=\msg,2=q,3=\gname,4=t,usedefault=@]{({#1},{#2},{#3},{#4})}

\newcommandx{\aQzero}[1][1=,usedefault=@]{
  {\ifempty{#1}{q_0}{q_{0#1}}}
}
\newcommand{\badbranches}[1][]{\beta\ifempty{#1}{}{({#1})}}
\newcommand{\aTrs}{\tset}
\newcommandx{\guardedaction}[2][1=\al,2=\aguard,usedefault=@]{
  {#1} \ifempty{#2}{}{/} {#2}
}
\newcommandx{\atrM}[4][1=q,2=\al,3={\hat q,\hat \al, \aguard},4=q',usedefault=@]{
  {#1} \xrightarrow[{#3}]{\guardedaction[{#2}][]} {{#4}}
}
\newcommandx{\atrS}[5][
  1={\sysconfig[@][@][\badbranches]},
  2=\al,
  3=\aguard,
  4={\sysconfig[\lstates'][\aconfigfn'][\badbranches]},
  5=\sysctx,usedefault=@
]{
  {#1} \xRightarrow{\qquad} {{#4}}
}
\newcommandx{\arevtrS}[2][
  1={\sysconfig[@][@][\badbranches]},
  2={\sysconfig[\lstates'][\aconfigfn'][\badbranches']},
  usedefault=@
]{
  {#1} \rightsquigarrow {#2}
}

\newcommandx{\enables}[2][1=\aconfigfn,2=\aguard,usedefault=@]{{#1} \vdash {#2}}
\newcommandx{\gprojfn}[5][1=\aG,2=\ptp,3=\cinit,4=\cfinal,5={},usedefault=@]{
  \mathbf{proj}_{#2}({#1},{#3},{#4}\ifempty{#5}{}{,{#5}})
}

\newcommandx{\rbp}[3][1=\aG,2=\aconfigfn,3=\achan,usedefault=@]{\mathtt{RBP}_{{#1},{#2}}\ifempty{#3}{}{({#3})}}

\newcommand{\apseudoCFSM}{\mathtt{M}}
\newcommandx{\pseudoseq}[2][1=\apseudoCFSM,2=\apseudoCFSM',usedefault=@]{{#1}  ; {#2}}
\newcommandx{\pseudoCFSM}[4][1=\aQ,2=\aQzero,3=\cfinal,4=\aTrs,usedefault=@]{(#1 \ ; #2 \ ; #3 \ ; #4)}
\newcommandx{\markt}[3][1=\hat{\al},2=\hat{q},3=\aguard,usedefault=@]{\%\big({#1} , {#2}, {#3}\big)}

\newcommandx{\borderfn}[2][1=\aconfig,2=\aloop,usedefault=@]{
  \mathsf{border}_{{#2}}\ifempty{#1}{}{({#1})}
}

\tikzset{
  mycallout/.style={
	 fill=gray!30, opacity=.5, overlay, align=center,
	 cloud callout, cloud puffs=10, aspect=1.9, cloud ignores aspect, cloud puff arc=100
  }
}

\newcommandx{\ggvisually}[8][1=5pt,2=15pt,3=5pt,4=5pt,5=1.0cm,6=\scriptsize,7={},8={},usedefault=@]{
  \def\dist{\hspace{#5}}
  $\begin{array}{c@{\dist}c@{\dist}c@{\dist}c@{\dist}c@{\dist}c}
	  \begin{tikzpicture}[node distance=0.9cm and 0.4cm, every node/.style={scale=.7,transform shape}]
		 \node[source] (srcint) {};
		 \node[sink,below=of srcint] (sinkint) {};
		 \node[mycallout, above = .3cm of srcint, xshift=1cm, callout absolute pointer={(srcint.east)}] {source node};
		 \node[mycallout, below = .3cm of sinkint, xshift=1cm, callout absolute pointer={(sinkint.west)}] {sink node};
		 \path[line] (srcint) -- (sinkint);
	  \end{tikzpicture}
	  &
		 \begin{tikzpicture}[node distance=0.9cm and 0.4cm, every node/.style={scale=.7,transform shape}]
			\mkint{}{int}[]
			\mkgraph{int}{int};
		 \end{tikzpicture}
	  &
		 \begin{tikzpicture}[node distance=.9cm and 0.4cm, every node/.style={scale=.7,transform shape}]
			\node[bblock] at (0,0) (g) {$\aG$};
			\node[node, below=of g] (s1) {};
			\node[bblock, below=of s1] (gp) {$\aG'$};
			\path[line,dotted] (g) -- (s1);
			\path[line,dotted] (s1) -- (gp);
		 \end{tikzpicture}
	  &
		 \begin{tikzpicture}[node distance=.4cm and 0.4cm, every node/.style={scale=.7,transform shape}]
			\node[bblock] at (-.7,0) (g) {$\aG$};
			\node[bblock] at (.7,0)  (gp) {$\aG'$};
			\node[node, above=of g] (f) {};
			\node[node, below=of g] (j) {};
			\node[node, above=of gp] (fp) {};
			\node[node, below=of gp] (jp) {};
			\path[line,dotted] (f) -- (g);
			\path[line,dotted] (g) -- (j);
			\path[line,dotted] (fp) -- (gp);
			\path[line,dotted] (gp) -- (jp);
			\mkfork{f,fp}[fork][][#1];
			\mkjoin{j,jp}[join][][#2];
			\mkgraph{fork}{join};
			\node[mycallout, above = .3cm of fork, xshift=-1cm, callout absolute pointer={(fork.west)}] {fork gate};
			\node[mycallout, above = -.9cm of join, xshift=-1cm, callout absolute pointer={(join.west)}] {join gate};
		 \end{tikzpicture}
	  &
		 \begin{tikzpicture}[node distance=.4cm and 0.4cm, every node/.style={scale=.7,transform shape}]
			\node[bblock] at (-.7,0) (g) {$\aG$};
			\node[bblock] at (.7,0)  (gp) {$\aG'$};
			\node[node, above=of g] (f) {};
			\node[node, below=of g] (j) {};
			\node[node, above=of gp] (fp) {};
			\node[node, below=of gp] (jp) {};
			\path[line,dotted] (f) -- (g);
			\path[line,dotted] (g) -- (j);
			\path[line,dotted] (fp) -- (gp);
			\path[line,dotted] (gp) -- (jp);
			\mkbranch{f,fp}[fork][][#3];
			\mkmerge{j,jp}[join][][#4];
         \mkgraph{fork}{join};
         \node[mycallout, above = .3cm of fork, xshift=-1cm, callout absolute pointer={(fork.west)}] {branch gate};
         \node[mycallout, above = -.9cm of join, xshift=-1cm, callout absolute pointer={(join.west)}] {merge gate};
       \end{tikzpicture}
     \ifempty{#7}{}{
     &
		 \begin{tikzpicture}[node distance=0.4cm and 0.4cm, every node/.style={scale=.7,transform shape}]
        \node[bblock] (g) {$\aG$};
        \node[node, above=.5cm of g] (f) {};
        \node[node, below=.5cm of g] (j) {};
        \path[line,dotted] (f) -- (g);
        \path[line,dotted] (g) -- (j);
        \mkloop[.4][1]{f}{j};
        \mkgraph[.3cm]{entryf}{exitj};
        \node[mycallout, above = .2cm of entryf, xshift=1.3cm, callout absolute pointer={(entryf.east)}] {loop entry};
        \node[mycallout, above = -.7cm of exitj, xshift=1.3cm, callout absolute pointer={(exitj.west)}] {loop exit};
      \end{tikzpicture}
     }
     \ifempty{#8}{}{
	  \\
     \text{#6 empty}
     &
     \text{#6 interaction}
     &
     \text{#6 sequential}
     &
     \text{#6 parallel}
     &
     \text{#6 branch}
		 &\text{#6 iteration}
   }
   \end{array}$
}

\newcommandx{\newggvisually}[5][1=5pt,2=15pt,3=\scriptsize,4={},5={},usedefault=@]{
  \def\w{1cm}
  \begin{minipage}[c]{\w}
	 \ifempty{#5}{}{\text{#3 empty}\\[#1]}
	 \begin{tikzpicture}[node distance=0.9cm and 0.4cm, every node/.style={scale=.7,transform shape}]
		\node[source] (srcint) {};
		\node[sink,below=of srcint] (sinkint) {};
		\node[mycallout, above = .3cm of srcint, xshift=1cm, callout absolute pointer={(srcint.east)}] {source node};
		\node[mycallout, below = .3cm of sinkint, xshift=1cm, callout absolute pointer={(sinkint.west)}] {sink node};
		\path[line] (srcint) -- (sinkint);
	 \end{tikzpicture}
  \end{minipage}
  \hfill
	  \begin{minipage}[c]{\w}
     \ifempty{#5}{}{\text{#3 interaction}\\[#1]}
		 \begin{tikzpicture}[node distance=0.9cm and 0.4cm, every node/.style={scale=.7,transform shape}]
			\mkint{}{int}[]
			\mkgraph{int}{int};
		 \end{tikzpicture}
	  \end{minipage}
	 \hfill
	  \begin{minipage}[c]{.1cm}
     \ifempty{#5}{}{\text{#3 sequential}\\[#1]}
		 \begin{tikzpicture}[node distance=.9cm and 0.4cm, every node/.style={scale=.7,transform shape}]
			\node[bblock] at (0,0) (g) {$\aG$};
			\node[node, below=of g] (s1) {};
			\node[bblock, below=of s1] (gp) {$\aG'$};
			\path[line,dotted] (g) -- (s1);
			\path[line,dotted] (s1) -- (gp);
		 \end{tikzpicture}
	  \end{minipage}
	  \hfill
	  \begin{minipage}[c]{\w}
     \ifempty{#5}{}{\text{#3 parallel}\\[#1]}
		 \begin{tikzpicture}[node distance=.4cm and 0.4cm, every node/.style={scale=.7,transform shape}]
			\node[bblock] at (-.7,0) (g) {$\aG$};
			\node[bblock] at (.7,0)  (gp) {$\aG'$};
			\node[node, above=of g] (f) {};
			\node[node, below=of g] (j) {};
			\node[node, above=of gp] (fp) {};
			\node[node, below=of gp] (jp) {};
			\path[line,dotted] (f) -- (g);
			\path[line,dotted] (g) -- (j);
			\path[line,dotted] (fp) -- (gp);
			\path[line,dotted] (gp) -- (jp);
			\mkfork{f,fp}[fork][][#1];
			\mkjoin{j,jp}[join][][#2];
			\mkgraph{fork}{join};
			\node[mycallout, above = .2cm of fork, xshift=-1cm, callout absolute pointer={(fork.west)}] {#3 fork gate};
			\node[mycallout, below = .2cm of join, xshift=-1cm, callout absolute pointer={(join.west)}] {#3 join gate};
		 \end{tikzpicture}
	  \end{minipage}
	  \hfill
	  \begin{minipage}[c]{\w}
     \ifempty{#5}{}{\text{#3 branch}\\[#1]}
		 \begin{tikzpicture}[node distance=.4cm and 0.4cm, every node/.style={scale=.7,transform shape}]
			\node[bblock] at (-.7,0) (g) {$\aG$};
			\node[bblock] at (.7,0)  (gp) {$\aG'$};
			\node[node, above=of g] (f) {};
			\node[node, below=of g] (j) {};
			\node[node, above=of gp] (fp) {};
			\node[node, below=of gp] (jp) {};
			\path[line,dotted] (f) -- (g);
			\path[line,dotted] (g) -- (j);
			\path[line,dotted] (fp) -- (gp);
			\path[line,dotted] (gp) -- (jp);
			\mkbranch{f,fp}[fork][][#1];
			\mkmerge{j,jp}[join][][#2];
         \mkgraph{fork}{join};
         \node[mycallout, above = .1cm of fork, xshift=-1cm, callout absolute pointer={(fork.west)}] {branch gate};
         \node[mycallout, below = .1cm of join, xshift=-1cm, callout absolute pointer={(join.west)}] {merge gate};
       \end{tikzpicture}
	  \end{minipage}
     \ifempty{#4}{}{
     \hfill
	  \begin{minipage}[c]{\w}
     \ifempty{#5}{}{\text{#3 iteration}\\[#1]}
		 \begin{tikzpicture}[node distance=0.4cm and 0.4cm, every node/.style={scale=.7,transform shape}]
        \node[bblock] (g) {$\aG$};
        \node[node, above=.5cm of g] (f) {};
        \node[node, below=.5cm of g] (j) {};
        \path[line,dotted] (f) -- (g);
        \path[line,dotted] (g) -- (j);
        \mkloop[.4][1]{f}{j};
        \mkgraph[.3cm]{entryf}{exitj};
        \node[mycallout, above = .2cm of entryf, xshift=1.3cm, callout absolute pointer={(entryf.east)}] {loop entry};
        \node[mycallout, above = -.7cm of exitj, xshift=1.3cm, callout absolute pointer={(exitj.east)}] {loop exit};
      \end{tikzpicture}
	  \end{minipage}
     }
}

  \newcommandx{\wwwcquote}[1][1=quo:w3c,usedefault=@]{
	 \ifempty{#1}{}{\begin{quote}\label{#1}}
		\lq\lq Using the Web Services Choreography specification, a
		\textcolor{orange}{contract} containing a global definition of the
		common \textcolor{orange}{ordering} conditions and constraints
		under which \textcolor{orange}{messages} are exchanged, is
		produced that describes, from a \textcolor{orange}{global
		  viewpoint} [...]  observable behaviour of all the parties
		involved.
		\textcolor{OliveGreen}{Each party} can then use the global definition to
		\textcolor{OliveGreen}{build and test solutions that conform to it}.
		The global specification is in turn \textcolor{OliveGreen}{realised by combination of} the
		resulting \textcolor{OliveGreen}{local systems} [...]\rq\rq
		\ifempty{#1}{}{\end{quote}}
  }

\newcommand{\tool}{\textsf{CAScr}}

\lstdefinelanguage{TypeScript}{
    morekeywords={typeof, new, true, false, catch, instanceof, =>, function, return, null, catch, switch, var, const, let, for, if, in, while, do, else, case, break, default, abstract},
    ndkeywords={class, export, boolean, string, any, number, throw, =>, extends, void, as, implements, import, this},
    comment=[l]{//},
    morestring=[b]",
    morestring=[b]',
    morestring=[b]`,
    morecomment=[s]{/*}{*/},
    sensitive=true,
}
\lstdefinelanguage{json}{
    literate=
     *{:}{{{\color{Periwinkle}{:}}}}{1}
      {,}{{{\color{Periwinkle}{,}}}}{1}
      {\{}{{{\color{Periwinkle}{\{}}}}{1}
      {\}}{{{\color{Periwinkle}{\}}}}}{1}
      {[}{{{\color{Periwinkle}{[}}}}{1}
      {]}{{{\color{Periwinkle}{]}}}}{1},
}
\lstdefinelanguage{Scribble}{
    morekeywords={aux, explicit, connect, disconnect, global, protocol, role, do, rec, choice, from, to, at, continue, or, rec, as, type},
    morestring=[b]",
    morestring=[b]',
    morecomment=[s]{/*}{*/},
    sensitive=true,
}
\definecolor{pine}{RGB}{1, 121, 111}

\begin{document}

\maketitle

\begin{abstract}
  Choreographic models support a correctness-by-construction principle
  in distributed programming.
  Also, they enable the automatic generation of correct message-based
  communication patterns from a global specification of the desired
  system behaviour.
  In this paper we extend the theory of choreography automata, a
  choreographic model based on finite-state automata, with two key
  features. First, we allow participants to act only in some of the
  scenarios described by the choreography automaton. While this seems
  natural, many choreographic approaches in the literature, and
  choreography automata in particular, forbid this behaviour. Second,
  we equip communications with assertions constraining the values that
  can be communicated, enabling a design-by-contract approach. We
  provide a toolchain allowing to exploit the theory above to
  generate APIs for TypeScript web programming. Programs communicating
  via the generated APIs follow, by construction, the prescribed
  communication pattern and are free from communication errors such as
  deadlocks.
\end{abstract}

\section{Introduction}\label{sec:intro}
The development of communicating systems is notoriously a challenging
endeavour.
In this application domain, both researchers and practitioners
consider choreographies a valid approach to tackle software
development
(e.g.~\cite{w3c:cho,bpmn,DBLP:journals/software/AutiliIT15,bon18,fmmt20}).
Besides being naturally geared toward scalability (due to the lack of
central components), choreographic models have been specifically
conceived to support a
\emph{correctness-by-construction}~\cite{w3c:cho} principle hinging on
the interplay between \emph{global} and \emph{local views}.
The former is a description of the interactions among (the \emph{role}
of) participants.
\iftr
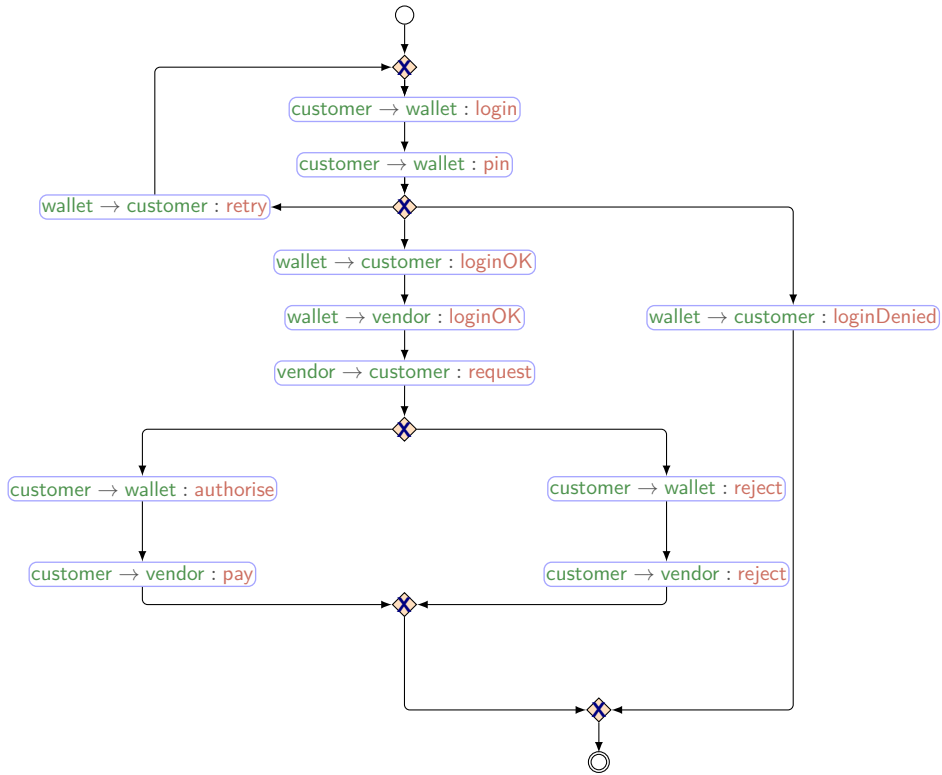
\begin{figure}[t]\centering
  \begin{minipage}{1\linewidth}\centering
	 \begin{tikzpicture}[scale=.8,every node/.style={transform shape, font=\fontsize{10}{30}\selectfont}]
		\mkint{}{login}[customer][wallet][login];
		\mkint{below = .5cm of login}{pin}[customer][wallet][pin];
		\mkseq{login}{pin};
		\mkgatebegin{login,login}[@][@][.7cm]{loopStart};
		\mkgateend{pin,pin}[@][@][.7cm]{loopEnd};
		\mkint{left = 2cm of loopEnd}{retry}[wallet][customer][retry];
		\mknseq{loopEnd,retry};
		\path[line] (retry) |- (loopStart);
		\mkint{below = .5cm of loopEnd}{ok1}[wallet][customer][loginOK];
		\mkint{below = .5cm of ok1}{ok}[wallet][vendor][loginOK];
		\mkint{below = .5cm of ok}{req}[vendor][customer][request];
		\mknseq{loopEnd,ok1,ok};
		\mkint{below left = of req, xshift = .75cm, yshift=-.5cm}{auth}[customer][wallet][authorise];
		\mkint{below right = of req, xshift = -.5cm,yshift=-.5cm}{rej}[customer][wallet][reject];
		\mkint{right = 2cm of ok}{ko}[wallet][customer][loginDenied];
		\path[line] (loopEnd) -| (ko);
		\mkint{below = of auth}{pay}[customer][vendor][pay];
		\mkint{below = of rej}{bye}[customer][vendor][reject];
		\mkseq{ok}{req};
		\mkseq{auth}{pay};
		\mkseq{rej}{bye};
		\mkbranch{auth,rej}[br][@][-5.3cm];
		\mkmerge{pay,bye}[br'][-1cm];
		\mkseq{req}{br};
		\mkmerge{br',ko}[end][@][6.5cm];
		\mkgraph{loopStart}{end};
	 \end{tikzpicture}
  \end{minipage}
\caption{BPMN-like description of the Online Wallet}
\label{fig:olwBPMN}
\end{figure}
\fi
We illustrate this through an OnLineWallet (OLW) service, adapted
from~\cite{spy} and akin to PayPal, used by vendors to process from
customers.
\iftr
\cref{fig:olwBPMN} describes our OLW protocol in terms of a
BPMN-like diagram (diamonds denote choices or merges).
The three participants involved are \ptp[customer], \ptp[wallet], and
\ptp[vendor].
\else
(See~\cite{glsty22TR} for a visual model of OLW).
This protocol involves three participants: \ptp[customer],
\ptp[wallet], and \ptp[vendor].
\fi
The former tries first to $\amsg[login]$ into its account on the
\ptp[wallet]\ server.
In case of failure, \ptp[wallet]\ may ask for a $\amsg[retry]$,
or may decide to deny access.
A successful authentication is communicated by \ptp[wallet]\ to
\ptp[customer]\ and \ptp[vendor]\ through the $\amsg[loginOK]$ message;
the \ptp[vendor]\ then sends a $\amsg[request]$ for payment to \ptp[customer], who
can $\amsg[authorise]$ or $\amsg[reject]$ the transaction.

A natural question to ask is \quo{can the OLW protocol be faithfully
  realised by distributed components?}
The answer to this question requires a careful formalisation which
we carry out in the next sections.
For the moment, we appeal to intuition and interpret \emph{realisation}
as the existence of a set of components that coordinate with each
other exclusively by message-passing and \emph{faithful} as the fact
that components execute all and only the communications prescribed by
the global view without incurring in communication errors such as
deadlocks.
Local views specify the behaviour of each participant \quo{in
  isolation}.
\iftr
For instance, the local view of \ptp[vendor]\ above consists of
an artefact which, after having received the notification from
\ptp[wallet], sends a $\msg[request]$ message to \ptp[customer], and
then waits for either a $\msg[pay]$ment or a $\msg[reject]$ion message
from \ptp[customer].
\else
For instance, the local behaviour of \ptp[vendor]\ in the OLW protocol
is to wait the notification from \ptp[wallet], send a $\msg[request]$
message to \ptp[customer], and then wait for either a $\msg[pay]$ment
or a $\msg[reject]$ion message from \ptp[customer].
\fi
Note that \ptp[vendor]\ is \quo{oblivious} of the interactions between
\ptp[customer]\ and \ptp[wallet].
Also, observe that, if \ptp[customer]\ fails to authenticate to
\ptp[wallet]\ (e.g., by typing a wrong password), then no payment
request can be made.
In this case, it does not make sense to involve \ptp[vendor]\ in the
protocol.
We call the ability to involve a participant only in some branches of
a protocol \emph{selective participation}.

Rather than an exception, selective participation is a norm in
distributed applications, e.g., for data validation, prevention of
server overload, or access control.
Consider, e.g., services giving public access to some resources
while requiring authentication to grant access to others.
Often, the authentication phase is outsourced to external services
(e.g., providing OAuth2.0~\cite{oauth2} and Kerberos~\cite{kerberos}
authentication).
In this case, accesses to public resources should be oblivious to
authentication services while protected resources are not involved in
the communication until the authentication phase is cleared (as for
\ptp[vendor]\ in our example).
Other examples of selective participation emerge from smart contracts
for online money transations (e.g., crowdfunding services
as~\cite{kickstarter}), where participants take part to some stages of
the communication only in case of a positive outcome of some
financial operation.

A paramount element for the correctness-by-construction principle is
the notion of \emph{well-formedness}, namely sufficient conditions
guaranteeing the faithful realisation of a protocol.
Actually, choreographies advocate the algorithmic derivation, by
\emph{projection}, of faithful realisations from well-formed global
views~\cite{w3c:cho}.
In fact, the so-called \emph{top-down} choreographic approach to
development consists of ($a$) the definition of a well-formed global
view of the protocol, ($b$) the projection of the global view onto local
ones, ($c$) the verification that each implemented component complies with
a local view.

Usually, global views abstract away from local computations; for
instance,
\iftr
the diagram of OLW in \cref{fig:olwBPMN}
\else
our description of OWL
\fi
does not specify how \ptp[Wallet]\ takes the decision of letting
\ptp[Customer]\ $\amsg[retry]$ the authentication or the strategy of
\ptp[Customer]\ to $\amsg[authorise]$ or not the payment.
Both these (and the other local) computations are blurred away because
they require to specify the data dependencies that local computations
should enforce.
As pioneered in~\cite{bhty10} in the context of global
types~\cite{HondaYC16}, assertion methods can abstractly handle those
dependencies by suitably constraining the payloads of interactions.
Roughly, this transfers design-by-contract~\cite{mey92} methods to
message-passing applications by imposing rely-guarantee relations on
interactions.
As shown in~\cite{bhty10}, this poses several challenges due to two
main reasons.
Firstly, pre-conditions ensuring the feasibility of some interactions
depend on information scattered across distributed participants.
Hence, it is necessary that data flow to participants so that all the
information necessary for a participant to guarantee some assertion is
available when needed.
This requires to restrict to \emph{history sensitive}~\cite{bhty10}
protocols, namely specifications have to be such that participants
required to guarantee an assertion are aware of the information needed
to satisfy it.
Secondly, a careless use of such assertions may lead to inconsistent
specifications so to eventually spoil the realisability of the
protocol.
This requires to restrict to \emph{temporally
  satisfiable}~\cite{bhty10} protocols, where no assertion ever
becomes inconsistent during the execution.

Models and results based on the top-down approach to choreography
abound in the literature (see, e.g., the survey~\cite{survey}).
This paper builds on \emph{choreography automata}
(c-automata)~\cite{BarbaneraLT20}; intutively, a c-automaton is a
finite-state machine whose transitions are labelled by
\iftr
interactions such as those in \cref{fig:olwBPMN}.
\else
interactions.
\fi
The use of automata brings several benefits.
On the one hand, automata models are well-known to both academics and
industrial computer scientists and engineers.
On the other hand, they allow one to exploit the well-developed theory
of automata.
Furthermore, automata do not have syntactic constraints imposed by
algebraic models such as multiparty session types (see,
e.g.,~\cite{HYC08,ScalasY19,CarboneM13}).
Indeed, as noted in~\cite{BarbaneraLT20}, c-automata seem to be
more flexible than \quo{syntax}-based formalisms such as
global graphs~\cite{gt18} or multiparty session types.
This is due to the fact that, in the latter family, well-formedness is
attained via syntactic restrictions that rule out unrealisable
protocols.
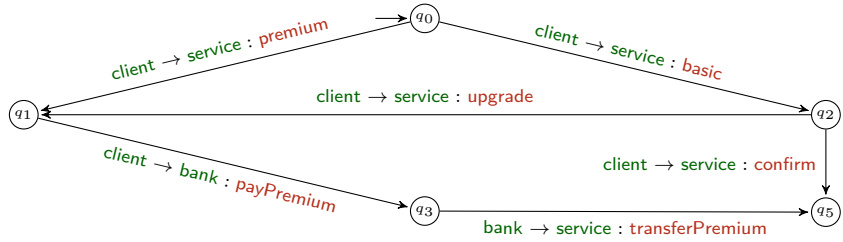
\begin{figure}[t]
  \[\begin{tikzpicture}[node distance=1cm and 5cm]
		\tikzstyle{every state}=[cnode]
		\tikzstyle{every edge}=[carrow]
		\node[state, initial, initial text={}] (q0) {$q_0$};
		\node[state, below left = of q0] (q1) {$q_1$};
		\node[state, below right = of q0] (q2) {$q_2$};
		\node[state, below right = of q1] (q3) {$q_3$};
		\node[state, below = of q2, yshift=.1cm] (q5) {$q_5$};
		\path (q0) edge node[above] {$\gint[client][Service][{\amsg[premium]}]$} (q1);
		\path (q0) edge node[above] {$\gint[client][Service][{\amsg[basic]}]$} (q2);
		\path (q2) edge node[above] {$\gint[client][Service][{\amsg[upgrade]}]$} (q1);
		\path (q1) edge node[below]{$\gint[client][bank][{\amsg[payPremium]}]$} (q3);
		\path (q2) edge node[left,rotate=90] {$\gint[client][Service][{\amsg[confirm]}]$} (q5);
		\path (q3) edge node[below] {$\gint[bank][Service][{\amsg[transferPremium]}]$} (q5);
	 \end{tikzpicture}\]
  \caption{Non-well-structured choreography}\label{fig:nonstru}
\end{figure}
Indeed, a distinguished feature of c-automata is that they admit
\emph{non-well-structured} interactions.
Let us explain this with the c-automaton in \cref{fig:nonstru},
modelling a choreography where a \ptp[client]\ registers to a
\ptp[service]\ according to two options.
If \ptp[client]\ opts for the basic level, then no payment is due,
while the premium option requires a \ptp[bank]\ payment.
Thus, we have a choice at $q_0$ between the $\msg[basic]{}$ and
$\msg[premium]{}$ service levels.
Then, in state $q_2$ of \cref{fig:nonstru}, \ptp[client]\ either
$\msg[confirm]$s the choice or decides to $\msg[upgrade]$.
(Selective participation is required since the bank only
acts in the \quo{left} run.)
In a structured model, the \quo{left} and the \quo{right} runs from
$q_0$ to $q_5$ must be different branches of a choice.
But those models cannot encode the $\msg[upgrade]{}$ transition that
intuitively allows one to move from one branch to the
other, before the end of the choice construct.

\noindent
\textbf{Contribution and structure.}  We provide two main
contributions to the theory of c-automata, as well as an
implementation in the setting of TypeScript programming.

First, we extend c-automata with selective participation, which,
although natural as seen above, is actually forbidden in many
choreographic models (e.g.,~\cite{HYC08,ScalasY19,CarboneM13})
including c-automata~\cite{BarbaneraLT20}.
For instance, we will use the OLW protocol, where \ptp[vendor]'s
involvement occurs only on successful authentication, as our running example.
  Our second contribution is the definition of \emph{asserted
	 c-automata}, that is a design-by-contract framework for
  c-automata.
  More precisely, we equip transitions with assertions constraining
  the exchanged messages, allowing one to specify such policies.
  For example, we can specify that the authentication of OLW
  \ptp[customer]{} can fail at most $3$ times.
  At a glance, asserted c-automata mimick the constructions introduced
  in~\cite{bhty10}.
  However, the generalisation of c-automata to selective participation
  (not featured in~\cite{bhty10}) and the greater flexibility
  introduced by non well-structured interactions require to address
  non-trivial technical challenges that we discuss in
  \cref{sec:achor}.
%\end{resub}

%\begin{resub}
  %
  The last contribution is a toolchain, dubbed \tool{}, based on the
  theory of c-automata with selective participation developed in this
  paper.
  More precisely, \tool{} allows one to specify a protocol using the
  Scribble framework~\cite{scribble-paper,FeatherweightScribble,nuscr}
  and to check its well-formedness relying on our theory.
%\end{resub}
%
Finally, \tool{} generates TypeScript APIs to implement the roles of
the original protocol.
%\begin{resub}
To the best of our knowledge, \tool{} is the first toolchain that
integrates Scribble
with the flexibility of the theory of c-automata.
%\end{resub}

Our paper is structured as follows. \cref{sec:back} introduces
notions on finite state automata, and in particular on
communicating finite state machines, to model
participants, and on c-automata.

\cref{sec:theory} develops the theory of c-automata. The main novelty
w.r.t.~\cite{BarbaneraLT20} is to allow for selective
participation. The resulting framework is more flexible with respect
to~\cite{BarbaneraLT20}, e.g., it allows one to prove that the OLW
protocol can be faithfully projected. Even in this more general
setting we can prove standard results: the implementation has the same
behaviour as the original specification
(\cref{th:projectionCorrectness}) and is free from deadlocks
(\cref{thm:df}).
Also, when focusing on one of the participants the projected system is
lock free (\cref{thm:lf}).

\cref{sec:achor} develops our second contribution, namely
design-by-contract in the setting of c-automata. More precisely,
c-automata are extended with assertions (\cref{def:aca}) and the
related theory is extended accordingly. Also in this setting the
implemented system faithfully executes its specification
(\cref{{th:acprojectionCorrectness}}) and it is deadlock free
(\cref{thm:ac-df}).

\cref{sec:apply} presents \tool{},
a novel, full toolchain---from the
%\begin{resub}
  Scribble \cite{scribble-paper,FeatherweightScribble,nuscr} %\end{resub}
specification of the communication protocol, to the generation of APIs---providing support for
distributed web development in TypeScript
%\begin{resub}
and relying on flexible c-automata with selective participation. %\end{resub}.

Finally, \cref{sec:related} discusses related work, while
\cref{sec:conc} draws some conclusions,
and sketches future directions.
\iftr A possible extension of our implementation is discussed in the appendix.
\else Proofs and auxiliary material can be found on the full version of the paper~\cite{glsty22TR}.
\fi

%%% Local Variables:
%%% mode: latex
%%% TeX-master: "main"
%%% End:

\section{Choreography Automata and Communicating Systems}\label{sec:back}
This section recalls basic notions about automata in general and
about choreography automata (c-automata)~\cite{BarbaneraLT20} and
systems of Communicating Finite State Machines (CFSMs)~\cite{bz83} in particular.
%
% These notions are instrumental for the technical development in the rest of the paper.
% In choreographic development,
% systems are described from a global point of view using some
% choreographic model.
% %
% In this paper we adopt c-automata, introduced in~\cite{BarbaneraLT20}.
% %
% As in~\cite{BarbaneraLT20}, c-automata are projected into systems of
% local descriptions modelled in terms of CFSMs.
% %
% In order to introduce c-automata, we first define finite-state
% automata (FSA).
%
Following~\cite{BarbaneraLT20}, global views, rendered as c-automata,
are projected into systems of local descriptions modelled as
CFSMs.
%
%In the next sections, we extend the theory in the literature by adding
%both the ability for participants to join only in some branches of a
%choice (selective participation) and assertions supporting
%design-by-contract.
We start by surveying finite-state automata (FSA).
\begin{definition}[FSA]\label{def:fsa}
  A \emph{labelled transition system} (LTS) is a tuple
  $\autom{Q}{q_0}{\lset}{\tset}$ where
  \begin{itemize}
  \item $Q$ is a set of states (ranged over by $s,q,\ldots$) and
	 $q_0 \in Q$ is the \emph{initial state};
  \item $\lset$ is a finite set of labels (ranged over by
	 $\ell,\ldots$);
  \item $\tset \subseteq Q \times (\lset \cup \set \emptyword) \times Q$
	 is a set of transitions where $\emptyword \not\in \lset$ is a
	 distinguished label.
  \end{itemize}
  A \emph{finite-state automaton} (FSA) is an LTS whose set of states
  is finite.
\end{definition}
When the LTS $A = \autom{Q}{q_0}{\lset}{\tset}$ is understood
we use the usual notations $s_1 \arro{\ell} s_2$ for the transition
$(s_1,\ell,s_2) \in \tset$ and $s_1 \arro{} s_2$ when there
exists $\ell$ such that $s_1 \arro{\ell} s_2$, as well as
$\arro{}^\star$ for the reflexive and transitive closure of $\arro{}$.
We denote as $\transset{\chora}{q}$ the set of transitions from $q$ in
$A$.
We occasionally write $q \in A$ and $(q, \aint, q') \in A$ instead of,
respectively, $q \in Q$ and $(q, \aint, q') \in \tset$, and likewise
for $\_\subseteq\_$.
%
%Hereafter, $\runset[A]$ is the set of finite paths of an LTS $A$
%starting from the initial state of $A$.
%
We recall standard notions on LTSs.
\begin{definition}[Traces and trace equivalence]\label{def:traces}
  A \emph{run} of an LTS $A = \conf{\sset, s_0, \lset, \tset}$ is a
  (possibly empty) finite or infinite sequence
  $\arun = (s_i \arro{\ell_i} s_{i+1})_{0 \leq i \leq n}$ of consecutive
  transitions starting at $s_0$ (assume $n = \infty$ if the run is
  infinite).
  The \emph{trace} (or \emph{word}) of $\pi$ is the concatenation of
  the labels $\traceof[\arun]$ of the run $\pi$, namely
  $\traceof[\arun] = \ell_0 \cdot \ell_1 \cdots \ell_n$.
  As usual, $\emptyword$ denotes the identity element of concatenation
  and the trace of an empty run is $\emptyword$.
  Function $\traceof[\cdot]$ extends homomorphically to sets of runs.
  Also, $s$-runs and $s$-traces of $A$ are, respectively, runs and
  traces of ${\conf{\sset, s, \lset, \tset}}$.
  The \emph{language of $A$} is
  $\lang[A] = \set{\traceof[\arun] \sst \arun \text{ is a run of }
	 A}$; \emph{$A$ accepts $\aword$} if $\aword \in \lang[A]$ and
  \emph{$A$ accepts $\aword$ from $s$} if
  $\aword \in \lang[\conf{\sset, s, \lset, \tset}]$.
  LTSs $A$ and $B$ are \emph{trace equivalent} if
  $\lang[A] = \lang[B]$.
\end{definition}

Bisimilarity~\cite{SangiorgiBisim} is an equivalence relation on LTSs
simpler to prove than trace equivalence which is implied by
bisimilarity, and coincides with it for deterministic LTSs.
\begin{definition}[Bisimulation]\label{def:bisim}
  % Let $A = \conf{\sset_A, s_{0A}, \lset, \tset_A}$ and
  % $B = \conf{\sset_B, s_{0B}, \lset, \tset_B}$ be two LTSs.  A relation
  % $\rcal \subseteq \sset_A \times \sset_B$ is a \emph{bisimulation} if
  % $(s_{0A},s_{0B}) \in \rcal$ and for every pair of states
  % $(p,q) \in \rcal$ and all labels $\ell$:
  % \begin{enumerate}
  % \item\label{it:bisim1} if $p \arro{\ell} p'$, then there is $q \arro{\ell} q'$ such
  % 	 that $(p',q') \in \rcal$;
  % \item\label{it:bisim2} if $q \arro{\ell} q'$, then there is $p \arro{\ell} p'$ such
  % 	 that $(p',q') \in \rcal$.
  % \end{itemize}
  Let $A = \conf{\sset_A, s_{0A}, \lset, \tset_A}$ and
  $B = \conf{\sset_B, s_{0B}, \lset, \tset_B}$ be two LTSs.
  A relation
  $\rcal \subseteq (\sset_A \times \sset_B) \cup (\sset_B \times
  \sset_A)$ is a \emph{(strong) bisimulation} if it is symmetric,
  $(s_{0A},s_{0B}) \in \rcal$, and for every pair of states
  $(p,q) \in \rcal$ and all labels $\ell$:
  \begin{align*}
	 \text{if } p \arro \ell p' \text{ then there is }
	 q \arro \ell q' \text{ such that } (p',q') \in \rcal
  \end{align*}
	 Relation $\rcal$ is a \emph{weak} bisimulation  if it is symmetric,
  $(s_{0A},s_{0B}) \in \rcal$, and for every pair of states
  $(p,q) \in \rcal$ and all labels $\ell$:
  \begin{itemize}
  \item if $p \arro \ell p'$ with $\ell \neq \emptyword$ then there is a run
	 $q \arro{\emptyword}^\star \arro \ell \arro{\emptyword}^\star
	 q'$ such that $(p',q') \in \rcal$ and
	\item if $p \arro \emptyword p'$ then there is a
	  run $q \arro{\emptyword}^\star q'$
          %of zero or more
	  %$\emptyword$-transitions
          such that $(p',q') \in \rcal$.
  \end{itemize}
\end{definition}
If two LTSs are bisimilar then they are
also trace equivalent.

A main role in our models is played by \emph{interactions} built on
the alphabet:
\begin{align*}
  \lint \mmdef & \Set{\gint \sst \p \neq \q \in \ptpuniv
				\qand \amsg \in \msgset}
  &&
	  \text{ranged over by lowercase Greek letters}
\end{align*}
where $\ptpuniv$ and $\msgset$ are, respectively, sets of participants
and of messages.
We assume $\ptpuniv \cap \msgset = \emptyset$.
An interaction $\gint$ specifies that participant \p sends a message
(of type) $\amsg$ to participant \q and participant \q receives
$\amsg$.
Hence, by construction, each send is paired with a unique receive and
vice versa.
In most choreographic models, this forbids to specify message losses, races, and deadlocks.
%
%\il{
  Adopting the terminology of the session type
  community (see, e.g.,~\cite{survey}),
  \begin{itemize}
  \item with \emph{message loss} we mean a send that cannot be matched
	 by a receive; this cannot happen in interactions since $\gint$
	 specifies both the send and the receive together;
	 \item with \emph{race} we mean a configuration where a receiver non-deterministically interacts with either of two senders (or a sender with either of two receivers),
		% with \emph{race} we mean a configuration where a receive can interact with two sends (or a send with two receives),
		depending
	 on the relative speed of their execution; this
	 cannot happen since an interaction specifies which send is
	 supposed to interact with which receive and vice versa (notably,
	 concurrency can take place without message races, e.g., if
	 participant $\p$ sends to participant $\q$ and at the same time
	 $\ptp[C]$ sends to $\ptp[D]$ there is no race);
		% Finally, deadlocks are states where
  \item with \emph{deadlock} we mean a configuration where two or more
	 participants are blocked waiting for one another
		% (e.g., both on a receive)
	 forming cyclic dependencies (e.g., \p is waiting for \q which
	 waits for \ptp[c], which in turns waits for \p); this cannot
	 happen either since an interaction specifies which participant has
	 to send and which one has to receive.
	 \end{itemize}
	 All these properties hold by construction in most choreographic
	 models. However, care is needed to ensure that these properties
	 are preserved when moving from the choreographic specification to
	 a distributed implementation. Such analysis has been performed for
	 many choreographic models in the literature (see~\cite{survey}).
 %}
\begin{definition}[Choreography automata]\label{def:chorautomata}
  A \emph{choreography automaton} (c-automaton) is an
  FSA on the alphabet $\lint$.
  Elements of $\lint^\star \cup \lint^\omega$ are \emph{choreography
	 words}, subsets of $\lint^\star \cup \lint^\omega$ are
  \emph{choreography languages}.
\end{definition}

The set of participants of a c-automaton is finite; we denote with
$\ptpset_{\chora}$ (or simply $\ptpset$ if $\chora$ is understood) the
set of participants of c-automaton $\chora$.
Given $\gint \in \lint$, we define $\ptpof[\gint] \mmdef \set{\p,\q}$
and extend it homorphically to (sets of) transitions.
We say that $\aint,\aint[b] \in \lint$ are \emph{independent}, written
$\aint \parallel \aint[b]$, if
$\ptpof[\aint] \cap \ptpof[{\aint[b]}] = \emptyset$.

\begin{example}[OLW's c-automaton]\label{ex:OLWca}
The c-automaton
  % \cref{fig:OLWca} represents the OLW example described in the
  % Introduction as a c-automaton.\finex
%%%  \begin{figure}[t]
    \[\begin{tikzpicture}[node distance=1cm and 3cm]
      \tikzstyle{every state}=[cnode]
      \tikzstyle{every edge}=[carrow]
      \node[state, initial, initial text={\aCM = }] (q0) {$q_0$};
		\node[state, below = 2cm of q0] (q1) {$q_1$};
		\node[state, right = of q1] (q2) {$q_2$};
		\node[state, right = of q2] (q3) {$q_3$};
		\node[state, right = of q0] (q4) {$q_4$};
		\node[state, right = of q4] (q4') {$q_5$};
		\node[state, right = of q4'] (q5) {$q_6$};
		\node[state, right = of q3] (q6) {$q_7$};
		\node[state, below left = of q5,yshift=.7cm,xshift=.5cm] (q7) {$q_8$};
		\path (q0) edge node[below] {$\gint[C][W][{\amsg[login]}]$} (q1);
		\path (q1) edge node[above,xshift=-.5cm] {$\gint[C][W][{\amsg[pin]}]$} (q2);
		\path (q2) edge[bend left] node[above,yshift=.3cm] {$\gint[W][C][{\amsg[retry]}]$} (q0);
		\path (q2) edge node[above] {$\gint[W][C][{\amsg[loginDenied]}]$} (q3);
		\path (q2) edge node[above] {$\gint[W][C][{\amsg[loginOK]}]$} (q4);
		\path (q4) edge node[above] {$\gint[W][V][{\amsg[loginOK]}]$} (q4');
		\path (q4') edge node[above] {$\gint[V][C][{\amsg[request]}]$} (q5);
		\path (q5) edge node[above] {$\gint[C][W][{\amsg[authorise]}]$} (q6);
		\path (q5) edge node[below] {$\gint[C][W][{\amsg[reject]}]$} (q7);
		\path (q6) edge node[below] {$\gint[C][V][{\amsg[pay]}]$} (q3);
		\path (q7) edge node[above] {$\gint[C][V][{\amsg[reject]}]$} (q3);
	 \end{tikzpicture}
    \]
	 models the OLW example in \cref{sec:intro}.
	 %%% \caption{OLW as a c-automaton (for readability we use initials of participants' names)}\label{fig:OLWca}
%%%  \end{figure}
%
	 \finex
\end{example}

We now survey communicating systems~\cite{bz83}, our formal model of local views.

\begin{definition}[Communicating system]
  A \emph{communicating finite-state machine} (CFSM) is an FSA on the
  set
  %\begin{align*}
    $\lact  \mmdef \{\aout, \ain \mid \p, \q \in \ptpuniv \qand  \amsg \in \msgset\}$
  %\end{align*}
  of \emph{actions}.
  
  Action $\aout$ is the send of message $\amsg$ from \p to \q, while
  action $\ain$ is the corresponding receive.
  The \emph{subjects} of an output and an input action, say $\aout$
  and $\ain$, are respectively \p and \q.
  A CFSM is \emph{\p-local} if all its transitions have labels with
  subject \p.  A \emph{(communicating) system} is a map
  $\aCS = (\aCM_{\p})_{\p \in \ptpset}$ assigning a $\p$-local CFSM
  $\aCM_{\p}$ to each participant $\p \in \ptpset$.
  We require that $\ptpset \subseteq \mathfrak{P}$ is finite and that
  any participant occurring in a transition of $\aCM_\p$ is in
  $\ptpset$.
\end{definition}

We now introduce the notion of projection from c-automata to systems
of CFSMs. %\il{
Intuitively, projection builds a system aimed at implementing the projected c-automaton. %}
Similar notions in the literature often take the name of
endpoint projection (see, e.g.,~\cite{HondaYC16,CarboneM13}).

\begin{definition}[Automata projection]\label{def:projection}
  The \emph{projection $\proj{\aint}{p}$ of an interaction
  $\aint$ on $\p \in \ptpuniv$} is
  \[
    \proj{(\gint)}{p} = \aout,
	 \hfill
	 \proj{(\gint[q][p])} p = \ain[q][p], \qqand
	 \hfill
	 \proj{\aint} p = \emptyword \text{ for any other label $\aint$}
  \]
  % \[
  %   \proj{\aint}{p} \mmdef \begin{cases}
  %     \aout & \text{if }  \aint = \gint
  %     \\
  %     \ain[q][p] &  \text{if } \aint=\gint[q][p]
  %     \\
  %     \emptyword & \text{otherwise}
  %   \end{cases}
  % \]
  Function $\proj\_ p$ extends homomorphically to transitions, runs,
  and choreography words.

  The \emph{projection} $\proj{{\chora}}{p}$ of a c-automaton
  $\chora = \conf{\sset, q_0, \lint, \tset}$ on a participant
  $\p \in \ptpset$ is obtained by determinising and minimising up-to
  language equivalence the \emph{intermediate} CFSM
  \[
    A_\p = \conf{\sset, q_0, \lact, \Set{(q \arro{\proj{\aint}{p}} q') \sst q \arro{\aint} q'  \in \tset}}
  \]
  The \emph{projection of $\chora$}, written $\proj{\chora}{}$, is the
  communicating system
  $(\proj{\chora}{\p})_{\p \in \ptpset}$.
\end{definition}

\begin{example}[Projecting OLW]\label{ex:OLWproj}
  We instantiate here projection on the c-automaton for
  the OLW protocol described in \cref{ex:OLWca}.
  In particular, the intermediate CFSM $A_{\ptp[v]}$ is
  % \begin{figure}[t]
  \[\begin{tikzpicture}[node distance=1cm and 3cm]
		\tikzstyle{every state}=[cnode]
		\tikzstyle{every edge}=[carrow]
      \node[state, initial, initial text={$A_{\ptp[v]}$ = }] (q0) {$q_0$};
		\node[state, below = 1.5cm of q0] (q1) {$q_1$};
		\node[state, right = of q1] (q2) {$q_2$};
		\node[state, right = of q2] (q3) {$q_3$};
		\node[state, right = of q0] (q4) {$q_4$};
		\node[state, right = of q4] (q4') {$q_5$};
		\node[state, right = of q4'] (q5) {$q_6$};
		\node[state, right = 1cm of q5] (q6) {$q_7$};
		\node[state, below left = of q5,xshift=.5cm,yshift=.7cm] (q7) {$q_8$};
		\path (q0) edge node[below] {$\emptyword$} (q1);
		\path (q1) edge node[below] {$\emptyword$} (q2);
		\path (q2) edge[bend left] node[above,yshift=.3cm] {$\emptyword$} (q0);
		\path (q2) edge node[below] {$\emptyword$} (q3);
		\path (q2) edge node[above] {$\emptyword$} (q4);
		\path (q4) edge node[above] {$\ain[W][V][{\amsg[loginOK]}]$} (q4');
		\path (q4') edge node[above] {$\aout[V][C][{\amsg[request]}]$} (q5);
		\path (q5) edge node[above] {$\emptyword$} (q6);
		\path (q5) edge node[below] {$\emptyword$} (q7);
		\path (q6) edge[bend left=5] node[above] {$\ain[C][V][{\amsg[pay]}]$} (q3);
		\path (q7) edge node[above] {$\ain[C][V][{\amsg[reject]}]$} (q3);
	 \end{tikzpicture}
  \]
  the determinisation of which yields the following CFSM $\proj{\chora}{{\ptp[v]}}$ for the vendor
  participant
  \[\begin{tikzpicture}[node distance=1cm and 3cm]
		\tikzstyle{every state}=[cnode]
		\tikzstyle{every edge}=[carrow]
      \node[state, initial, initial text={$\proj{\chora}{{\ptp[v]}}$ = }] (q0) {$Q_0$};
		\node[state, right = of q0] (q5) {$Q_1$};
		\node[state, right = of q5] (q5') {$Q_2$};
		\node[state, right = of q5'] (q3) {$Q_3$};
		\path (q0) edge node[above] {$\ain[W][V][loginOK]$} (q5);
		\path (q5) edge node[above] {$\aout[V][C][{\amsg[request]}]$} (q5');
		\path (q5') edge[bend left=7] node[above] {$\ain[C][V][{\amsg[pay]}]$} (q3);
		\path (q5') edge[bend right=7] node[below] {$\ain[C][V][{\amsg[reject]}]$} (q3);
	 \end{tikzpicture}
  \]
  Noteworthy, due to determinisation, states of the projection
  correspond to (not necessarily disjoint) sets of states of the
  starting c-automaton.   Indeed, in $\proj{\chora}{{\ptp[v]}}$ we have $Q_0 = \set{q_0,q_1,q_2,q_3,q_4}$, $Q_1 = \set{q_5}$, $Q_2 = \set{q_6,q_7,q_8}$, and $Q_2 = \set{q_3}$.
  \finex
\end{example}

We present below the semantics of communicating systems.
% \il{
We consider a synchronous semantics. Essentially, a system can execute
an interaction $\gint$ if two of its participants can provide
complementary actions $\aout$ and $\ain$ (while the others do not
move), and can take an $\epsilon$ action if one of its participant can
do it (while the others do not move).
% }
\begin{definition}[Semantics of communicating systems]\label{def:syncSem}
  Let $\aCS = (\aCM_{\p})_{\p \in \ptpset}$ be a communicating system
  where $\aCM_{\p} = \conf{\sset_{\p}, q_{0\p}, \lact, \tset_{\p}}$
  for each participant $\p \in \ptpset$.

  A \emph{configuration of } $\aCS$ is a map
  $\aConf = (q_{\p})_{\p \in \ptpset}$ assigning a \emph{local state}
  $q_{\p}\in \sset_{\p}$ to each $\p \in \ptpset$.
  %
  % \eMcomm[obvious since $\aConf$ is a map]{We} denote $q_{\p}$ by
  % $\aConf(\p)$.
  %
  The \emph{semantics} of $\,\aCS$ is the c-automaton
  $\ssem{\aCS}=\conf{\sset,\aConf_0,\lint,\tset}$ where
  \begin{itemize}
  \item $\sset$ is the set of configurations of $\aCS$, as
	 defined above, and $\aConf_0: \p \mapsto q_{0\p}$ for each
	 $\p \in \ptpset$ is the \emph{initial} configuration of $\sset$
       \item $\tset$ is the set of transitions
         \begin{itemize}
	 \item $\aConf_1 \arro{\gint} \aConf_2$ such that
    % \in\, \tset$  if
           %\begin{enumerate}[(a)]
           \begin{itemize}
    \item $\aConf_1(\p) \arro{\aout} \aConf_2(\p) \in \tset_{\p}$ and
      $\aConf_1(\q) \arro{\ain} \aConf_2(\q) \in \tset_{\q}$,
      and
    \item for all $\ptp[x] \in \ptpset \setminus \set{\p,\q}$, $\aConf_1(\ptp[x]) =
      \aConf_2(\ptp[x])$
           \end{itemize}
    %\end{enumerate}
  \item $\aConf_1 \arro[]{\emptyword} \aConf_2$ such that $\aConf_1(\p)
	 \arro{\emptyword} \aConf_2(\p) \in \tset_{\p}$, and for all
	 $\ptp[x] \in \ptpset \setminus \set \p$, $\aConf_1(\ptp[x]) =
	 \aConf_2(\ptp[x])$.
         \end{itemize}
  \end{itemize}
\end{definition}

%%% Local Variables:
%%% mode: latex
%%% TeX-master: "main"
%%% End:

\section{Flexible Choreography Automata}\label{sec:theory}
We now introduce a theory of c-automata enabling faithful
realisations, %\il{,
which is formalised as language equivalence between
  a c-automaton and the semantics of its projection as proved
  in~\cref{th:projectionCorrectness}. However, not all c-automata can
  be faithfully realised, hence we need to restrict to
  \emph{well-formed} c-automata. Well-formedness is defined as %}
    the
conjunction of two properties, \emph{well-sequencedness} and
\emph{well-branchedness}.
Both these properties are inspired from~\cite{BarbaneraLT20}.
However, well-branchedness is generalised to allow participants to act
on some of the scenarios specified by the c-automaton only upon
request from other participants.
We call this feature \emph{selective participation}, since a
participant may act on a branch only if %\il{
selected to be involved %}
by
some other participant.
This is disallowed in many choreographic formalisms
(e.g.,~\cite{HYC08,ScalasY19,CarboneM13}), including choreography
automata~\cite{BarbaneraLT20}.
On the other hand, well-sequencedness is strengthened since the
formulation in~\cite{BarbaneraLT20} is not enough to ensure
faithful realisations.
%
%\il{
We start by defining concurrent transitions, exploited in the
definition of well-sequencedness.
\begin{definition}[Concurrent transitions]\label{def:concurrent}
  Two
consecutive transitions $q \arro{\aint} q' \arro{{\aint[b]}} q''$ are \emph{concurrent} if there is $q'''$ such that
	 $q \arro{\aint[b]} q''' \arro{\aint} q''$.
\end{definition}
Essentially, two transitions are concurrent if they give rise to a
commuting diamond.  %}

\begin{definition}[Well-sequencedness]\label{def:wellseq}
A c-automaton $\chora$ is \emph{well-sequenced} if for each two
consecutive transitions $q \arro{\aint} q' \arro{{\aint[b]}} q''$
either
  \begin{enumerate}[(a)]
  \item\label{seq:dep} $\aint \not\parallel \aint[b]$, i.e., $\aint$ and
	 $ \aint[b]$ are not independent (hence
	 $\ptpof[\aint] \cap \ptpof[{\aint[b]}] \neq
	 \emptyset$), or
  \item\label{seq:indep} there is $q'''$ such that
    $q \arro{\aint[b]} q''' \arro{\aint} q''$ %\il{
    (i.e., the transitions are concurrent); %};
    furthermore
	 % , for each
	 % transition $q''' \arro{{\aint[g]}} q''''$ \il{(if any)} we have
	 % that
	 % $\ptpof[{\aint[g]}] \cap (\ptpof[\aint]) \cup
	 % \ptpof[{\aint[b]}])=\emptyset$.
	 for each transition $q''' \arro{\aint[g]} q''''$,
	 $\aint[g] \parallel \aint$ and $\aint[g] \parallel
	 \aint[b]$.  \end{enumerate}
\end{definition}
Intuitively, well-sequencedness forces the explicit representation of
concurrency among interactions with disjoint sets of participants as
commuting diamonds.
The second part of clause \eqref{seq:indep} in \cref{def:wellseq}
rules out the entanglement of choices with commuting diamonds, while
enabling to compose an arbitrary number of independent actions.
This condition, absent in~\cite{BarbaneraLT20}, does not allow them to
enforce faithful realisations as shown in the next example.
\begin{example}
  Consider the c-automaton below.
  \[
    \begin{tikzpicture}[node distance=2cm]
      \tikzstyle{every state}=[cnode]
      \tikzstyle{every edge}=[carrow]
      \node[state, initial, initial text={$\chora=$}] (0) {$q_0$};
      \node[state, right= of 0, yshift=.7cm] (1) {$q_1$};
      \node[state, right= of 0, yshift=-.7cm] (2) {$q_2$};
      \node[state, right= 5cm of 0] (3) {$q_3$};
      \node[state, right= 5cm of 1] (4) {$q_4$};
		\node[state, right= 5cm of 2] (5) {$q_5$};
		\node[state, right= 5cm of 3] (6) {$q_6$};
      \path
      (0) edge node[above] {$\gint[A][B][{\amsg[m]}]$} (1)
      (1) edge node[above, near end] {$\gint[C][D][{\amsg[n]}]$} (3)
      (1) edge node[above] {$\gint[C][B][{\amsg[r]}]$} (4)
      (0) edge node[below] {$\gint[C][D][{\amsg[n]}]$} (2)
      (2) edge node[below, near end] {$\gint[A][B][{\amsg[m]}]$} (3)
      (2) edge node[below] {$\gint[C][B][{\amsg[r]}]$} (5)
      (3) edge node[below, near start] {$\gint[C][B][{\amsg[r]}]$} (6)
      (4) edge node[above] {$\gint[C][D][{\amsg[n]}]$} (6)
      (5) edge node[below] {$\gint[A][B][{\amsg[m]}]$} (6)
      ;
    \end{tikzpicture}
    \]
    In $\proj{\chora}{}$, participant $\ptp[C]$ can immediately send
    $\amsg[r]$ to $\ptp[B]$, since it is not involved in transition $q_0
    \arro{\gint[A][B][{\amsg[m]}]} q_1$. Similarly, $\ptp[B]$ can
    immediately receive $\amsg[r]$ from $\ptp[C]$, since it is not involved in
    transition $q_0 \arro{\gint[C][D][{\amsg[n]}]} q_2$. Thus, a
    transition with label $\gint[C][B][{\amsg[r]}]$ is enabled in the
    initial configuration of the semantics of $\proj{\chora}{}$. However, no transition with the same label is enabled in the initial state of $\chora$, hence the implementation is not faithful.\finex
\end{example}
%}

The following auxiliary concepts are instrumental in the definition of
well-branchedness (cf.\ \cref{def:wb}).
Given a word $\aword$, $\pref[\aword]$ denotes the set of its
prefixes.

\begin{definition}[Full awareness]\label{def:fullaware}
  Let $(\arun_1,\arun_2)$ be a pair of $q$-runs of a c-automaton
  $\chora$.
  Participant $\p \in \ptpof[\arun_1] \cap \ptpof[\arun_2]$ is
  \emph{fully aware} of $(\arun_1,\arun_2)$ if there are
  $\aint_1 \neq \aint_2 \in \lint$ such that
  $\p \in \ptpof[\aint_1] \cap \ptpof[\aint_2]$ and 
  \begin{enumerate}
  \item\label{it:diff} either $\aint_h$ is the first interaction in $\lang[\arun_h]$
	 % \in \aint_h \lint^\star$
    for $h=1,2$
  \item\label{it:fa} or for $h \in \set{1,2}$ there is a proper prefix $\hat\arun_i$ of
	 $\arun_i$ such that $\traceof[{\proj{\hat\arun_1} p}] =
	 \traceof[{\proj{\hat\arun_2} p}]$, the partners of
	 \p in $\aint_h$ are fully aware of $(\hat \arun_1, \hat
	 \arun_2)$, $\traceof[\hat \arun_h] \aint_h \in
	 \pref[{\traceof[\arun_h]}]$,  and $\aint_h$ does not occur on $\arun_{3-h}$.
  \end{enumerate}
\end{definition}
Intuitively, a participant \p is fully aware of two
$q$-runs when able to ascertain which branch has been taken.
This happens either when \p itself chooses \eqref{it:diff}, or when
\p is informed of the choice by interacting with some other
participant already fully aware of the $q$-runs \eqref{it:fa}.
\begin{example}[Full awareness in OLW]\label{ex:fullaware}
  Let us consider the runs $\arun_1=q_2
  \arro{\gint[W][C][{\amsg[loginOK]}]} q_4
  \arro{\gint[w][v][{\amsg[loginOK]}]} q_5$ and $\arun_2=q_2
  \arro{\gint[W][C][{\amsg[loginDenied]}]}
  q_3$ of the OLW c-automaton $\aCM$ in \cref{ex:OLWca}.
  Both $\ptp[w]{}$ and
  $\ptp[c]{}$ are fully-aware of
  $(\arun_1,\arun_2)$ since they occur in the first interaction in
  both the runs (\cref{def:fullaware}\eqref{it:diff}).
  Participant $\ptp[v]{}$ is not fully-aware of
  $(\arun_1,\arun_2)$ since it occurs on $\arun_1$ only.

  Take now the runs $\arun_3=q_6 \arro{\gint[C][W][{\amsg[reject]}]}
  q_8 \arro{\gint[C][V][{\amsg[reject]}]} q_3$ and $\arun_4=q_6
  \arro{\gint[C][W][{\amsg[authorise]}]} q_7
  \arro{\gint[C][V][{\amsg[pay]}]}
  q_3$ in $\aCM$.
  As before, both participants $\ptp[w]{}$ and
  $\ptp[c]{}$ are fully-aware of
  $(\arun_3,\arun_4)$ since they occur in the first interaction in
  both the runs. Participant
  $\ptp[v]{}$ is fully-aware of
  $(\arun_3,\arun_4)$ as well, since its partner
  $\ptp[c]{}$ is fully-aware of $(q_6
  \arro{\gint[C][W][{\amsg[reject]}]} q_8,q_6
  \arro{\gint[C][W][{\amsg[authorise]}]} q_7)$.\finex
\end{example}

To establish well-branchedness of a c-automaton %\il{
  we have to ensure
  that for each choice, namely for each state with (at least) two
  non-independent outgoing transitions, and each participant \p, if \p
  has to take different actions in the branches starting from the two
  transitions, then \p is fully-aware of the taken branch. In
  principle, such a condition should be checked on all pairs of
  coinitial paths. However, this would lead to redundant checks, hence
  below we borrow from~\cite{BarbaneraLT20} the notion of $q$-spans,
  namely pairs of paths from $q$ on which we will perform the
  check. Essentially, %}
  we have to handle choices with loops on some
branches and we have to consider \quo{long-enough} branches.
  More precisely, a $q$-run in a c-automaton $\chora$ is a
  \emph{pre-candidate $q$-branch} if each of its cycles has at
  most one occurrence within the whole run (i.e., if $\arun'$ is a
  $q'$-run included in $\arun$ and ending in $q'$, then $\arun'$ has
  exactly one occurrence in $\arun$); a \emph{candidate $q$-branch} is
  a maximal pre-candidate $q$-branch with respect to the prefix order.
\begin{definition}[$q$-span]\label{def:span}
  A pair $(\arun,\arun')$ of pre-candidate $q$-branches of $\chora$ is
  a \emph{$q$-span} if % $\arun$ and $\arun'$ are
  \begin{enumerate}
  \item\label{span:cofinal} %\iLcomm[if we replace it with ``a
									 %candidate q-branch and a pre-candidate
									 %q-branch which share only the first and
									 %the last node of the latter'' we should
									 %be able to drop the last item in the
									 %def. of well-branchedness]{}
	 either $\arun$ and $\arun'$ are cofinal, with no common node but $q$ and the last one;
  \item\label{span:maximal} or $\arun$ and $\arun'$ are candidate $q$-branches with no common
	 node but $q$;
  \item\label{span:loop} or $\arun$ and $\arun'$ are a candidate $q$-branch and a loop on $q$
	 with no other common nodes.
  \end{enumerate}
\end{definition}

We can now introduce well-branchedness.

\begin{definition}[Well-branchedness]\label{def:wb}
  A c-automaton $\chora$ is \emph{well-branched} if it is
  deterministic and for each of its states $q$
  %with some outgoing
  %transitions ($\transset \chora q \neq \emptyset$)
  there is a
  partition $T_1,\dots,T_k$ of $\transset{\chora}{q}$ such that
  \begin{itemize}
  \item for all $1 \leq i \neq j \leq k$,
	 $\ptpof[T_i] \cap \ptpof[T_j] = \emptyset$ and for each
	 $q \arro{\aint_i} q_i \in T_i$, $q \arro{\aint_j} q_j \in T_j$
	 there exists $q'$ such that $q_i \arro{\aint_j} q'$ and
	 $q_j \arro{\aint_i} q'$
  \item for all $1 \leq i \leq k$,
	 $\bigcap_{t \in T_i} \ptpof[t] \neq \emptyset$ and for all
	 $\p \in \ptpof[\chora] \setminus \bigcap_{t \in T_i} \ptpof[t]$
	 and $q$-span $(\arun_1,\arun_2)$ starting from transitions in
	 $T_i$, if $\proj{\arun_1}{p} \neq \proj{\arun_2}{p}$ then either
	 \p is fully aware of $(\arun_1,\arun_2)$ or there is
	 $i \in \set{1,2}$ such that $\p \not\in \ptpof[\arun_i]$ and
	 \begin{enumerate}
	 \item\label{wb:aware} the first transition in $\arun_{3-i}$
		involving \p is with a fully aware participant of
		$(\arun_1,\arun_2)$ and
	 \item\label{wb:continuations} for all runs $\arun'$ such that $\arun_i\arun'$ is a
		candidate $q$-branch of $\chora$ the first transition in
		$\arun'$ involving \p is with a participant which is fully aware
		of $(\arun_1,\arun_2)$.
	 \end{enumerate}
  \end{itemize}
\end{definition}
Intuitively, a c-automaton is well-branched if for any state with
multiple
outgoing
transitions (both clauses in \cref{def:wb} trivially
hold when $\transset{\chora}q$ is empty or a singleton)
, we can group them in equivalence classes.
Transitions in different classes are concurrent, hence they give rise
to commuting diamonds.
Transitions in the same class are choices: one participant, belonging
to all the (initial) transitions, makes the choice, and any other
participant \p is either fully aware of the $q$-runs or
% \il{
it is inactive in some branch $\arun_i$ (condition
$\p \not\in \ptpof[\arun_i]$). In the last case, \p has to interact
with a fully aware partner (i) on each continuation $\arun'$ (if any)
of $\arun_i$ as well as (ii) inside the other branch, $\arun_{3-i}$.
% }
Intuitively, (i) is necessary to make \p aware of when the
choice is fully completed and (ii) on whether the branch on which \p needs
to act has been taken.
%
%Note that \cref{def:wb} allows a participant to take part in a choice
%even when it is not fully aware of some branches (provided that it can
%understand whether any of the branches it acts on has been taken or
%not).
%
%\il{
  At the price of increasing the technical complexity, the second
  clause in \cref{def:wb} can be relaxed. Indeed, right now it
  requires a participant $\p$, occurring in one branch only, to
  interact (both in the branch where it occurs and in the
  continuations after the merge of the two branches) with a
  fully-aware participant. We could instead allow $\p$ to interact
  with a chain of other participants occurring only in the same
  branch, and such that the last participant in the chain interacts
  with a fully-aware participant.  %}
%
%We opted for simplicity instead
%of generality here.

\begin{example}[OLW is well-branched]
  Let us show that the c-automaton in \cref{ex:OLWca} is
  well-branched.
  The only states for which well-branchedness is not trivial are $q_2$
  and $q_6$ (the others have at most one outgoing transition).
  %\il{
	   %
	 In both the cases we have a single equivalence class where
	 $\ptp[w]$ and $\ptp[c]$ are in all the first transitions; hence
	 they are both fully-aware in all the possible spans.
  Let us check the condition for $\ptp[v]$.
	 Let us consider $q_6$. There is one $q_6$-span, with branches with
	 states $q_6,q_8,q_3$ and $q_6,q_7,q_3$, which fits case
	 \ref{span:cofinal} in \cref{def:span}.
  %}.
  %
  As discussed in \cref{ex:fullaware}, in this $q_6$-span $\ptp[v]{}$
  is fully-aware, hence the condition is satisfied. %\il{
    Let us now
	 consider $q_2$. Here we have a loop with states $q_2$, $q_0$,
	 $q_1$, $q_2$, a candidate $q_2$-branch with states $q_2,q_3$, and
	 two candidate $q_2$-branches with a common prefix (states
	 $q_2,q_4,q_5,q_6$) and two continuations (states $q_6,q_8,q_3$ and
	 $q_6,q_7,q_3$). Any combination of the self-loop with the
	 candidate $q_2$-branches fit in case \ref{span:loop} in
	 \cref{def:span}, while the pairings of the first candidate
	 $q_2$-branch with any of the others fit in case \ref{span:cofinal}
	 in \cref{def:span}. %}
           In the $q_2$-spans above $\ptp[v]$ occurs
           only in the one %\il{
             with two continuations. %}
%  starting with
%  $q_2 \arro{\gint[W][C][{\amsg[loginOK]}]} q_4
%  \arro{\gint[w][v][{\amsg[loginOK]}]} q_5$ (which has two possible
%  continuations), and
  Since there it interacts with $\ptp[c]{}$ which is
  fully-aware, %\il{
    condition \ref{wb:aware} in \cref{def:wb} holds. Condition \ref{wb:continuations} holds trivially,
    since %}
    the branches join only in state $q_3$ which has no outgoing
  transitions.\finex
\end{example}

%\il{
  \begin{example}[Non well-branched c-automata]\label{ex:nonwb}
    Consider the c-automaton below.
    \[
    \begin{tikzpicture}[node distance=3cm]
      \tikzstyle{every state}=[cnode]
      \tikzstyle{every edge}=[carrow]
      \node[state, initial, initial text={$\chora=$}] (0) {$q_0$};
      \node[state, above right of=0, yshift=-1.8cm] (1) {$q_1$};
      \node[state, right of=1] (2) {$q_2$};
      \node[state, right of=2] (3) {$q_3$};
      \node[state, below right of=0, yshift=1.8cm] (4) {$q_4$};
		\node[state, right of=4] (5) {$q_5$};
		\node[state, right of=5] (6) {$q_6$};
      \path
      (0) edge node[above] {$\gint[A][B][{\amsg[l]}]$} (1)
      (1) edge node[above] {$\gint[B][C][{\amsg[n]}]$} (2)
      (2) edge node[above] {$\gint[C][D][{\amsg[l]}]$} (3)
      (0) edge node[below] {$\gint[A][B][{\amsg[r]}]$} (4)
      (4) edge node[below] {$\gint[B][C][{\amsg[n]}]$} (5)
      (5) edge node[below] {$\gint[C][D][{\amsg[r]}]$} (6)
      ;
    \end{tikzpicture}
    \]
    Here, $\ptp[C]$ is not fully-aware since it interacts with
	 $\ptp[B]$ (which is fully-aware) receiving the same message on
	 both the branches.
	 Hence, its first different interactions are with $\ptp[D]$, which
	 is not fully-aware. Indeed, $\ptp[D]$ gets different messages, but
	 from $\ptp[C]$ which is not fully aware either.  Thus, $\ptp[C]$
	 and $\ptp[D]$ can decide, e.g., to take the lower branch even if
	 $\ptp[A]$ and $\ptp[B]$ took the upper one, thus producing a trace
	 $\gint[A][B][{\amsg[l]}] \cdot \gint[B][C][{\amsg[n]}] \cdot
	 \gint[C][D][{\amsg[r]}]$ not part of the language of $\chora$.
	 \finex
  \end{example}
%}

%\il{
  \begin{example}[Non well-branchedness with selective participation]
    Consider the c-automaton:
    \[
    \begin{tikzpicture}[node distance=3cm]
      \tikzstyle{every state}=[cnode]
      \tikzstyle{every edge}=[carrow]
      \node[state, initial, initial text={$\chora=$}] (0) {$q_0$};
      \node[state, below right of=0, yshift = 1.5cm] (1) {$q_1$};
      \node[state, above right of=1, yshift = -1.5cm] (2) {$q_2$};
      \node[state, right of=2] (3) {$q_3$};
      \path
      (0) edge node[above] {$\gint[A][C][{\amsg[n]}]$} (2)
      (0) edge node[below] {$\gint[A][B][{\amsg[n]}]$} (1)
      (1) edge node[below] {$\gint[B][C][{\amsg[n]}]$} (2)
      (2) edge node[above] {$\gint[B][D][{\amsg[n]}]$} (3)
      ;
    \end{tikzpicture}
    \]
    Here, $\ptp[B]$ occurs in the bottom branch only, interacting with
	 $\ptp[A]$ which is fully-aware, as required. However, after the
	 merge of the two branches, $\ptp[B]$ interacts with $\ptp[D]$
	 which is not fully aware, thus violating condition
	 \ref{wb:continuations} in \cref{def:wb}. Indeed the interaction
	 $\gint[B][D][{\amsg[n]}]$ is enabled since the initial
	 configuration, against the prescription of $\chora$.\finex
  \end{example}
%}

\begin{definition}[Well-formedness]\label{def:wf}
  A c-automaton $\chora$ is \emph{well-formed} if it is both
  well-sequenced and well-branched.
\end{definition}

Well-formed c-automata enjoy relevant
properties.
First, for each well-formed c-automaton the semantics of
the projection is bisimilar to the starting c-automaton.
\iftr
\begin{lemma}\label{lem:intfromruns}
  Let $\chora$ be a well-formed c-automaton and $A_{\p}$ and $A_{\q}$
  be the intermediate CFSM for two participants $\p$ and $\q$ of
  $\chora$ (cf. \cref{def:projection}).  If
  \begin{align*}
  q {\arro \emptyword}^\star q_{\p} \arro{\aout}{q'_{\p}} \text{ in } A_{\p}
  \qqand q {\arro \emptyword}^\star q_{\q} \arro{\ain}{q'_{\q}} \text{ in } A_{\q}
  \end{align*}
  then the state $q$ of $\chora$ has an outgoing transition with label
  $\gint$.
\end{lemma}
\begin{proof}
  Let the two runs for \p and \q in the intermediate CFSMs be the
  projections of runs $\arun_\p$ and $\arun_\q$ in $\chora$.
  We have two cases, depending on whether $\arun_\p = \arun_\q$.

  If $\arun_\p = \arun_\q$ then the last transition of $\arun_\p$
  is concurrent to all the previous ones.
  Indeed, the previous transitions are projected to $\emptyword$ both on
  $\p$ and on $\q$, hence neither $\p$ nor $\q$ can occur in the
  label.
  Thus, by well-sequencedness, the state $q$ has a transition with
  label $\gint$.

  Assume $\arun_\p \neq \arun_\q$.
  By construction,
  \[
	 \arun_\p = \arun \ t_\p \ \arun'_\p \ q_\p \arro\gint q'_\p
	 \qqand
	 \arun_\q = \arun \ t_\q \ \arun'_\q \ q_\q \arro\gint q'_\q
  \]
  for a $q$-run $\arun$, two transitions $t_\p$ and $t_\q$, two runs
  $\arun'_\p$ and $\arun'_\q$; observe that the ending state of
  $\arun$, say $\hat q$, is also the source state of $t_\p$ and $t_\q$
  while $\arun'_\p$ (resp.\ $\arun'_\q$) ends in $q_\p$ (resp.\ $q_\q$).

  Note that $\p \not\in \ptpof[\arun \ t_\p \ \arun'_\p]$ and
  $\q \not\in \ptpof[\arun \ t_\q \ \arun'_\q]$.
  If $\q \not\in \ptpof[\arun \ t_\p \ \arun'_\p]$ or
  $\p \not\in \ptpof[\arun \ t_\q \ \arun'_\q]$ then the thesis
  immediately follows by well-sequencedness as before.
  Therefore we can assume $\p \in \ptpof[\arun \ t_\q \ \arun'_\q]$
  ($\q \in \ptpof[\arun \ t_\p \ \arun'_\p]$ could be assumed as
  well).

  Consider now the two $\hat q$-runs $\arun_1 = t_\p \ \arun'_\p$ and
  $\arun_2 = t_\q \ \arun'_\q$.
  By well-branchedness, there is a partition of
  $\transset\chora{\hat q}$ satisfying the conditions of
  \cref{def:wb}.
  Then $t_\p$ and $t_\q$ cannot belong to the same equivalence class
  of such partition since neither \p nor \q are fully aware of
  $(\arun_1,\arun_2)$ and the first interaction of \p on $\arun_1$
  is with \q.
  Hence, $t_\p = \hat q \arro{\aint_1} q_1$ and
  $t_\q = \hat q \arro{\aint_2} q_2$ necessarily belong to different
  equivalence classes.
  Therefore, by \cref{def:wb}, there is a state $\hat q'$ such that
  $q_1 \arro{\aint_2} \hat q'$ and $q_2 \arro{\aint_1} \hat q'$.
  Hence, the transitions of $\arun'_\p$ and those of $\arun'_\q$ form
  commuting diamonds and therefore there is a $q$-run in $\chora$
  where \q does not occur and all the transitions involving \p in
  $\arun'_\q$ follow a transition with label $\gint$ (easily by
  induction on the length of $\arun'_\p$ and $\arun'_\q$).
  The thesis then follows since, as before, the transition labelled by
  $\gint$ commutes with any preceding transition by
  well-sequencedness.
\end{proof}

\begin{lemma}\label{lemma:diffclass}
  % Let $\aint_1$ and $\aword_{\aint[b]}= \aint[b]_1 \dots \aint[b]_m$
  % be words accepted from a $q$-span of a deterministic well-formed
  % c-automaton such that $\aint_1$ and $\aint[b]_1$ are
  % independent. Then either $\aint_1 = \aint[b]_i$ for some $i$ and it
  % commutes with all $\aint[b]_j$ for $j \in \set{1,\dots,i}$ or the
  % same holds for each continuation.
  %
  %
  Let $\chora$ be a well-formed c-automaton and, for
  $i \in \set{1,2}$, $t_i = q \arro{\aint_i} q_i'$ the first two
  transitions of a $q$-span $(\arun_1,\arun_2)$ in $\chora$ such that
  $t_1$ and $t_2$ are concurrent.
  Then either $\aint_1$ occurs on $\traceof[\arun_2]$ or for each run
  $\arun_2 \ \arun$ in $\chora$, $\aint_1$ occurs in
  $\traceof[\arun]$.
\end{lemma}
\begin{proof}
  % The proof is by coinduction.  By well-branchedness there is a
  % commuting diamond with sides $\aint_1 \aint[b]_1$ and
  % $\aint[b]_1 \aint_1$. If $\aint_1=\aint[b]_2$ we are
  % done. Otherwise, by well-sequencedness, we have that $\aint_1$ and
  % $\aint[b]_2$ are independent and the thesis follows
  % by coinductive hypothesis.
  %
  The proof is by coinduction.
  By well-branchedness there is a state $q'$ such that
  \[\begin{tikzpicture}[node distance=.5cm and 1cm]
		\tikzstyle{every state}=[cnode]
		\tikzstyle{every edge}=[carrow]
    	\node[state] (q) {$q$};
		\node[state, above right = of q] (q1) {$q_1$};
		\node[state, below right = of q] (q2) {$q_2$};
		\node[state, below right = of q1] (q') {$q'$};
		\path (q) edge node[above] {$\aint_2$} (q1);
		\path (q) edge node[below] {$\aint_2$} (q2);
		\path (q2) edge node[below] {$\aint_1$} (q');
		\path (q1) edge node[above] {$\aint_2$} (q');
	 \end{tikzpicture}
  \]
  If $q'$ occurs in $\arun_2$ we are done.
  Otherwise, by well-sequencedness
  (cf. \cref{def:wellseq}\eqref{seq:indep}), $q_1 \arro{\aint_2} q'$
  is concurrent with the first transition of $\arun_2$ and the thesis
  follows by the coinductive hypothesis.
\end{proof}

\begin{proposition}\label{lemma:cui}
  Let $\chora$ be a well-formed c-automaton,
  $\p,\q \in \ptpset_{\chora}$, and $\aint \in \lint$.
  If $\aword_{\p} \aint,\ \aword_{\q} \aint,\ \aword \in \lang[\chora]$
  are three words such that $\proj{\aword_{\p}}{p}=\proj{\aword}{p}$
  and $\proj{\aword_{\q}}{q}=\proj{\aword}{q}$ then
  $\aword \aint \in \lang[\chora]$.
\end{proposition}
\begin{proof}
  The proof is by case analysis on the form of the words.

  If $\aword = \aword_\p = \aword_\q = \emptyword$ then the thesis
  follows trivially.
  Otherwise at least one among those words is not empty.
  Let $\aint = \gint$ and, towards a contradiction, suppose
  $\aword \aint \notin \lang[\chora]$.

  Let $\arun$, $\arun_\p$, and $\arun_\q$ be runs of $\chora$ such
  that
  \begin{align}\label{eq:runs}
	 \aword = \traceof[\arun],
	 \qquad
	 \aword_\p\ \aint = \traceof[\arun_\p],
	 \qqand
	 \aword_\q\ \aint = \traceof[\arun_\q]
  \end{align}
  respectively; and let $\hat q$ be the state from where at least two
  of the three runs in~\eqref{eq:runs} start to become different.

  We first consider the case where one of the words is empty.
  \begin{itemize}
  \item If $\aword = \emptyword$, we show that $\gint$ is enabled in
	 the initial state contradicting our assumption that
	 $\aword\ \aint = \gint \not\in \lang[\chora]$.
	 From the hypothesis, $\p$ does not occur in $\arun_\p$ and $\q$
	 does not occur in $\arun_\q$.
	 If $\p$ also does not occur in $\arun_\q$ then we can commute
	 $\gint$ and the thesis follows.
	 Otherwise, either $\p$ or $\q$ should be fully aware, but this is
	 not possible since each of them occurs in one of the runs only.
  \item If $\aword_\p = \emptyword$ then there is a transition
	 labelled with the interaction $\aint = \gint$ from the initial
	 state of $\chora$, say $q_0$.
	 Such transition cannot be concurrent with the first transition of
	 $\arun$ (otherwise, by \cref{lemma:diffclass}, it would occur on
	 all continuations of $\arun$ contradicting our assumption that
	 $\aword\ \aint \not\in \lang[\chora]$).
    Hence, by well-branchedness, the initial transitions of $\arun$
	 and $\arun_\p$ must belong to a same partition of
	 $\transset{\chora}{q_0}$.
	 Again, by well-branchedness, \q must be fully aware of
	 $(\arun_\p,\arun)$.
	 The only possibility is that \q occurs in the initial transitions
	 of $\arun$ and $\arun_\q$.
	 The condition $\proj \aword q = \proj{\aword_\q} q$ implies that
	 $\arun$ and $\arun_\q$ must have the same initial transition and
	 therefore they share a non-empty prefix; let $\hat q'$ be the last
	 state of the longest of such prefixes and $t \neq t_\q$ be the first
	 transitions from $\hat q'$ on $\arun$ and $\arun_q$ respectively.
	 \begin{itemize}
	 \item If $t$ and $t_\q$ are concurrent then by well-branchedness
		the last transition of $\arun_\q$, say $t_\p$ (note that the label of $t_\p$
		is $\aint = \gint$) must be
		concurrent with the last transition of $\arun$, hence there is a
		run $\arun\ t_\p$ in $\chora$, which violates our assumption
		that $\aword\ \aint \in \lang[\chora]$.
	 \item Otherwise \q is not fully aware of the choice because
		otherwise $t = t_\q$.
		Hence by well-branchedness \p occurs on $\arun_\p$ contrary to
		the assumption that $\proj \aword p = \emptyword$.
	 \end{itemize}
	 In all cases we derive a contradiction, hence
	 $\aword \ \aint \in \lang[\chora]$.
  \item If $\aword_\q$ is empty then the proof is as in the previous
	 case.
  \end{itemize}

  For the case that none of the words is empty (i.e.,
  $\aword_\p \neq \emptyword$, $\aword_\q \neq \emptyword$, and
  $\aword \neq \emptyword$) we analyse how the runs branch.

  Let $\hat q$ be the first state on $\arun$ after which $\arun$ and
  $\arun_\p$ start to diverge along two different transitions $t$ of
  $\arun$ and $t_\p$ of $\arun_\p$.

  Then $t$ and $t_\p$ must be in the same equivalence class of the
  partition of $\transset \chora {\hat q}$ otherwise the last
  transition of $\arun_\p$, say $t'$ (which is labelled with $\aint$)
  would commute with all the transitions of $\arun$; hence
  $\arun \ t'$ would be a run in $\chora$ contrary to our assumption
  that $\aword\ \aint \not\in \lang[\chora]$.
  (Such partition should also include a transition $t_\q$ on
  $\arun_\q$ for the same reason.)

  The projection on \q of $\arun$ and $\arun_\q$ differ.
  Assume the number of interactions involving \p in $\aword$, which is
  the same as those in $\aword_\p$, is not $0$.
  Then, \p should be fully aware of $(\arun,\arun_\p)$.
  However, if $t$ and $t_\p$ involve \p then $t = t_\p$ by the
  hypothesis that $\proj \aword p = \proj{\aword_\p} p$. Hence
  condition~\eqref{it:diff} of \cref{def:fullaware} does not apply.
  Condition~\eqref{it:fa} of \cref{def:fullaware} does not apply
  either since there is an action on one side only after equal traces.
  Hence we have a contradiction if the number of interactions involving \p is not $0$.

  Now, let the number of interactions involving \p be $0$.
  Hence, \p should be fully aware of $(\arun, \arun_\p)$.
  Using the same reasoning on $(\aword,\aword_{\p})$, however, the
  number of occurrences of \p in $\aword$ is $0$, hence it cannot be
  fully aware of $(\aword, \aword_\q\ \gint)$.  Again we have a
  contradiction, hence this case can never happen.
  % \hsl
  % If we have at least two equivalence classes, then the initial
  % transitions of one of the runs must be in an equivalence class
  % different from the equivalence class of the others.
  % %
  % Hence, the initial transition is concurrent with the initial
  % transitions of both the other runs.
  % %
  % Thanks to \cref{lemma:diffclass} then it occurs in the run or in a
  % possible continuation.
  % %
  % If none of the occurrences of the transition in the other traces is
  % the final occurrence of $\gint$, then we can anticipate it on all
  % the runs, and remove it from all the runs obtaining shorter
  % runs. Hence the thesis follows by inductive hypothesis. If an
  % occurrence is the final $\gint$, since it is the first occurrence of
  % $\gint$ on the path, then there are no previous
  % occurrences. However, since occurrences of $\gint$ are the same in
  % all the runs due to the condition of projection in the hypothesis
  % then we have a contradiction, hence this case can never happen.
\end{proof}
\fi

\begin{theorem}\label{prop:bisim}
  $\chora$ is bisimilar to $\ssem{\proj{\chora}{}}$ for any
  well-formed c-automaton $\chora$.
\end{theorem}
\iftr
\begin{proof}
  Let $\chora=\conf{Q, q_0,\lint,\tset}$ and let
	 $\sset$ be the set of configurations of
	 $\ssem{\proj{\chora}{}}$.
  We show by coinduction that  the relation
  \begin{align*}
	 \rcal = \set{(q,\confel) \in Q \times \sset \sst q \in
	 \confel(\p) \text{ for each } \p \in \ptpset_{\chora}}
  \end{align*}
  is a bisimulation.
  (Recall that, due to determinisation and minimisation, for each
  $\confel \in \sset$ and each participant $\p$, $\confel(\p)$ is a
  subset of $Q$).
  Since bisimulation implies trace equivalence, we also have that
  corresponding elements are reachable via the same trace.

  Let $(q,\confel) \in \rcal$ and consider a challenge from
  $\ssem{\proj{\chora}{}}$, namely $\confel \arro \gint \confel'$.
  By definition of synchronous semantics,
  $\confel(\p) \arro \aout \confel'(\p)$,
  $\confel(\q) \arro \ain \confel'(\q)$ and
  $\confel(\ptp[x]) = \confel'(\ptp[x])$ for each
  $\ptp[x] \notin \set{\p,\q}$.
  By definition of determinisation, there are
  $\tilde q_\p \in \confel(\p)$ and $\tilde q_\q \in \confel(\q)$ such
  that
  \[
	 \tilde q_\p {\arro \emptyword}^\star q_\p \arro{\aout}{q'_{\p}}
	 \text{ in } A_{\p}
	 \qqand
	 \tilde q_\q {\arro \emptyword}^\star q_\q \arro{\ain}{q'_{\q}}
	 \text{ in } A_\q
  \]
  If $q = \tilde q_\p = \tilde q_\q $ then $q \arro \gint q'$ by
  \cref{lem:intfromruns}.
  Otherwise, since $(q,\confel) \in \rcal$, each run in
  $\ssem{\proj \chora {}}$ that reaches $\confel$ has, for all
  participants $\ptp[x] \in \ptpset$, a corresponding run that reaches
  $q$ in $A_{\ptp[x]}$.
  Now consider a word $w$ that reaches $q$ and matching words $w_\p$
  and $w_\q$ obtained by lifting to $\chora$ runs reaching
  $\tilde q_\p$ and $\tilde q_\q$ in the respective auxiliary
  CFSMs. By construction, they are in the hypothesis of
  \cref{lemma:cui}, hence $q \arro\gint q'$ also in this case.
  We now show that $q' \in \confel'(\ptp[x])$ for each $\ptp[x]$.
  By definition, for each $\ptp[x]$, in the intermediate CFSM
  $A_{\ptp[X]}$ we have $q_{\ptp[x]} \arro{\proj \gint x} q'_{\ptp[x]}$ for some
  $q_{\ptp[x]} \in \confel(\ptp[x])$ and some $q'_{\ptp[x]}$.
  \begin{itemize}
  \item If $\ptp[x] \notin \set{\p,\q}$ then
	 $\proj \gint x = \emptyword$ hence
	 $q_{\ptp[x]} \in \confel(\ptp[x])$ implies
	 $q'_{\ptp[x]} \in \confel'(\ptp[x])$ as required since
		$\confel(\ptp[x])$ contains the $\varepsilon$-closure of its
		elements by construction.
  \item If $\ptp[x] = \p$ then $\proj \gint x = \aout$.
	 Thus $q_{\ptp[x]} \arro \aout q'_{\ptp[x]}$ and, as shown above,
	 $\confel(\p) \arro \aout \confel'(\p)$.
	 Since $q_{\ptp[x]} \in \confel(\p)$ and $\chora$ is deterministic
	 then $q'_{\ptp[x]} \in \confel'(\ptp[x])$
	 % \iLcomm[not sure this is relevant, should come from determinism, but I have to think more about this]{}
	 since $\confel'(\ptp[x])$ is the
		$\varepsilon$-closure of $q'_{\ptp[x]}$ by construction.
  \item If $\ptp[x] = \q$ then the reasoning is analogous to the previous case.
  \end{itemize}

  Let us now consider a challenge from $\chora$, namely
  $q \arro \gint q'$.
  By definition of projection and $\varepsilon$-closure,
  $\confel(\p) \arro \aout \confel'(\p)$,
  $\confel(\q) \arro \ain \confel'(\q)$, and
  $\confel(\ptp[x]) = \confel'(\ptp[x])$ for each
  $\ptp[x] \notin \set{\p,\q}$.
  By definition of synchronous semantics $\confel \arro \gint
  \confel'$ as desired. For each participant $\ptp[y]{}$ from $q \in \confel(\ptp[y])$
  we get $q' \in \confel'(\ptp[y])$, hence the thesis follows.
\end{proof}
\fi

An immediate consequence of \cref{prop:bisim} is that the language of
a well-formed c-automaton coincides with the language of the semantics of its
projection.

\begin{corollary}\label{th:projectionCorrectness}
  $\lang[\chora] = \lang[\ssem{\proj{\chora}{}}]$  for any well-formed c-automaton $\chora$.
\end{corollary}
\iftr
\begin{proof}
  From \cref{prop:bisim} given that bisimulation implies trace equivalence.
\end{proof}
\fi

We now show that projections of well-formed c-automata do not
deadlock.  To this end, we need to extend CFSMs with a concept of
final state.  Intuitively, a state is final in the projection on some
participant $\p$ of a given c-automaton $\chora$ iff one of the
corresponding states of $\chora$ (remember that states of the
projection are sets of states of $\chora$) has an outgoing maximal
path along with $\p$ is not involved. Formally:

\begin{definition}[Final states in projected CFSMs]\label{def:final}
  % Let $\chora$ be a c-automaton and $\p$ one of its participants.
  % We extend the CFSM
  % $\proj{\chora}{\p}=\autom{{\bf Q}}{Q_0}{\lset}{\tset}$ with a set
  % of final states defined as follows:
  % $\finals=\set{Q \in {\bf Q} \,|\, \exists q \in Q, \pi \textrm{
  % maximal path from } q \textrm{ in } \chora. \p \notin \ptpof[\pi]
  % }$ \hsl
  Let $\chora$ be a c-automaton and \p one of its participants.
  A state $Q$ of $\proj \chora p$ is \emph{final} if in $\chora$ there
  is $q \in Q$ and a candidate $q$-branch such that
  $\p \not\in \ptpof[\arun]$.
\end{definition}

\begin{definition}[Deadlock freedom]\label{def:DF}
  The projection of a c-automaton is \emph{deadlock-free} if for each
  of its reachable configurations $\confel$ either $\confel$ has an
  outgoing transition or, for each participant $\p$, $\confel(\p)$ is
  final.
\end{definition}

\begin{theorem}[Projections of well-formed c-automata are deadlock-free]\label{thm:df}
  Let $\chora$ be a well-formed c-automaton.
  Then $\proj{\chora}{}$ is deadlock-free.
\end{theorem}
\iftr
\begin{proof}
  Let us assume, towards a contradiction, that $\proj{\chora}{}$ is
  not deadlock-free.
  Then there is a reachable configuration $\confel$ in
  $\ssem{\proj{\chora}{}}$ with no outgoing transition and there
  exists a participant $\p$ such that $\confel(\p)$ is not final.
  Then, by \cref{def:final}, for each $q \in
	 \confel(\p)$ and each candidate $q$-branch
	 $\arun$ in $\chora$, $\p \in \ptpof[\pi]$.
  %
  % Consider the bisimulation
  % $\rcal = \set{(q,\confel) \in Q \times \sset \sst q \in \confel(\p)
  % 	 \text{ for each } \p \in \ptpset}$ from the proof of
  % \cref{prop:bisim}.  Select the $\tilde q$ such that
  % $(\tilde q,\confel) \in \rcal$.
  % Thanks to the bisimulation,
  %
  From the proof of \cref{prop:bisim},
  $\confel$ is bisimilar to one such $q$; hence, $\confel$ should answer
  the challenge from the first action of $\arun$, hence it has an
  outgoing transition against the hypothesis.
\end{proof}
\fi

\begin{example}[C-automaton with deadlock]\label{ex:deadlock}
	Consider the c-automaton
    \[
    \begin{tikzpicture}[node distance=3cm]
      \tikzstyle{every state}=[cnode]
      \tikzstyle{every edge}=[carrow]
      \node[state, initial, initial text={$\chora=$}] (0) {$q_0$};
      \node[state, above right of=0, yshift = -1.8cm] (1) {$q_1$};
      \node[state, right of=1] (2) {$q_2$};
      \node[state, right of=2] (3) {$q_3$};
      \node[state, right of=3] (4) {$q_4$};
      \node[state, below right of=0, yshift = 1.8cm] (5) {$q_5$};
      \node[state, right of=5] (6) {$q_6$};
      \node[state, right of=6] (7) {$q_7$};
      \node[state, right of=7] (8) {$q_8$};

      %\node[state, above of=4] (9) {$q_9$};
      %\node[state, above of=8] (10) {$q_{10}$};
      %
      \path
      (0) edge node[above] {$\gint[A][B][{\amsg[l]}]$} (1)
      (1) edge node[above] {$\gint[B][C][{\amsg[n]}]$} (2)
      (2) edge node[above] {$\gint[C][D][{\amsg[l]}]$} (3)
      (3) edge node[above] {$\gint[C][A][{\amsg[l]}]$} (4)
      (0) edge node[below] {$\gint[A][B][{\amsg[r]}]$} (5)
      (5) edge node[below] {$\gint[B][C][{\amsg[n]}]$} (6)
      (6) edge node[below] {$\gint[C][D][{\amsg[r]}]$} (7)
      (7) edge node[below] {$\gint[C][A][{\amsg[r]}]$} (8)

      (3) edge[dashed,loop above] node[above] {$\gint[C][D][{\amsg[l]}]$} (3)
      (7) edge[dashed,loop below] node[below] {$\gint[C][D][{\amsg[r]}]$} (7)
      ;
    \end{tikzpicture}
    \]
	 % below,
	 obtained by adding the transitions from states $q_3$ and $q_7$ to the one in \cref{ex:nonwb}.
	 Disregard the dashed transitions.
	 If, as discussed in \cref{ex:nonwb}, $\ptp[C]$ and $\ptp[D]$
	 decide to take the bottommost branch while $\ptp[A]$ and $\ptp[B]$ take
	 the uppermost one, we can reach a configuration $\confel$ where $\ptp[C]$
	 wants to send $\msg[r]$ to $\ptp[A]$, but $\ptp[A]$ is only
	 willing to take $\msg[l]$. Hence, no transition is possible and we
	 have a deadlock.  Due to \cref{thm:df} this is possible only since
	 the c-automaton is not well-formed.\finex
 \end{example}
%}

We can refine the result above by focusing on a single participant.

\begin{definition}[Lock freedom]\label{def:LF}
  The projection of a c-automaton is \emph{lock-free} if for each of
  its reachable configurations $\confel$ and each participant $\p$,
  either $\confel(\p)$ is final or $\confel$ %\il{
  has at least an outgoing transition and %}
  for each candidate $\confel$-branch
  $\arun$ we have $\p \in \ptpof[\arun]$.
\end{definition}
%\il{
  Lock freedom is strictly stronger than deadlock freedom. Indeed, each configuration $\confel$ and a participant $\p$ such that $\confel(\p)$ is not final has an outgoing transition, hence it is not a deadlock. However, there are systems which are deadlock-free but not lock-free, as discussed below.
  \begin{example}[C-automaton with locks (but no deadlock)]
    Consider again the c-automaton from \cref{ex:deadlock}, 
	 including the dashed self-loops.  There is now no deadlock, since
	 the configuration $\confel$ has an outgoing transition, namely a
	 self-loop involving $\ptp[C]$ and $\ptp[D]$.  However, $\confel$
	 is a lock for $\ptp[A]$. Indeed, it is not final for $\ptp[A]$,
	 yet $\ptp[A]$ does not take part in the branch corresponding to
	 the execution of the self-loop.
  \end{example}
%}

\begin{theorem}[Projections of well-formed c-automata are lock-free]\label{thm:lf}
  Let $\chora$ be a well-formed c-automaton.
  Then $\proj{\chora}{}$ is lock-free.
\end{theorem}
\iftr
\begin{proof}
  Let us assume, towards a contradiction, that $\proj{\chora}{}$ is
  not lock-free.
  Then there is a reachable configuration $\confel$ in
  $\ssem{\proj{\chora}{}}$ and a participant $\p$ such that
  $\confel(\p)$ is not final and either there is no outgoing transition or there is a candidate $\confel$-branch $\arun$
  with $\p \notin \ptpof[\arun]$. In the first case the configuration is a deadlock and we have a contradiction from \cref{thm:df}.

  Otherwise, by \cref{def:final}, for each $q \in
	 \confel(\p)$ and each candidate $q$-branch
	 $\arun$ in $\chora$, $\p \in \ptpof[\pi]$.

  From the proof of \cref{prop:bisim},
  $\confel$ is bisimilar to one such $q$. Hence, each candidate $\confel$-branch matches a candidate $q$-branch in $\chora$, thus it contains interactions where $\p$ participates.
\end{proof}
\fi

%%% Local Variables:
%%% mode: latex
%%% TeX-master: "main"
%%% End:

\section{Design-by-Contract}
\label{sec:achor}
We now extend the theory of choreography automata and communicating
systems to handle specifications amenable to predicate over data
exchanged through a protocol.
The basic idea is to frame the design-by-contract theory proposed
in~\cite{bhty10} for global types in the context of c-automata.
This theory advocates \emph{global assertions} to specify and verify
contracts among participants of a protocol.
Taking inspiration from Design-by-Contract (DbC)~\cite{mey92}, widely
used in the practice of sequential programming~\cite{hoa69,flo67}, a
global assertion is a global type decorated with logical formulae
predicating on the payload carried by interactions.
Just as in the traditional DbC, the use of logical predicates allows
one to specify protocols where the content of messages is somehow
constrained.
%%% %
%%% The endeavour of generalising c-automata with assertions adds new
%%% challenges.
%%% %
%%% A main difficulty is that we cannot exploit a tree-like structure
%%% to define our constructions.
\subsection{Asserted choreography automata}\label{sec:aca}
To specify protocols that encompass
constraints on payloads, we extend c-automata to \emph{asserted
  c-automata}.
%In order to deal with specifications of protocols that encompass
%constraints on payloads, we extend c-automata to \emph{asserted
%  c-automata}.
The structure of messages is reshaped to
account for sorted data in interactions and predicate over the payload
of communications.
More precisely, the set of \emph{messages} $\msgset$ consists of
\emph{tagged tuples} $\atuple \tuplevar$ where $\atag$ is a \emph{tag}
and $\tuplevar = \varid_1 \ \asort_1, \ldots, \varid_h \ \asort_h$ is
a tuple of pairwise distinct sorted variables (namely,
$\varid_i = \varid_j \implies i = j$ for $1 \leq i \leq j \leq h$).
The set of variables of
$\tuplevar = \varid_1 \ \asort_1, \ldots, \varid_h \ \asort_h$ is
$\fvar[\tuplevar] \mmdef \set{\varid_1, \ldots, \varid_h}$ and,
accordingly
%\[
  $\fvar[\amsg] \mmdef \fvar[\tuplevar]$ and
  $\fvar[\gint] \mmdef
  \fvar[\amsg]$
%\]
are the set of variables of $\amsg$ and of $\gint$ respectively.
Intuitively, %besides sender, receiver and type of a message,
now an
interaction specifies also the sort of the values communicated by the
sender and the \quo{local} variables where the receiver
\quo{stores} those values.
\begin{example}[OLW variable sorts]
  When asking \ptp[customer]\ for another login attempt, \ptp[wallet]\
  can send a message $\amsg[retry]\conf{\varid[msg] \ \asort[string]}$
  where the payload $\varid[msg]$ yields an error message.
  \finex
\end{example}

We borrow from~\cite{bhty10} (with minor syntactic changes) the
first-order logic to specify the constraints on payloads; the set
$\assertionset$ of logical formulae are derived from the
following grammar
\begin{align}\label{eq:pred}
  \apred, \predB
  \bnfas
  &
	 \truek \bnfalt \falsek
	 \bnfalt \phi(e_1, \ldots, e_n)
	 \bnfalt \neg \apred
    \bnfalt \apred \land \predB
    \bnfalt \apred \ENTAILS \predB
	 \bnfalt \EXISTS {\varid \ \asort} \apred
%	 \bnfalt \ldots
	 % \qquad\qquad%\\nonumber   e \bnfas \const \bnfalt  \varid \bnfalt \cdots
\end{align}
In~\eqref{eq:pred}, $\phi$ ranges over pre-defined atomic predicates
with fixed arities and sorts (e.g., $\asort[bool]$, $\asort[int]$,
etc) \cite[\S 2.8]{men87} and $e_1,\ldots,e_n$ denote expressions.
Instead of fixing a specific language of expressions, we just assume
that they encompass usual data types of programming languages and
variables $\varid$.
Also, we assume that sorts of expressions can be inferred (hence, we
occasionally omit sorts and tacitly assume that usage of variables is
consistent with respect to their sort).
  For simplicity, we consider only basic sorts (as
  in~\cite{bhty10}). More complex static data structures can be handled similarly, while
  dynamic data structures (e.g., pointers) require to extend
  our theory with suitable semantics of value passing (e.g.,
  deep-copy).
  %
%
% Hereafter, we denote a finite vector of pair-wise disjoint
% interaction variables by $\vec{\varid}$.

Let $\fvar[e]$ be the set of variables occurring in expression $e$;
likewise $\fvar[\apred]$ denotes the set of free variables of
predicate $\apred \in \assertionset$, while $\bvar[\apred]$ denotes
the bound variables in $\apred$ (defined in the standard way).
Hereafter, assume that $\fvar[\apred] \cap \bvar[\apred] = \emptyset$.
\begin{example}[OLW payloads]\label{ex:payload}
  The payloads of the OLW protocol which we will use through the paper
  are those in the following FSA:
  \[\begin{tikzpicture}[node distance=1cm and 3cm, scale = .8, transform shape]
      \tikzstyle{every state}=[cnode]
      \tikzstyle{every edge}=[carrow]
      \node[state, initial, initial text = {}] (q0) {$q_0$};
		\node[state, below = 3cm of q0] (q1) {$q_1$};
		\node[state, right = of q1] (q2) {$q_2$};
		\node[state, right = of q2] (q3) {$q_3$};
		\node[state, right = of q0] (q4) {$q_4$};
		\node[state, right = of q4] (q4') {$q_5$};
		\node[state, right = of q4'] (q5) {$q_6$};
		\node[state, right = of q5] (q6) {$q_7$};
		\node[state, below left = of q5,xshift=.5cm] (q7) {$q_8$};
		\path (q0) edge node[below] {$\gint[C][W][{\amsg[login]\conf{\varid[account]\ \asort[int]}}]$} (q1);
		\path (q1) edge node[below] {$\gint[C][W][{\amsg[pin]\conf{\varid[pin]\ \asort[int]}}]$} (q2);
		\path (q2) edge[bend left] node[above,yshift=.3cm] {$\gint[W][C][{\amsg[retry]\conf{\varid[msg]\ \asort[string]}}]$} (q0);
		\path (q2) edge node[below] {$\gint[W][C][{\amsg[loginDenied]\conf{\varid[msg]\ \asort[string]}}]$} (q3);
		\path (q2) edge node[above] {$\gint[W][C][{\amsg[loginOK]\conf{}}]$} (q4);
		\path (q4) edge node[above] {$\gint[W][V][{\amsg[loginOK]\conf{}}]$} (q4');
		\path (q4') edge node[above] {$\gint[V][C][{\amsg[request]\conf{\varid[bill]\ \asort[int]}}]$} (q5);
		\path (q5) edge node[above] {$\gint[C][W][{\amsg[authorise]\conf{}}]$} (q6);
		\path (q5) edge node[below] {$\gint[C][W][{\amsg[reject]\conf{}}]$} (q7);
		\path (q6) edge[bend left] node[below] {$\gint[C][V][{\amsg[pay]\conf{\varid[payment]\ \asort[int]}}]$} (q3);
		\path (q7) edge[bend left] node[below] {$\gint[C][V][{\amsg[reject]\conf{}}]$} (q3);
	 \end{tikzpicture}
  \]
  Notice that some messages have empty payloads.
  \finex
\end{example}

We will consider FSAs where transitions are decorated with
\emph{assertions}, namely formulae in $\assertionset$ predicating on
variables of the FSAs.
The interplay between payloads and assertions requires some care to
handle iterative behaviour and the scoping of variables.
In fact, we will need to slightly
change the FSA above
to handle the iteration of the authentication phase.

  %
  % To handle assertions on iterative computations,
  % transitions are labelled with recursion variables; more precisely,
  % we use FSAs on the set $\widehat \lint$ (ranged over by $\lambda$),
  % defined as the union of $\lint$ and the set of recursion variables.
  % %
  % Also, recursion variables are associated to: (i) a set of variables,
  % corresponding to the formal parameters of the recursive invocation,
  % (ii) an assertion, corresponding to the loop invariant to be
  % maintained through the iteration, and (iii) a state, corresponding
  % to the start of the iteration.
  % %
  % Formally, a \emph{recursion context} $\recctx$ maps each recursion
  % variable ${\gtvar r}$ to a triplet $({\tuplevar}, {\apred[A]}, q)$ consisting
  % of a tuple of sorted variables $\tuplevar$, a predicate
  % $\apred \in \assertionset$, and a state $q$ of the FSA so that,
  % assuming $\recctx(\gtvar r) = (\tuplevar, \apred, q)$ and
  % $\recctx(\gtvar r') = (\tuplevar', \apred', q')$ we have that
  % $\gtvar r \neq \gtvar r'$ implies $q \neq q'$.
  %

%\il{
  %
  Iterative computations require a few more ingredients.
  First we fix a \emph{recursion context} $\recctx$ which maps each
  recursion variable $\gtvar r$ to a triplet $(\tuplevar, \apred, q)$
  consisting of
  \begin{itemize}
  \item a set of sorted variables $\tuplevar$ which identify
	 the formal parameters of $\gtvar r$,
  \item a predicate $\apred \in \assertionset$, the loop invariant to
	 be maintained through the iteration, and
  \item a state $q$ of the FSA identifying the start of the iteration.
  \end{itemize}
  We assume that if $\recctx(\gtvar r) = (\tuplevar, \apred, q)$ and
  $\recctx(\gtvar r') = (\tuplevar', \apred', q')$ then
  $\gtvar r \neq \gtvar r'$ implies $q \neq q'$ and
  $\tuplevar \cap \tuplevar' = \emptyset$.
  Then we use FSAs on the set $\widehat \lint$ (ranged over by
  $\lambda$), defined as the union of $\lint$ and the set of
  \emph{recursive calls} which are defined as pairs $\reccall r \iota$
  of a recursive variable and a map assigning expressions to recursive
  parameters of $\gtvar r$.
  %
%}
% \begin{itemize}
% \item if $\gtvar r \neq \gtvar r'$ then $q \neq q'$
% \item if $\gtvar r = \gtvar r'$ then $q = q'$ and
%   $\asort_i = \asort'_i$ for all $1 \leq i \leq h$.
% \end{itemize}

\begin{example}[OLW iteration]\label{ex:rec}
  Using assertions, the constraint on the authentication phase of the
  OLW protocol described in \cref{sec:intro} can be specified as
  follows:
  \[\begin{tikzpicture}[node distance=1.5cm and 4cm, scale=.8, transform shape]
      \tikzstyle{every state}=[cnode]
      \tikzstyle{every edge}=[carrow]
      \node[state, initial, initial text = {}] (q0') {$q_0'$};
      \node[state, right = 3cm of q0'] (q0) {$q_0$};
		\node[state, below = 3cm of q0] (q1) {$q_1$};
		\node[state, right = of q1] (q2) {$q_2$};
		\node[state, right = of q2] (q3) {$q_3$};
		\node[state, right = of q0] (q4) {$q_2'$};
		\node[state, right = of q4] (q4') {$q_4$};
		\path (q0') edge node[above] {$\reccall r {\varid[try] \mapsto 0}$} node[below]{$0 \leq \varid[try] \leq 3$} (q0);
		\path (q0) edge node[above] {$\gint[C][W][{\amsg[login]\conf{\varid[account]\ \asort[int]}}]$} node[below]{$\truek$} (q1);
		\path (q1) edge node[above] {$\gint[C][W][{\amsg[pin]\conf{\varid[pin]\ \asort[int]}}]$}node[below]{$\truek$} (q2);
		\path (q4) edge node[above] {$\reccall r {\varid[try] \mapsto \varid[try]+1}$} node[below]{$0 \leq \varid[try] \leq 3$} (q0);
		\path (q2) edge node[above, xshift=.2cm] {$\gint[W][C][{\amsg[loginDenied]\conf{\varid[msg]\ \asort[string]}}]$}
		node[below] {$\varid[try] \geq 3 \land \varid[msg] = \text{"fail"}$} (q3);
		\path (q2) edge node[above] {
		  $\gint[W][C][{\amsg[retry]\conf{\varid[msg]\ \asort[string]}}]$
		}
		node[below, rotate=-90,xshift=1cm,yshift=1.3cm]{
		  $\begin{array}{c}
			  0 \leq \varid[try] < 3
			  \\\land\\
			  \varid[msg] = \text{"fail"}
			\end{array}
		  $
		 } (q4);
		\path (q2) edge node[above] {
		  $\gint[W][C][{\amsg[loginOk]\conf{}}]$
		}
		node[below]{%, near end]{
		  $\begin{array}{c}
			  0 \leq \varid[try] \leq 3
			  %\\\land\\
			  %\varid[msg] = \text{"success"}
			\end{array}
		  $
		 } (q4');
	 \end{tikzpicture}
  \]
  where
  $\recctx(\gtvar r) = (\set{\varid[try]}, 0 \leq \varid[try] \leq 3,
  q_0)$.
  The automaton above refines the left part of the c-automaton in
  \cref{ex:payload}. In particular, states with the same names do
  correspond. States $q_0'$ and $q_2'$ are new (in particular $q_0'$
  is the new initial state), introduced to correctly model iteration.
	 %
	 %The ac-automaton
	 The assertions on the transitions from states $q_0'$ and $q_2'$
	 model recursive calls where the $\varid[try]$ parameter is
	 respectively set to $0$ and incremented (cf. \cref{ex:reqs}).
  \finex
\end{example}

Transitions $t=(q, (\lambda, \apred), q')$, written as
$q \arro[\apred]\lambda q'$, are interpreted according to their label:
\begin{itemize}
\item If $\lambda = \gint$ then $t$ (dubbed \emph{interaction
	 transition}) establishes a rely-guarantee relation: when $t$ is
  fired, \p \emph{guarantees} $\apred$ while \q \emph{assumes} that $\apred$
  holds.
\item
  % If $\lambda = \gtvar r$ then $t$ (dubbed
  % \emph{iteration transition}) specifies an invariant $\apred$ that
  % should be maintained through each loop corresponding to
  % $\gtvar r$.
	 %
	 If $\lambda = \reccall r \iota$ then $t$ (dubbed \emph{iteration
		transition}) records the invariant $\apred$ (fixed by the
	 recursion context $\recctx$) that should be maintained through
	 each loop corresponding to $\gtvar r$.
\end{itemize}

Variable scoping requires attention, as best illustrated by the
following example.
\begin{example}[Confusion]\label{ex:confusion}
  In the following FSA
  \[\begin{tikzpicture}[node distance=1cm and 2cm]
      \tikzstyle{every state}=[cnode]
      \tikzstyle{every edge}=[carrow]
      \node[state, initial, initial text = {}] (q0) {$q_0$};
		\node[state, above right = of q0, xshift=1cm, yshift=-.5cm] (q1) {$q_1$};
		\node[state, below right = of q0, xshift=1cm, yshift=.5cm] (q2) {$q_2$};
		\node[state, right = 6cm of q0] (q3) {$q_3$};
		\node[state, right = of q3] (q4) {$q'$};
		\path (q0) edge node[above] {$\gint[p][r][\atuple{\varid\ \asort[bool]}]$} node[below]{$\truek$} (q1);
		\path (q1) edge node[above] {$\gint[q][r][\atuple{\varid[w]\ \asort[int]}]$} node[below]{$\truek$} (q3);
		\path (q0) edge node[above] {$\gint[q][r][\atuple{\varid\ \asort[int]}]$} node[below]{$\truek$} (q2);
		\path (q2) edge node[above] {$\gint[p][r][\atuple{\varid[u]\ \asort[bool]}]$} node[below]{$\truek$} (q3);
		\path (q3) edge node[above] {$\gint[r][s][\atuple{\varid[y]\ \asort[int]}]$} node[below]{$\varid[y] = \varid \mod 2$} (q4);
	 \end{tikzpicture}
  \]
  it is not clear
  if the assertion on the transition from $q_3$ predicates on the variable $\varid$ bound in the
  interaction between \p\ and \ptp[r]\ or in the one between \q\ and
  \ptp[r], hence its sort is not clear.
  \finex
\end{example}

The binding and scoping of variables yield a first
difference w.r.t.~\cite{bhty10}, where syntactic structures
of global assertions facilitate the definition of these notions.
The lack of syntactic structures of c-automata requires instead to
introduce constructions to handle variables.

Let us now consider recursion.
An FSA $A$ \emph{respects} a recursion context $\recctx$ when there
are no loops without iteration transitions and for each iteration
transition $t = q \arro[\apred]{\reccall r \iota} \hat q$ in $A$ with
$\recctx(\gtvar r) = (\tuplevar, \apred, \hat q)$
  \begin{enumerate}[(a)]
  \item\label{it:selfloop} $t$ is the only outgoing transition of $q$
	 and $q \neq \hat q$ and
  \item\label{it:guard} either $q$ is the initial state of $A$ or there is a
	 unique transition entering $q$ and it is an interaction transition.
  \end{enumerate}
  %
  % Condition~\eqref{eq:rec} disciplines iteration;
  Condition~\eqref{it:selfloop} forbids self-loops
  while~\eqref{it:guard} forces iterations to be guarded by
  interactions.
\begin{example}[OLW is respectful]\label{ex:reqs}
  The requirements imposed by respectfulness are met by the FSA in \cref{ex:rec}.  \finex
\end{example}
  For an FSA
  $A = (Q, q_0, \widehat \lint \times \assertionset, \tset)$ on
  $\widehat \lint \times \assertionset$, we let $\simpleruns[q][A]$
  denote the set of simple paths\footnote{A path is \emph{simple} if
	 no state occurs twice on it.} reaching the state $q \in Q$ from
  $q_0$; also, $\fvar[q {\arro[\apred]\aint} q'] \mmdef \fvar[\aint]$
  and
  $\fvar[q {\arro[\apred]{\reccall r \iota}} q'] \mmdef \fvar[\gtvar
  r] \mmdef \tuplevar$ if
  $\recctx(\gtvar r) = (\tuplevar, \apred, \hat q)$.
  Finally, we say that a transition $t \in \tset$ from a state
  $q \in Q$ \emph{fixes a variable $\varid$ (in $A$)} if
  $\varid \in \fvar[t]$ and, for each path
  $\arun \in \simpleruns[q][A]$ there is no transition $t' \in \arun$
  that fixes $\varid$.

The next definition addresses the issues of confusion and
respectfulness described above.
\begin{definition}[Asserted c-automata]\label{def:aca}
  % Let $A = (Q, q_0, \widehat \lint \times \assertionset, \tset)$ be an
  % FSA on $\widehat \lint \times \assertionset$.
  % %
  % A transition $t \in \tset$ from a state $q \in Q$ \emph{fixes a
  % 	 variable $\varid$ (in $A$)} if $\varid \in \fvar[t]$ and, for each
  % path \eMcomm[{it was\\$\arun \in \simpleruns[q][\chora]$}]{$\arun
  % 	 \in \simpleruns[q][A]$} there is no transition $t' \in
  % \arun$ that fixes $\varid$.

  An FSA, say $\chora$, on the alphabet
  $\widehat \lint \times \assertionset$ such that
  \begin{enumerate}
  \item\label{eq:sort} for each co-final span $(\arun,\arun')$ in
	 $\chora$, if there are $t \in \arun$ and $t' \in \arun'$ such that
	 both $t$ and $t'$ fix $\varid$ then $t$ and $t'$ assign the same
	 sort to $\varid$
  \item\label{eq:rec} $\chora$ respects the (fixed) recursion context $\recctx$
  \item\label{eq:detaca} the underlying c-automaton obtained by removing the assertions from $\chora$ is deterministic
  % \item\label{eq:var} for each variable $\varid$ occurring in $\chora$
  % 	 there is a transition that fixes $\varid$;
  % \item\label{eq:rec} for each transition
  % 	 $t = q \arro[{\apred[B]}]{\gtvar r} q'$ in $\chora$
  % 	 with $\recctx(\gtvar r) = (\tuplevar, \apred, \hat q)$
  % 	 \begin{enumerate}[(a)]
  % 	 \item\label{it:selfloop} $t$ is the only outgoing transition of $q$ and $q' \neq q$
  % 	 \item\label{it:guard} $q = q_0$ or there is a unique transition entering $q$ and it is an
  % 		interaction transition
  % 	 \item\label{it:inv} $\apred[B]$ is obtained by substituting in
  % 		$\apred$ each $\varid \in \tuplevar$ with some expressions
  % 		$e_{\varid}$.
  % 	 \end{enumerate}
  \end{enumerate}
  is an \emph{asserted c-automaton} (ac-automaton for short).
\end{definition}
Intuitively, one can think of a variable $\varid$ fixed at a
transition $t$ as \quo{local} to the receiver of the interaction
labelling $t$; also, the sender of the interaction is aware of the value
to be assigned to $\varid$.
Condition~\eqref{eq:sort} in \cref{def:aca} simply avoids confusion
on the sort of a variable when it could be assigned along different
paths.
%
% Condition~\eqref{eq:var} does not allow assertions to predicate on
% \quo{free} variables; this does not limit expressiveness: when
% violated, \eqref{eq:var} can be restored by alpha-converting
% variables.
% The notion of well-sequencedness extends to ac-automata in the natural
% way.
% \begin{definition}[Well-sequenced ac-automata]\label{def:aca-ws}
%   An ac-automaton $\chora$ is \emph{well-sequenced} if for each two
%   consecutive transitions
%   $q \arro[\apred]{\aint} q' \arro[{\apred[B]}]{{\aint[b]}} q''$, if
%   $\aint \not\parallel \aint[b]$ then there is $q'''$ such that
%   $q \arro[{\apred[B]}]{\aint[b]} q''' \arro[\apred]{\aint} q''$ and,
%   for each transition $q''' \arro[\apred']{\aint[g]} q''''$ in
%   $\chora$, $\aint[g] \parallel \aint$ and
%   $\aint[g] \parallel \aint[b]$.
% \end{definition}

Without loss of generality, we can assume that
$\fvar[t] \cap \bvar[\apred] = \emptyset$ for all transitions and
predicates $\apred$ of an ac-automaton; in fact, such condition can be
enforced by simply renaming bound variables in predicates.
Hereafter, we write $q \arro\lambda q'$ instead of
$q \arro[\truek]\lambda q'$.

%%% Local Variables:
%%% mode: latex
%%% TeX-master: "main"
%%% End:

\subsection{Consistent choreography automata}\label{sec:wa}
Our interpretation of transitions as rely-guarantee relations
requires some care.
Indeed, for a transition $t$ to be viable, participants involved in
$t$ must \quo{know} the variables used in $t$.
In particular, if $t$ is an interaction variable then the sender and
receiver in $t$ must \quo{know} the assertion in $t$ and participants
involved in an iteration should \quo{know} the invariant of the loop.
Before formalising this in the next definition, we introduce the
  auxiliary concept of \emph{assertion of a path of an ac-automaton},
  which yields the conjunction of all assertions in $\arun$ while
  substituting recursive variables with actual values of recursive
  calls.
  Formally, if $t = q_1 \arro[{\apred}]{\reccall{r}{\iota}} q_2$ then
  $\nabla(t)=\iota$, otherwise $\nabla(t)$ is the empty substitution.
  %, the \emph{assignment of a transtion $t$
%	 in an ac-automaton $\chora$}, be the substitution $\iota$ if
  %$\apred[B] = \apred \iota$ if
  
  %and
  %$\recctx(\gtvar r) = (\tuplevar, \apred, q)$
  %otherwise $\nabla(t)$
  
  %
  Then the \emph{assertion of a path $\arun$} is defined
  as
  $\assertionof[\arun][] =
  \assertionof[\arun][\mathsf{id}]$ where
  % , given a substitution $\iota$ from the recursion parameters in
  % $\chora$ to expressions we define
  \[
	 \assertionof[\emptyword][\iota] \mmdef \truek
	 \quad\qand\quad
	 \assertionof[{q \arro[\apred] \lambda q' \
		\arun}][\iota] \mmdef \apred\iota' \land
	 \assertionof[{\arun}][\iota']
	 \qand[with]
	 \iota' = \iota[\nabla(q \arro[\apred] \lambda q')]
  \]
  Namely, the assertion of a path is the conjunct of all
  the assertions of its transitions once the recursion
  parameters are updated with their actual values.
We can now define the notion of \emph{knowledge} of a variable.

\begin{definition}[Knowledge]\label{def:knows}
  Let $\chora$ be an ac-automaton.
  A participant $\p \in \ptpset$ \emph{knows $\varid$ at a transition
	 $t = q \arro[\apred]\lambda q'$ in $\chora$} if
  \begin{itemize}
  \item either $t$ fixes $\varid$ and
	 \begin{enumerate}[(a)]
	 \item if $\lambda \in \lint$ then $\p \in \ptpof[\lambda]$ and
	 \item
		% if $\lambda$ is a recursion variable with
		% $\recctx(\lambda) = (\tuplevar, \apred[B], q')$ and participant
		% $\p \in \ptpset$ is on a cycle from $q'$ to $q'$ then
		% $\varid \in \fvar[\tuplevar]$
		if $\lambda = \reccall r \iota$ with
		$\recctx(\gtvar r) = (\tuplevar, \apred, q')$ and 
		$\p$ is on a cycle from $q'$ to $q'$ then
		$\varid \in \tuplevar$
	 \end{enumerate}
  \item or $\varid \in \fvar[\apred]$ and there are a variable
	 $\varid[u]$ and a transition $t'$ on each path
	 $\arun \in \simpleruns[q][\chora]$ such that \p knows $\varid[u]$ at $t'$
	 and $\assertionof[\arun][] \ENTAILS \varid = \varid[u]$ holds.
  % \item or, for each path $\arun \in \simpleruns[q][\chora]$, there is a
  % 	 variable $\varid[y]$ known to \p at a transition $t' \in \arun$
  % 	 such that $\assertionof[q][\chora] \land \apred \ENTAILS \varid = \varid[y]$
  % 	 holds.
  \end{itemize}
  Let $\knw$ be the set of variables that \p knows at $t$ in
  $\chora$.
\end{definition}
\begin{example}[OLW knowledge]
  In the FSA of \cref{ex:payload} both $\ptp[vendor]$ and
  $\ptp[customer]$ know $\varid[bill]$ at the outgoing transition of
  state $q_5$.
  Also, $\ptp[customer]$ and $\ptp[wallet]$ know the recursion variable
  $\varid[try]$ of the ac-automaton in \cref{ex:rec}.
  \finex
\end{example}
  The notion of knowledge in \cref{def:knows} is more complex than the
  one in~\cite{bhty10}; this is an effect of the higher
  complexity in the notions of binding and scoping of variables.
\cref{def:knows} is instrumental to transfer the concept of
\emph{history-sensitivity} introduced in~\cite{bhty10} to ac-automata.
\begin{definition}[History sensitiveness]\label{def:hs}
  An ac-automaton $\chora$ is \emph{history-sensitive} if the
  following holds for each transition $t = q \arro[\apred]\lambda q'$
  in $\chora$
  \begin{enumerate}
  \item\label{it:trans} $\lambda = \gint$ implies
	 $\fvar[\apred] \subseteq \knw$, namely $\ptp$ knows each
	 variable free in $\apred$ at $t$.
  \item\label{it:rec}
	 % $\lambda = \gtvar r$ implies
	 % $\fvar[\gtvar r] \subseteq \knw$ for each participant \eMcomm[{it
	 % 	was\\$\ptp[x] \in \ptpset$}]{$\ptp[p] \in
	 % 	\ptpset$} occurring on a cycle from $q'$ to $q'$.
	 $\lambda = \reccall r \iota$ implies $\fvar[\gtvar r] \subseteq
	 \knw$ for each $\ptp[p] \in \ptpset$ occurring on a cycle from
	 $q'$ to $q'$.
  \end{enumerate}
\end{definition}
Condition~\eqref{it:trans} guarantees that the assertion of a
transition cannot predicate on variables not \quo{accessible} to the
participants of the interaction.
Condition~\eqref{it:rec} ensures that participants involved in a loop
are aware of the loop invariant.
  The notion of history sensitivity in~\cite{bhty10} relies on the
  fact that participant \ptp[p]\ knows a variable
  $\varid$ on each interaction involving $\varid$.
  Here instead a weaker notion is adopted since, due to selective
  participation, the c-automaton may have a transition fixing
  $\varid$ but not involving \ptp[p].
\begin{example}[OLW is history-sensitive]
  The ac-automaton in \cref{ex:rec} is history-sensitive.
  In particular, note that the variable
  $\varid[try]$ in the assertion on the transition from $q_2$ to
  $q_3$ is known to $\ptp[customer]$ and
  $\ptp[wallet]$ since it is in the invariant of the authentication
  loop.
  \finex
\end{example}

For a transition $t$ of an ac-automaton $\chora$ to be enabled, it is
not enough that the source state of $t$ is reachable from the initial
state of $\chora$.
In fact, the transition $t$ can be fired if the information accumulated by
the participants ensures the satisfiability of the assertion of $t$.
To formalise this notion we introduce the following definitions.
%
% Given a state $q$ of an ac-automaton $\chora$, the logical formula
% \begin{align}\label{eq:preof}
%   \preof[q][\chora] \mmdef \bigvee_{\arun \text{ run to } q \text{ in } \chora}
%   \assertionof[\arun][]
% \end{align}
% is the \emph{precondition of $q$ (in $\chora$)}.
%
% The \emph{enabling condition} of a state $q$ of
% $\chora$ is the logical formula
% \begin{align}\label{eq:ec}
%   \enablingof[q][\chora] \mmdef
%   \preof[q][\chora] \ENTAILS \bigvee_{q \arro[\apred]\lambda q' \in \chora}\EXISTS {\fvar[\lambda]} \apred
% \end{align}
Given a state $q$ of an ac-automaton $\chora$, we let
\begin{align}\label{eq:preof}
  \preof[q][\chora] \mmdef \set{\assertionof[\arun][] \sst \arun \text{ run to } q \text{ in } \chora \text{ and } \assertionof[\arun][] \text{ is satisfiable}}
\end{align}
be the set of \emph{preconditions} of $q$ (in $\chora$) and
\begin{align}\label{eq:ec}
  \enablingof[q][\chora] \mmdef \bigcup_{ \apred[B] \in \preof[q][\chora]}\Set{
  \apred[B] \ENTAILS \bigvee_{q \arro[\apred]\lambda q' \in \chora}\EXISTS {\fvar[\lambda]} \apred
  }
\end{align}
be the set of \emph{enabling conditions} of $q$ (in $\chora$)

Similarly to~\cite{bhty10} for global types, progress of ac-automata
cannot be guaranteed if there is a possible computation leading to a
state with no enabled transitions.
Hence, we adapt from~\cite{bhty10} the notion of \emph{temporal
  satisfiability}.
\begin{definition}[Temporal satisfiability]\label{def:tsat}
  % An ac-automaton $\chora$ is \emph{temporally satisfiable} if
  % $\enablingof[t][\chora]$ is satisfiable for each transition $t$ of $\chora$.
	 %
	 An ac-automaton $\chora$ is \emph{temporally satisfiable} if for
	 each $q \in \chora$ reachable from the initial state of $\chora$
	 %and with some outgoing transition there is a transition $t$ from
	 %$q$ such that
	 each formula in $\enablingof[q][\chora]$ is satisfiable.
\end{definition}
\begin{example}[OLW is temporally satisfiable]
	 The ac-automaton in \cref{ex:rec} is temporally satisfiable because
	 the enabling conditions of all the nodes are satisfiable.
         %of $q_2$ entails the assertion on the transition from $q_2$ to $q_3$.
	 %
	 However, if the assertion on the transition from $q_2$ to $q_3$
	 were replaced by e.g., $\varid[try] > 3 \land \varid[msg] =
	 \text{"fail"}$ then temporal satisfiability would be violated
	 because the precondition of the simple path from $q_0$ to $q_2$
	 would not entail $0 \leq \varid[try] < 3 \lor \varid[try] > 3$.
	 % 
	 % \iLcomm[False: transitions out of $q_2$ are mutually exclusive hence
	 % they cannot both be implied, maybe it is enough to imply one
	 % transition out of each state?]{The ac-automaton in \cref{ex:rec}
	 % is temporally satisfiable.}
	 %
	 % \iLcomm[To be updated with new states and conditions]{However,
	 % if the assertion on the transition from $q_3$ to $q_4$ were
  	 % $\varid[try] < 0$ instead of \eMcomm[check: it was\\$try <
  	 % 3$]{$\varid[try] <
  	 % 	2$} then temporal satisfiability would be violated because
  	 % \iLcomm[write it?]{the precondition} of the simple path from
  	 % $q_0$ would not entail $\varid[try] < 0$.}
  \finex
\end{example}

% \eMcomm[R17A\\removed\\Many notions defined on c-automata
% straightforwardly transfer to ac-automata.
%
% Roughly speaking, many notions do not depend on assertions, which
% are the basic difference between c- and ac-automata.
%
% For the sake of precision we give here the detailed constructions
% even if the presentation becomes heavier.  ]{
%
% As before, well-formedness of ac-automata is defined as the
% conjunction of well-sequencedness and well-branchedness where the
% latter is defined below.
%
As c-automata, ac-automata are well-formed if they are well-sequenced
and well-branched; these two notions are as for c-automata modulo the
presence of assertions, which are disregarded; \iftr we refer to
\cref{def:ac-sw,def:ac-fullaware,def:ac-wb} \else see~\cite{glsty22TR}
\fi for the formal definitions.
\iftr
  % Let us first give the definition of well-sequencedness.
  % 
  We state explicitly the definitions for well-formedness for
  ac-automata.
  \begin{definition}[Well-sequencedness for ac-automata]\label{def:ac-sw}
	 An ac-automaton $\chora$ is \emph{well-sequenced} if for each two
	 consecutive transitions
	 $q \arro[\apred]{\aint} q' \arro[{\apred[B]}]{{\aint[b]}} q''$ either
	 \begin{enumerate}[(a)]
	 \item\label{seq:ac-dep} $\aint \not\parallel \aint[b]$ or
	 \item\label{seq:ac-indep} there is $q'''$ such that
		$q \arro[{\apred[B]}]{\aint[b]} q''' \arro[\apred]{\aint} q''$;
		furthermore for each transition
		$q''' \arro[\apred']{\aint[g]} q''''$, $\aint[g] \parallel \aint$
		and $\aint[g] \parallel \aint[b]$.
	 \end{enumerate}
  \end{definition}

  For well-branchedness we need to slightly adjust the notion of trace;
  traces of runs of an ac-automaton simply ignore assertions:
  \begin{align*}
	 \traceof[\emptyset] \mmdef \emptyword
	 \qqand
	 \traceof[{q \arro[\apred]\lambda q'\ \arun}] \mmdef \lambda\ \arun
  \end{align*}
  Likewise, we let $\ptpof[(\aint, \apred)] \mmdef \ptpof[\aint]$.
  We can now tune up full-awareness for ac-automata.
  \begin{definition}[Full awareness for ac-automata]\label{def:ac-fullaware}
	 Let $(\arun_1,\arun_2)$ be a pair of $q$-runs of an ac-automaton
	 $\chora$.
	 Participant $\p \in \ptpof[\arun_1] \cap \ptpof[\arun_2]$ is
	 \emph{fully aware} of $(\arun_1,\arun_2)$ if there are two labels
	 $(\aint_1,\apred_1), (\aint_2,\apred_2) \in \lint \times
	 \assertionset$ such that $\aint_1 \neq \aint_2$,
	 $\p \in \ptpof[\aint_1] \cap \ptpof[\aint_2]$, and
	 \begin{enumerate}
	 \item\label{it:ac-diff} either $\aint_h$ is the first interaction in $\lang[\arun_h]$
		% \in \aint_h \lint^\star$
		for $h=1,2$
	 \item\label{it:ac-fa} or there are proper prefixes $\hat\arun_1$ of
		$\arun_1$ and $\hat\arun_2$ of
		$\arun_2$ such that $\traceof[{\proj{\hat\arun_1} p}] =
		\traceof[{\proj{\hat\arun_2} p}]$, the partner of
		\p in $\aint_h$ is fully aware of $(\hat \arun_1, \hat
		\arun_2)$, and $\traceof[\hat \arun_h] \aint_h \in
		\pref[{\traceof[\arun_h]}]$ for $h \in \set{1,2}$.
	 \end{enumerate}
  \end{definition}
  Notice that, analougously to full-awareness for c-automata,
  \cref{def:ac-fullaware} considers only interaction labels.
  Finally we define well-branchedness.
  \begin{definition}[Well-branchedness for ac-automata]\label{def:ac-wb}
  An ac-automaton
  $\chora$ is \emph{well-branched} if for all of its states
  $q$ there is a partition
  $T_1,\dots,T_k$ of $\transset{\chora}{q}$ such that
	 \begin{itemize}
	 \item for all $1 \leq i \neq j \leq k$, $\ptpof[T_i] \cap
		\ptpof[T_j] = \emptyset$ and for each $q \arro[\apred_i]{\aint_i}
		q_i \in T_i$, $q \arro[\apred_j]{\aint_j} q_j \in
		T_j$ there exists $q'$ such that $q_i \arro[\apred_j]{\aint_j} q'$ and $q_j
		\arro[\apred_i]{\aint_i} q'$
	 \item for all $1 \leq i \leq k$,
		$\bigcap_{t \in T_i} \ptpof[t] \neq \emptyset$ and for all
		$\p \in \ptpof[\chora] \setminus \bigcap_{t \in T_i} \ptpof[t]$
		and $q$-span $(\arun_1,\arun_2)$ starting from transitions in
		$T_i$, if $\proj{\arun_1}{p} \neq \proj{\arun_2}{p}$ then either
		\p is fully aware of $(\arun_1,\arun_2)$ or there is
		$i \in \set{1,2}$ such that $\p \not\in \ptpof[\arun_i]$ and
		\begin{itemize}
		\item the first transition in $\arun_{3-i}$ involving \p is with a
		  participant which is fully aware of $(\arun_1,\arun_2)$ and
		\item for all runs $\arun'$ such that $\arun_i\arun'$ is a candidate $q$-branch of
		  $\chora$ the first transition in $\arun'$ involving \p is with a
		  participant which is fully aware of $(\arun_1,\arun_2)$.
		\end{itemize}
	 \end{itemize}
  \end{definition}
\fi
Finally, we can define \emph{consistent} ac-automata.

\begin{definition}[Consistency]\label{def:wa}
  % Let $\chora$ be an ac-automaton.
  %
  % The c-automaton \emph{subsumed} by $\chora$ is the determinisation
  % of the c-automaton obtained by ignoring the assertions of $\chora$.
  %
  An ac-automaton is \emph{consistent} if it is history-sensitive,
  temporally satisfiable, and well-formed.
\end{definition}

%%% Local Variables:
%%% mode: latex
%%% TeX-master: "main"
%%% End:

\subsection{Asserted communicating systems}\label{sec:aCS}

Projecting ac-automata requires to handle asserted transitions.
We therefore extend communicating systems to \emph{asserted
  communicating systems} (a-CSs for short), which basically are
communicating systems where CFSMs are \emph{asserted} (a-CFSMs for
short), namely they have transitions decorated with formulae
in $\assertionset$.
The synchronous semantics of a-CSs can be defined as an LTS similarly
to the semantics of communicating systems.
In fact, configurations can be defined as in \cref{def:syncSem} taking
into account assertions when synchronising transitions.
This basically means that assertions are used to verify that a sent
message guarantees the expectation of its receiver, that is the
assertion the receiver relies upon.

Recall that a \emph{prenex normal form} is a formula
$\mathcal{Q} \apred$ where $\mathcal{Q}$ is a sequence of quantifiers
and variables (called \emph{prefix}) and $\apred$ is a
quantifiers-free logical formula (called \emph{matrix})~\cite{men87}.
If $\apred, \apred[B] \in \assertionset$ then
$\apred \pcirc \apred[B]$ is a logical formula obtained by quantifying
with the prefix of a prenex normal form $\apred'$ logically equivalent
to $\apred$ the conjunction of $\apred[B]$ with the matrix of
$\apred'$.
Similarly to assertions for paths on ac-automata, we define
assertions of a run of an a-CFSM
\[
  \assertionofcm[\emptyword][]  \mmdef  \truek
  \quad\qqand\quad
  \assertionofcm[{q \arro[\apred] \ell q' \ \arun}][]
  \mmdef
  \apred \pcirc \assertionofcm[\arun][]
\]
The preconditions of a state of an a-CFSM are defined as for
ac-automata but for the use of the assertion function
$\assertionofcm[][]$ for CFSMs instead of the corresponding one for
ac-automata.
%
 %Also, given a logical formula $\apred \in \assertionset$ and a set $\varid[X]$ of variables, we let $\fexists$ denote the formula obtained by existentially quantifying the variables in $\fvar[\apred] \setminus \varid[X]$.

\begin{definition}[Semantics of a-CS]\label{def:acfsm}
  The \emph{semantics} of an a-CS
  $\aCS = (\aCM_{\p})_{\p \in \ptpset}$ is the transition system
  $\ssem{\aCS}$ defined by taking the set of configurations as in
  \cref{def:syncSem} and as set of transitions the smallest set
  including
  \begin{itemize}
  \item $\confel_1 \arro[\apred]{\gint} \confel_2$ if $\p, \q \in \ptpset$ and
	 \begin{itemize}
	 \item $\pstate{1}{\p} \arro[\apred]{\aout} \pstate{2}{\p}$ in
		$\aCM_{\p}$,
		$\pstate{1}{\q} \arro[{\apred[B]}]{\ain} \pstate{2}{\q}$ in
		$\aCM_{\q}$ and,
		  there are $\apred' \in \preof[\pstate 1 \p][\aCM_{\p}]$ and
		  $\apred[B]' \in \preof[\pstate 1 \q][\aCM_{\q}]$ such that it holds
$%		  \[
			 (\apred' \ENTAILS \apred) \land (\apred[B]' \ENTAILS
			 \apred[B]) \land (\apred' \pcirc \apred[B]' \pcirc \apred)
			 \ENTAILS \EXISTS{\fvar[\amsg]}{\apred[B]}
$%		\]
	 \item and $\pstate{1}{\ptp[x]} = \pstate{2}{\ptp[x]}$ for all
		$\ptp[x] \in \ptpset \setminus \set{\p,\q}$
	 \end{itemize}
  \item $\confel_1 \arro[\apred]{\emptyword} \confel_2$ if
	 $\p \in \ptpset$ and
	 \begin{itemize}
	 \item $\pstate{1}{\p} \arro[\apred]{\emptyword} \pstate{2}{\p}$ in
		$\aCM_{\p}$ and there is
		$\apred' \in \preof[\pstate 1 \p][\aCM_{\p}]$ such that
		$\apred' \ENTAILS \apred$
	 \item and $\pstate{1}{\ptp[x]} = \pstate{2}{\ptp[x]}$ for all
		$\ptp[x] \in \ptpset \setminus \set{\p}$.
	 \end{itemize}
  \end{itemize}
\end{definition}

Like the projection of communicating systems
(cf. \cref{def:projection}), the projection of a-CSs relies on the
determinisation and minimisation of a-CFSMs.
The presence of assertions imposes to adapt the classical
constructions on FSA to a-CFSMs.
More precisely, we have to generalise equality on labels of the form
$(\lambda,\apred)$.
Essentially, this is done by (injectively) renaming the variables
occurring in actions and assertions decorating transitions.
For $\sigma$ an endofunction on variables and
$\amsg = \atuple{\varid_1 \ \asort_1, \ldots, \varid_h \ \asort_h}$
let
$\amsg\sigma \mmdef \atuple{\sigma(\varid_1) \ \asort_1, \ldots,
  \sigma(\varid_h) \ \asort_h}$; we define
\begin{align*}
  \emptyword \sigma \mmdef \emptyword
  \qqand
  (\gint)\sigma \mmdef \gint[@][@][m']
  \quad\qqand[where]
  \amsg[m'] = \amsg\sigma
\end{align*}
Two labels $(\lambda, \apred)$ and $(\lambda', \apred')$ are
\emph{equivalent}, in symbols
$(\lambda, \apred) \sim (\lambda', \apred')$, if there is an injective
substitution of variables such that $\lambda = \lambda'\sigma$ and
$\apred$ is logically equivalent to $\apred'\sigma$.
We will similarly consider equivalence on
$\lact \times \assertionset$.

The $\varepsilon$-closure of an a-CFSM
$\aCM = \autom Q {q_0} \lact \tset$ is the map
$\eclos[][\aCM] : Q \to 2^{Q \times \assertionset}$ defined assigning
to each state $q$ of $\aCM$ the set of states reachable with
$\emptyword$-transitions together with their assertions; more
precisely, for each $q \in Q$, $\eclos[q][\aCM]$ is the smallest set
satisfying
\begin{align}\label{eq:eclos}
  \eclos[q][\aCM] \mmdef \set{(q, \truek)}
  \cup
  \bigcup_{(q',\apred) \in \eclos[q][\aCM]}\Set{(q'', \apred \pcirc \apred') \sst q' \arro[\apred'] \emptyword q'' \in \tset}
\end{align}
Removal of $\emptyword$-transitions from an a-CFSM $\aCM$ is computed,
using~\eqref{eq:eclos}, similarly to the classical algorithm on FSAs
$\conf{\mathbb{Q}, \eclos[q_0], \lact, \mathbb{T}}$ where
\begin{align*}
  \mathbb{Q} & = \set{\eclos[q][\aCM] \sst q \in Q}
  \qqand\\
  \mathbb{T} & = \set{Q \arro[\apred_1 \pcirc \apred \pcirc
  \apred_2]\ell Q' \sst q_1 \arro[\apred] \ell q_2 \in
  \tset \text{ for some } (q_1, \apred_1) \in Q \text{
	 and } (q_2, \apred_2) \in Q'}
\end{align*}
Handling assertions in the determinisation algorithm requires
some care.
We illustrate the problem in the following example.
\begin{example}[Non-determinism \& assertions]
  Consider the two a-CFSMs below
  \[\begin{tikzpicture}[node distance=.05cm and 3cm]
      \tikzstyle{every state}=[cnode]
      \tikzstyle{every edge}=[carrow]
      \node[state, initial, initial text={$\aCM = $}] (q0) {$q_0$};
		\node[state, above right = of q0] (q1) {$q_1$};
		\node[state, below right = of q0] (q2) {$q_2$};
		\path (q0) edge[bend left=15] node[above] {$\ell$} node[below]{$\apred$} (q1);
		\path (q0) edge[bend right=15] node[above] {$\ell$} node[below]{$\apred[B]$} (q2);
	 \end{tikzpicture}
	 \qquad\qquad
	 \begin{tikzpicture}[node distance=.35cm and 2cm]
      \tikzstyle{every state}=[cnode]
      \tikzstyle{every edge}=[carrow]
      \node[state, initial, initial text={$\aCM' =$}] (q0) {$q_0$};
		\node[state, above right = of q0] (q1) {$q_1$};
		\node[state, below right = of q0] (q2) {$q_2$};
		\node[state, right = 3cm of q0] (q1q2) {$q'$};
		\path (q0) edge[bend left=15] node[above] {$\ell$} node[below]{$\apred \land \neg \apred[B]$} (q1);
		\path (q0) edge node[above,near end] {$\ell$} node[below,near end]{$\apred \land \apred[B]$} (q1q2);
		\path (q0) edge[bend right=15] node[above] {$\ell$} node[below]{$\neg \apred \land \apred[B]$} (q2);
	 \end{tikzpicture}
  \]
  If both $\apred$ and $\apred[B]$ are satisfiable then $\aCM$ has a
  non-deterministic behaviour.
  We therefore aim to define a determinisation algorithm which on
  $\aCM$ yields something like $\aCM'$.
  Also, the new state $q'$ should provide transitions
  corresponding to both transitions from $q_1$ and $q_2$.
\finex
\end{example}

Let $\aCM = \autom Q {q_0} \lact \tset$ be a CFSM.
A state $q \in Q$ \emph{is non-deterministic on $\ell \in \lact$} if
its \emph{derivative in $\aCM$ with respect to $\ell$}, defined as
$\cmderiv[q,\ell][\aCM] \mmdef \set{(\apred, q') \sst q
  \arro[\apred]\ell q' \text{ in } \aCM}$, has more than one element.
Also, if $X,Y \subseteq \cmderiv[q,\ell][\aCM]$ then
$\assertionnd \mmdef \displaystyle{\bigwedge_{(\apred,q) \in X} \apred
  \land \bigwedge_{(\apred[B],q) \in Y} \neg \apred[B]}$.
The determinisation of $\aCM$ is obtained by applying the classical
FSA determinisation algorithm to the $\varepsilon$-closure of the
a-CFSM
$\aCM' = \autom{Q'}{q_0} \lact {\tset' \cup \tset'' \cup \tset'''}$
where
\begin{align*}
  Q' \mmdef & Q \cup \bigcup_{q \in Q, \ell \in \lact}\Set{\conf X \sst
			q \text{ is non-deterministic on }
			\ell \text{ and } \emptyset \neq X \subseteq \cmderiv[q,\ell][\aCM]
			}
  \\
  \tset' \mmdef & \Set{q \arro[\apred]\ell q' \in \tset \sst \cmderiv[q,\ell][\aCM] \text{ is a singleton}}
  \\
  \tset'' \mmdef & \bigcup_{\emptyset \neq X \subseteq \cmderiv[q,\ell][\aCM]}\Set{
				  q \arro[{\assertionnd[X,Y]}]\ell \conf{X} \sst \cmderiv[q,\ell][\aCM] \text{ not a singleton}
				  \qand
				  Y = \cmderiv[q,\ell][\aCM] \setminus X
				  }
  \\
  \tset''' \mmdef & \bigcup_{\emptyset \neq X \subseteq \cmderiv[q,\ell][\aCM]}\Set{
					\conf{X} \arro[\apred]\ell q' \sst \text{there is } q \arro[\apred]\ell q' \in \tset
					\qand[with] \set q \times \assertionset \cap X \neq \emptyset
				  }
\end{align*}
Basically, we ($i$) introduce a new state $\conf X$ for any
combination of assertions of $\ell$-transitions, ($ii$) replace
non-deterministic behaviours on $\ell$ with a set of
$\ell$-transitions with \quo{disjoint} assertions, and ($iii$) let
state $\conf X$ have the transitions that any of the states $q \in X$
has in $\aCM$.

We remark that the adaptation of the determinisation algorithm is
imposed by the use of a-CFSMs to model local behaviour.
This is a main technical difference with respect to~\cite{bhty10}
where local types with assertions, which need no determinisation, play
the role of a-CFSMs.

The projection of an ac-automaton acts as the projection of c-automata
on interactions and accommodates the variables not known to the
participant by existentially quantifying them.
This requires to consider the points in the ac-automaton where
variables are fixed.
\begin{definition}[Projection of ac-automata]\label{def:proj-aca}
  The \emph{projection on $\p \in \ptpset$} of an asserted transition
  $t$ in an ac-automaton $\chora$ on $\ptpset$, written
  $\proj{t}{\chora,p}$, is defined by:
  \[
    \proj{t}{\chora,p} \ = \begin{cases}
      q \arro[\apred]{\aout} q' & \text{if }  t = q \arro[\apred]{\gint} q'
      \\[1em]
		q \arro[\apred]{\ain[q][p]} q' &  \text{if } t = q \arro[\apred]{\gint[q][p]} q'
      \\[1em]
      q \arro[{\EXISTS {\varid[X]} \apred(\nabla(t))}]{\emptyword} q' & \text{if } t = q \arro[\apred]{\lambda} q', \ \p \not\in \ptpof[\lambda],
		\text{ and } \varid[X] = \set{\varid \in \fvar[\apred] \sst t \text{ fixes } \varid \text{ in } \chora}
    \end{cases}
  \]
  The \emph{projection} of $\chora$ on $\p \in \ptpset$, denoted
  $\proj{\chora} p$, is obtained by determinising and minimising
  up-to-language equivalence the \emph{intermediate} a-CFSM
  \[
    A_\p = \conf{\sset, q_0, \lact, \Set{\proj{(q \arro[\apred]{\lambda} q')}{\chora,p}
		  \sst q \arro[\apred]{\lambda} q' \text{ in } \chora}}
  \]
  where ($i$) syntactic equality of labels is replaced by $\sim$ and
  ($ii$) $\emptyword$-transitions are those with label of the form
  $(\emptyword, \apred)$.
  The \emph{projection of $\chora$}, written $\proj{\chora}{}$, is the
  a-CS $(\proj{\chora}{\p})_{\p \in \ptpset}$.
\end{definition}

\iftr
  Well-formed consistent ac-automata are deadlock-free; the proof
  mimics the one for c-automata.
  Let $q \arro[\apred_1,\ldots,\apred_n] \emptyword q'$ abbreviate
  $q \arro[\apred_1] \emptyword \ldots \arro[\apred_n] \emptyword q'$.
  \begin{lemma}\label{lem:ac-intfromruns}
	 If $A_\p$ and $A_\q$ are the intermediate a-CFSM for two
	 participants $\p$ and $\q$ of a well-formed ac-automaton $\chora$
	 and
	 \begin{align*}
		q {\arro[\apred_1,\ldots,\apred_n] \emptyword} q_{\p} \arro[\apred]{\aout}{q'_{\p}} \text{ in } A_{\p}
		\qqand
		q {\arro[\apred_1',\ldots,\apred_n'] \emptyword} q_{\q} \arro[\apred]{\ain}{q'_{\q}} \text{ in } A_{\q}
	 \end{align*}
	 then there is a state $q'$
	 such that $q \arro[\apred] \gint q' \in \transset \chora q$.
  \end{lemma}
  \begin{proof}
	 The proof of \cref{lem:intfromruns} can be repeated by observing
	 that assertions do not play any role for concurrent transitions.
	 Details follow.

	 Let the two runs for \p and \q be the projections of runs $\arun_\p$
	 and $\arun_\q$ in $\chora$.
	 We have two cases, depending on whether $\arun_\p = \arun_\q$.

	 If $\arun_\p = \arun_\q$ then the last transition of $\arun_\p$
	 is concurrent to all the previous ones.
	 Indeed, the previous transitions are projected to $\emptyword$ both on
	 $\p$ and on $\q$, hence neither $\p$ nor $\q$ can occur in the
	 label.
	 Thus, by well-sequencedness, the state $q$ has a transition with
	 label $(\gint,\apred)$.

	 Assume $\arun_\p \neq \arun_\q$.
	 By construction,
	 \[
		\arun_\p = \arun \ t_\p \ \arun'_\p \ q_\p \arro[\apred]\gint q'_\p
		\qqand
		\arun_\q = \arun \ t_\q \ \arun'_\q \ q_\q \arro[\apred]\gint q'_\q
	 \]
	 for a $q$-run $\arun$, two transitions $t_\p$ and $t_\q$, two runs
	 $\arun'_\p$ and $\arun'_\q$; observe that the ending state of
	 $\arun$, say $\hat q$, is also the source state of $t_\p$ and $t_\q$
	 while $\arun'_\p$ (resp. $\arun'_\q$) ends in $q_\p$ (resp. $q_\q$).

	 Note that $\p \not\in \ptpof[\arun \ t_\p \ \arun'_\p]$ and
	 $\q \not\in \ptpof[\arun \ t_\q \ \arun'_\q]$.
	 If $\q \not\in \ptpof[\arun \ t_\p \ \arun'_\p]$ or
	 $\p \not\in \ptpof[\arun \ t_\q \ \arun'_\q]$ then the thesis
	 immediately follows by well-sequencedness as before.
	 Therefore, without loss of generality, assume
	 $\p \in \ptpof[\arun \ t_\q \ \arun'_\q]$ (the case where \q occurs
	 in $\arun \ t_\p \ \arun'_\p$ is similar and hence omitted).

	 Consider now the two $\hat q$-runs $\arun_1 = t_\p \ \arun'_\p$ and
	 $\arun_2 = t_\q \ \arun'_\q$.
	 By well-branchedness, there is a partition of
	 $\transset\chora{\hat q}$ satisfying the conditions of
	 \cref{def:ac-wb} (using $\sim$ for label equality).
	 Then $t_\p$ and $t_\q$ cannot belong to the same equivalence class
	 of such partition since neither \p nor \q are fully aware of
	 $(\arun_1,\arun_2)$ and the first interaction of \p on $\arun_1$
	 is with \q.
	 Hence, $t_\p = \hat q \arro{\aint_1} q_1$ and
	 $t_\q = \hat q \arro{\aint_2} q_2$ necessarily belong to different
	 equivalence classes.
	 Therefore, again by well-branchedness, there is a state $\hat q'$
	 such that $q_1 \arro{\aint_2} \hat q'$ and
	 $q_2 \arro{\aint_1} \hat q'$.
	 Hence, the transitions of $\arun'_\p$ and those of $\arun'_\q$ form
	 commuting diamonds and therefore there is a $q$-run in $\chora$
	 where \q does not occur and all the transitions involving \p in
	 $\arun'_\q$ follow a transition with label $\gint$ (easily by
	 induction on the length of $\arun'_\p$ and $\arun'_\q$).
	 The thesis then follows since, as before, the transition labelled by
	 $\gint$ commutes with any preceding transition by
	 well-sequencedness.
  \end{proof}

  \begin{lemma}\label{lemma:ac-diffclass}
	 % Let $\chora$ be a well-formed ac-automaton and, for
	 % $i \in \set{1,2}$, $t_i = q \arro{\aint_i} q_i'$ the first two
	 % transitions of a $q$-span $(\arun_1,\arun_2)$ in $\chora$ such that
	 % $t_1$ and $t_2$ are concurrent.

	 % Then either $\aint_1$ occurs on $\traceof[\arun_2]$ or for each run
	 % $\arun_2 \ \arun$ in $\chora$, $\aint_1$ occurs in
	 % $\traceof[\arun]$.
	 Let $\chora$ be a well-formed ac-automaton and $(\arun_1,\arun_2)$ a
	 $q$-span in $\chora$ with first transitions $t_1$ and $t_2$,
	 respectively.
	 Let $t_1 = q \arro{\aint_1} q_1'$.  If $t_1$ and $t_2$ are
	 concurrent then either $\aint_1$ occurs on $\traceof[\arun_2]$ or
	 for each run $\arun_2 \ \arun$ in $\chora$, $\aint_1$ occurs in
	 $\traceof[\arun]$.
  \end{lemma}
  \begin{proof}
	 We can reshape the proof of \cref{lemma:diffclass} observing that
	 assertions are immaterial to the reasoning.
  \end{proof}

  \begin{lemma}\label{lemma:ac-cui}
	 Let $\chora$ be a well-formed ac-automaton,
	 $\p,\q \in \ptpset_{\chora}$, and $\aint \in \lint$.
	 If $\aword_{\p} \aint,\ \aword_{\q} \aint,\ \aword \in \lang[\chora]$
	 are three words such that $\proj{\aword_{\p}}{p}=\proj{\aword}{p}$
	 and $\proj{\aword_{\q}}{q}=\proj{\aword}{q}$ then
	 $\aword \aint \in \lang[\chora]$.
  \end{lemma}
  \begin{proof}
	 We can reshape the proof of \cref{lemma:cui} observing that
	 assertions are immaterial to the reasoning.
  \end{proof}
\fi

We show that projections of consistent ac-automata yield deadlock-free
asserted communicating systems.
The next result corresponds to \cref{prop:bisim} for ac-automata.
The main differences are (i) that consistency of ac-automata is
required (as opposed to well-formedness for c-automata) and (ii) that
an ac-automaton is weakly bisimilar to the corresponding projected
system due to the fact that iterative transitions of the ac-automaton
are projected on $\emptyword$-transitions.
%
% This requires to consider \emph{asserted bisimulations} so that
% asserted communicating systems are required to match only transitions
% of ac-automata whose assertions are actually satisfiable.
% %
% We therefore give the following definition.
% \begin{definition}[Asserted bisimulation]\label{def:a-bis}
%   A relation $\mathcal{R}$ among the states of an ac-automaton
%   $\chora$ and of an a-CS $\aCS$ is an \emph{asserted bisimulation} is
%   whenever $(q,s) \in \mathcal{R}$
%   \begin{itemize}
%   \item if $q \arro[\apred] \aint q'$ in $\chora$ with
% 	 $\aint \in \lint$ and
% 	 $\preof[q][\chora] \ENTAILS \EXISTS{\fvar[\apred]} \apred$ holds then
% 	 $s \arro[\apred] \aint s'$ in $\aCS$
%   \item if $t = q \arro[\apred] {\gtvar r} q'$ in $\chora$ with
% 	 $\preof[q][\chora] \ENTAILS \EXISTS{\fvar[\apred]} \apred$ holds then
% 	 $s \arro[\EXISTS{\varid[X]} \apred] \emptyword s'$ in $\aCS$ with
% 	 $\varid[X] = \set{\varid \in \fvar[\apred] \sst t \text{ fixes } \varid \text{ in } \chora}$

%   \item
%   \end{itemize}
% \end{definition}
\begin{proposition}\label{prop:ac-bisim}
  Any consistent ac-automaton $\chora$ is weakly
  bisimilar to $\ssem{\proj{\chora}{}}$.
\end{proposition}
\iftr
\begin{proof}
  Let $\chora=\conf{Q, q_0,\lint,\tset}$ and let $\sset$ be the set of
  configurations of $\ssem{\proj \chora {}}$.
  Recall that, due to determinisation and minimisation, for each
  $\confel \in \sset$ and each participant $\p$, $\confel(\p)$ is a
  subset of $Q \cup (Q \times \assertionset)$.
  Also, $q \in \confel(\p)$ holds if $q$ belongs to $\confel(\p)$ or
  if there is an assertion $\apred$ such that
  $(q,\apred) \in \confel(\p)$.
  We show by coinduction that  the relation
  \begin{align*}
	 \rcal = \set{(q,\confel) \in Q \times \sset \sst q \in
	 \confel(\p) \text{ for each } \p \in \ptpset_{\chora}}
  \end{align*}
  is a weak bisimulation, where in $\chora$ iterative transitions are
  treated as $\emptyword$-transitions.
  Since weak bisimulation implies trace equivalence,
  we also have that corresponding elements are reachable via the same
  trace.

  Let $(q,\confel) \in \rcal$; fixed
  $\apred[B] \in \preof[q][\chora]$, by definition of temporal
  satisfiability we have that
  \[
	 \apred[B] \ENTAILS \bigvee_{q \arro[\apred]\lambda q' \in
		\chora}\EXISTS {\fvar[\lambda]} \apred \qquad \text{is
		satisfiable}
  \]
  and we take an enabled transition $t$ from $q$ in $\chora$ (i.e.,
  the assertion of $t$ is entailed by $\apred[B]$).
  We have two cases depending on whether the challenge $t$ is an
  interaction or an iterative transition.
  \begin{itemize}
  \item If $t = q \arro[\apred] \gint q'$ then, by definition of
	 projection (cf. \cref{def:proj-aca}) and $\varepsilon$-closure,
	 $\confel(\p) \arro[\apred] \aout \confel'(\p)$ is in
	 $\proj \chora p$ and $\confel(\q) \arro[\apred] \ain \confel'(\q)$
	 is in $\proj \chora q$ (since $q \in \confel(\p) \cap
	 \confel(\q)$) while $\confel(\ptp[x]) = \confel'(\ptp[x])$ for
	 each $\ptp[x] \notin \set{\p,\q}$.
	 Since $\apred[B] \in \preof[q][\chora]$, by definition, there is a
	 run $\arun$ to $q$ in $\chora$ such that $\assertionof[\arun][]$
	 is satisfiable; hence $\proj \arun p$ is a run to $q$ in (the
	 intermediate of) $\proj \chora p$ with
	 $\assertionof[{\proj \arun p}][]$ satisfiable and entailing
	 $\apred$, and likewise for $\q$.
	 The thesis then follows from the definition of semantics of a-CS
	 since $\confel \arro[\apred] \gint \confel'$ because $\apred$
	 entails $\EXISTS{\fvar[\amsg]} \apred$.
  \item If $t = q \arro[\apred]{\reccall r \iota} q'$, then, by
	 definition of projection (cf. \cref{def:proj-aca}), for all
	 $\p \in \ptpset_{\chora}$ the projected a-CFSM $\proj \chora p$
	 contains the transition
	 $q \arro[{\EXISTS {\varid[X]} \apred(\nabla(t))}]{\emptyword} q'$
	 with
	 $\varid[X] = \set{\varid \in \fvar[\apred] \sst t \text{ fixes }
		\varid \text{ in } \chora}$.
	 Let $s$ be the configuration of $\ssem{\proj \chora {}}$
	 such that $q \in \confel(\p)$ for all $\p \in \ptpset_{\chora}$.
	 By \cref{def:acfsm}, $\ssem{\proj \chora {}}$ has a run
	 \begin{align*}
		\arun = s \arro[{\EXISTS {\varid[X]} \apred(\nabla(t))}]{\emptyword} \cdots  \arro[{\EXISTS {\varid[X]} \apred(\nabla(t))}]{\emptyword}s'
	 \end{align*}
	 such that $q' \in s'(\p)$ for all $\p \in \ptpset_{\chora}$ (run
	 $\arun$ is obtained by firing transition $t$ in a-CFSM
	 $\proj \chora p$).
	 Hence, $(q', s') \in \rcal$ as required.
  \end{itemize}

  Now, let $(q,\confel) \in \rcal$ and consider a challenge from
  $\ssem{\proj{\chora}{}}$, namely
  $\confel \arro[\apred] \gint \confel'$ (note that $\proj{\chora}{}$
  does not have $\emptyword$-transitions since the a-CFSMs projected
  from $\chora$ are determinised).
  By \cref{def:acfsm}, $\confel(\ptp[x]) = \confel'(\ptp[x])$ for each
  $\ptp[x] \notin \set{\p,\q}$,
  $\confel(\p) \arro[\apred]{\aout} \confel'(\p)$, and
  $\confel(\q) \arro[{\apred}]{\ain} \confel'(\q)$ are respectively in
  $\aCM_{\p} = \proj \chora p$ and $\aCM_{\q} = \proj \chora q$.
  Hence, by definition of determinisation, there are runs such that
  \[
	 \tilde q_\p \arro[\apred_1, \ldots, \apred_n] \emptyword q_\p \arro[\apred]{\aout}{q'_{\p}}
	 \text{ in } A_{\p} \qqand
	 \tilde q_\q \arro[\apred_1', \ldots, \apred_n'] \emptyword q_\q \arro[\apred]\ain q'_{\q}
	 \text{ in } A_\q
  \]
  where $A_\p$ and $A_\q$ are the intermediate automata of $\aCM_\p$
  and $\aCM_\q$ respectively.

  We have two cases:
  \begin{itemize}
  \item If $q=\tilde q_\p=\tilde q_\q $ then $q \arro[\apred]\gint q'$
	 by \cref{lem:ac-intfromruns}.
  \item Otherwise, since $q$ and $\confel$ are in the bisimulation,
	 they are also reached by the same trace.
	 Now consider a word $\aword$ that reaches $q$ and matching words
	 $\aword_\p$ and $\aword_\q$ obtained by lifting to $\chora$ runs
	 reaching $\tilde q_\p$ and $\tilde q_\q$ in the respective
	 auxiliary automata.
	 By construction, they are in the hypothesis of
	 \cref{lemma:ac-cui}, hence $q\arro[\apred]\gint q'$.
  \end{itemize}
  Therefore, in both cases we have $t = q \arro[\apred]\gint q'$ is in
  $\chora$.

  We have then to show that $q' \in \confel'(\ptp[x])$ for each
  $\ptp[x]$.
  By construction (\cref{def:proj-aca}), for each $\ptp[x]$, the
  intermediate a-CFSM $A_{\ptp[X]}$ contains the transition
  $t_{\ptp[x]} = q_{\ptp[x]} \arro[\apred_{\ptp[x]}]{\proj{(\gint)} x}
  q'_{\ptp[x]}$ for some $q_{\ptp[x]} \in \confel'(\ptp[x])$ and some
  $q'_{\ptp[x]}$ where
  \begin{itemize}
  \item if $\ptp[x] \notin \set{\p,\q}$ then
	 $\proj {(\gint)} x = \emptyword$ and $\apred_{\ptp[x]}$ is
	 ${\EXISTS {\varid[X]} \apred}$ with
	 $\varid[X] = \set{\varid \in \fvar[\apred] \sst t_{\ptp[x]} \text{
		  fixes } \varid \text{ in } \chora}$.
	 Hence $q_{\ptp[x]} \in \confel(\ptp[x])$ implies
	 $q'_{\ptp[x]} \in \confel'(\ptp[x])$ as required since
	 $\confel(\ptp[x])$ contains the $\varepsilon$-closure of its
	 elements by construction.
  \item If $\ptp[x] = \p$ then $\proj{(\gint)} x = \aout$ and
	 $\apred_{\ptp[x]} = \apred$.
	 Thus $q_{\ptp[x]} \arro[\apred]{\aout} q'_{\ptp[x]}$ and, as shown
	 above, $\confel(\p) \arro[\apred]{\aout} \confel'(\p)$.
	 Since $q_{\ptp[x]} \in \confel(\p)$ and $\chora$ is deterministic
	 then $q'_{\ptp[x]} \in \confel'(\ptp[x])$ since
	 $\confel'(\ptp[x])$ is the $\varepsilon$-closure of $q'_{\ptp[x]}$
	 by construction.
  \item If $\ptp[x] = \q$ then the reasoning is analogous to the previous case.
  \end{itemize}

  Finally, by \cref{def:acfsm}, there are
  $\apred' \in \preof[\confel(\p)][\aCM_{\p}]$ and
  $\apred'' \in \preof[\confel(\q)][\aCM_{\q}]$ such that both
  $\apred'$ and $\apred''$ entail $\apred$.
  Hence the thesis follows since $\apred$ entails
  $\EXISTS{\fvar[\amsg]}{\apred}$.
\end{proof}
\fi

As for c-automata, \cref{prop:ac-bisim} ensures that the language of a
consistent ac-automaton coincides with the language of its projection.

\begin{corollary}\label{th:acprojectionCorrectness}
  $\lang[\chora] = \lang[\ssem{\proj{\chora}{}}]$  for any consistent ac-automaton $\chora$.
\end{corollary}
\iftr
\begin{proof}
  From \cref{prop:ac-bisim} given that bisimulation implies trace equivalence.
\end{proof}
\fi

Final states and deadlock freedom of an ac-automaton are defined as
for c-automata (cf. \cref{def:final} and \cref{def:DF} respectively)
modulo the different labels of transitions.

\begin{theorem}[Projections of consistent ac-automata are deadlock-free]\label{thm:ac-df}
  If $\chora$ is a consistent ac-automaton then $\proj \chora {}$ is
  deadlock-free.
\end{theorem}
\iftr
\begin{proof}
  By contradiction, assume that there is a reachable configuration
  $\confel$ in $\ssem{\proj{\chora}{}}$ with no outgoing transition
  and a participant $\p$ for which $\confel(\p)$ is not final.
  Then, by definition of final state, for each $q \in \confel(\p)$ and
  each candidate $q$-branch, $\p \in \ptpof[\pi]$.
  From the proof of \cref{prop:ac-bisim}, there is a configuration
  $\confel$ bisimilar to $q$; hence, $\confel$ should answer the
  challenge from the first action of $\arun$, hence it has an outgoing
  transition against the hypothesis.
\end{proof}
\fi

Observe that \cref{thm:ac-df} requires ac-automata to be consistent;
in particular, it requires history sensitiveness (cf. \cref{def:hs})
and temporal satisfiability (cf. \cref{def:tsat}).
The two requirements ensure that assertions on the transitions do not
spoil deadlock freedom.

%%% Local Variables:
%%% mode: latex
%%% TeX-master: "main"
%%% End:

\section{TypeScript Programming via Flexible C-Automata}\label{sec:apply}
We showcase the main theoretical results and constructions in this
paper with a tool, \tool{}, the first implementation of Scribble
\cite{scribble-paper,FeatherweightScribble,nuscr} that relies on
c-automata, for deadlock-free distributed programming.
%\begin{resub}
\tool{} takes the popular \emph{top-down approach} %\il{
  to system development based on choreographic models, %},
following the original methodology of Scribble and
multiparty session types \cite{HYC08}.
The top-down approach enables \emph{correctness-by-construction}:
a developer provides a global description for the whole
communication protocol; by projecting the global protocol,
APIs are generated from local CFSMs,
which ensure the safe implementation of each participant.
%\end{resub}
% We showcase the main results and constructions
% in this paper with a tool, \tool{}, the first
% implementation of Scribble
% \cite{scribble-paper,FeatherweightScribble,nuscr}
% that relies on the theory of choreography automata, for
% deadlock-free distributed programming.
%\tool{} accomodates for delayed joins and targets TypeScript.
%Figure \ref{fig:tool} showcases the full toolchain.
The core theory
of c-automata from \cref{sec:theory} guarantees deadlock freedom for
the distributed implementation of flexible global protocols.  As a
first application we target web development, supporting in particular
the TypeScript programming language.

In this section we present our development in three steps:
\begin{enumerate}
\item
  %% \emph{translation of global protocols into choreography automata}:
  %% we formally define a function that relates global protocols,
  %% from multiparty session types theory, to choreography automata,
  %% and give a first discussion on the relation
  %% between the two formalisms;
  %\begin{resub}
    \emph{translation of global protocols into choreography automata}:
    for the specification of global protocols, \tool{} relies on
    the Scribble language, and
    global Scribble protocols are formally global multiparty session types
    protocols \cite{FeatherweightScribble}; we define a function that maps these into
    choreography automata, and discuss the relation
    between the two formalisms;
  %\end{resub}
\item \emph{protocol specification and projections}:
from the specification of the global protocol,
\tool{} generates, through its translation into c-automata and
the subsequent projection,
a collection of CFSMs, which are the
abstract representation of the communication behaviour
of each participant (cf. part \textcolor{pine}{\bf (a)},
\cref{fig:tool} on page~\pageref{fig:tool});
\item \emph{API generation for deadlock-free distributed web
	 development}: we discuss our choice of targeting TypeScript and
  web development, and illustrate how \tool{} provides support for
  this (cf. part \textcolor{pine}{\bf (b)}, \cref{fig:tool} on
  page~\pageref{fig:tool}); finally we comment on possible extensions.
  % to support \emph{assertions}.
\end{enumerate}
%% We exploit the OLW example to showcase the different
%% features of \tool{}, throughout the whole section.

\subsection{From Multyparty Session Protocols to C-Automata}\label{sec:global}
C-automata and asserted c-automata can be directly produced by the
system designer and fed to our approach to ensure their correct
behaviour. However, to improve the usability of the approach, our
implementation, detailed in the next section, integrates c-automata
with the Scribble framework. This framework is based on the
theory of global types, hence we study below the relations between
global types and c-automata.
The syntax of global types is given by the following grammar:
\begin{eqnarray*}
%\begin{align*}
  % \gtact{\pi} & \bnfas & & & \text{Global Actions} \\
  % & ~ & \gtactmsgnosort{p}{q}{l} & & \text{Message} \\
  % & ~ & \gtactprefix{\pi_1}{\pi_2} & & \text{Prefix} \\\\
  \dgt{G} \quad \bnfas \quad
  %& % &  & & \text{Global Type} \\ &
  \gtend % & & & \text{Termination} \\
  \quad \bnfalt \quad \gtrecur{r}{G} % & & & \text{Recursive Type} \\
  \quad \bnfalt \quad \gtvar r % & & & \text{Type Variable} \\
  \quad \bnfalt \quad \gtsigma{i \in I} \gint[@][\q_i][\amsg_i];\dgt{G}_i
  % & & & \text{Choice}
  %\end{align*}
\end{eqnarray*}
We simply write $\gint[@][\q_i][\amsg_i];\dgt{G}_i$ instead of
$\gtsigma{i \in I} \gint[@][\q_i][\amsg_i];\dgt{G}_i$ when $I = \set{i}$.
In a recursive type $\gtrecur r G$ all occurrences of the recursion
variable $\gtvar r$ in $\dgt G$ are bound (this is the only binder for
global types); we moreover assume that the occurrences of $\gtvar r$ in
$\dgt G$ are guarded.
Hereafter we assume the so-called Barendregt convention, that is
names of bound variables are all distinct and different from names of free variables.

The operational semantics of global types is the LTS induced by the
rules in \cref{fig:globallts} where labels are drawn from the alphabet $\lint$.
\begin{figure}
   \begin{small}
  \[
    \begin{array}{ll}
		\multicolumn 2 c{
      \text{\textsc{[Choice]}}\
      \cfsmtrans{\gtsigma{i \in I}\gint[@][\q_i][\amsg_i];\dgt{G}_i
		  }{
		  \gint[@][\q_j][\amsg_j]}{\dgt G_j}
		\quad (j\in I)
		}
      \\[2mm]%\\
      \text{\textsc{[Rec]}}\
      \vcenter{\infer{\cfsmtrans{\gtrecur{r}{G}}{\aint}{\dgt{G'}}}{
		  \cfsmtrans{\dgt{G}[\gtrecur{r}{G}/\dgt{\mathbf{r}}]}{\aint}{\dgt{G'}}
		  }
		  }
      \qquad\qquad
      \text{\textsc{[Pass]}}\
		\vcenter{
		  \infer{
		  \cfsmtrans{\gtsigma{i \in I}\gint[@][\q_i][m_i];\dgt{G_i}}{\aint}{\gtsigma{i \in I}\gint[@][\q_i][\amsg_i];
		  \dgt{G'_i}}
		  }{
		  \cfsmtrans{\dgt{G_j}}{\aint}{\dgt{G'_j}}
										& \p,\q_j \not\in \ptpof[\aint]
      & \forall j \in I
		  }}
    \end{array}
    \]
     \end{small}\vspace{-2mm}
  \caption{\label{fig:globallts}LTS semantics over global types}\vspace{-2mm}
\end{figure}
Since the semantics of global types is an LTS, it can be
represented as a c-automaton only if it is finite state. Unfortunately,
the interplay between rule \textsc{[Pass]} and recursion allows one to
generate infinite state LTSs, as shown below.

\begin{example}[Infinite-state LTS]\label{ex:infinite}
  Let
  $\dgt G_\text{inf} = \gtrecur r {\aint;(\aint;\gtvar r +
	 \aint[g];\aint[d];\gtvar r+ \aint[b].\gtend)}$ where
  $\aint[d] \parallel \aint$, $\aint[d] \parallel \aint[g]$,
  $\aint \not\parallel \aint[g]$,
  $\aint[b] \not\parallel \aint[d]$, and
  $\aint[b] \not\parallel \aint$.
  Note that the traces $(\aint\ \aint[g])^n$ are included in
  the semantics for all $n > 1$. Executing $(\aint\ \aint[g])^n$ results in the following computation:
  \begin{eqnarray*}
  %\[
  \gtrecur r {\aint;(\aint;\gtvar r + \aint[g];\aint[d];\gtvar r + \aint[b].\gtend)} %&
  \arro{\aint} %&
  \aint;\gtrecur r {\aint;(\aint;\gtvar r + \aint[g];\aint[d];\gtvar r)} + \aint[g];\aint[d];\gtrecur r {\aint;(\aint;\gtvar r + \aint[g];\aint[d];\gtvar r) + \aint[b].\gtend}\\
  %&
  \arro{\aint[g]}
  %&
  \aint[d];\gtrecur r {\aint;(\aint;\gtvar r + \aint[g];\aint[d];\gtvar r + \aint[b].\gtend)}%\\
  %&
  \quad\dots\quad%\\
  %&
  \arro{\aint[g]}
  %&
  \aint[d]^n;\gtrecur r {\aint;(\aint;\gtvar r + \aint[g];\aint[d];\gtvar r+\aint[b].\gtend)}
%\]
\end{eqnarray*}
States
$\aint[d]^n;\gtrecur r {\aint;(\aint;\gtvar r +
  \aint[g];\aint[d];\gtvar r+\aint[b].\gtend)}$ and
$\aint[d]^m;\gtrecur r {\aint;(\aint;\gtvar r +
  \aint[g];\aint[d];\gtvar r+\aint[b].\gtend)}$ are bisimilar only if
$n=m$. Indeed, one needs to execute $n$ times $\aint[d]$ (and an
$\aint$) before being able to execute $\aint[b]$.
\finex
\end{example}

It is worth remarking that the semantics in \cref{fig:globallts}
yields finite-state LTSs on global types without consecutive
independent transitions,
a restriction actually considered in many global type formalisms,
since rule \textsc{[Pass]} never applies.
Likewise, the semantics consisting of rules \textsc{[Choice]} and
\textsc{[Rec]} only generates finite-state LTSs.

Function $\cat$ below defines a c-automaton with subterms of $\dgt G$
as states, $\dgt G$ as initial state, labels in $\lint$, and
transitions inductively defined by the function $\catr$ below:
\begin{eqnarray*}
  \catr[\gtend] = \catr[\gtvar r] %& = &
  =
  \emptyset
   \ \qquad\qquad%\\
   \catr[\gtrecur{r}{G}] %& = &
   =
   \catr[\dgt{G}] \cup \set{\trans{\gtvar r}{\epsilon}{\gtrecur{r}{G}}, \trans{\gtrecur{r}{G}}{\epsilon}{\dgt G} }
   %\\
   %\catr[\gtvar r] & = & \emptyset
   \\
   \catr[{\gtsigma{i \in I} \gint[@][\q_i][\amsg_i];\dgt{G}_i}] %& = &
   =
   \bigcup_{j \in I} (\set{\trans{\gtsigma{i \in I} \gint[@][\q_i][\amsg_i];\dgt{G}_i}{\gint[@][\q_j][\amsg_j]}{\dgt{G}_j}} \cup
   \catr[\dgt{G}_j])
\end{eqnarray*}

\begin{proposition}\label{lem:nopass}
  Let $\dgt{G}$ a global type. The language of $\cat$ coincides with
  the language generated by rules \textsc{[Choice]} and \textsc{[Rec]}
  of the semantics of $\dgt{G}$.
\end{proposition}
\iftr
\begin{proof}
  First note that the languages of $\cat[\gtrecur{r}{G}]$ and of
  $\cat[{\dgt{G}[\gtrecur{r}{G}/\dgt{\mathbf{r}}]}]$ do coincide. The
  only non trivial point is when the recursion variable is reached.
  % In
  % $\cat[\gtrecur{r}{G}]$ we go back to the initial state, in
  % $\cat[{\dgt{G}[\gtrecur{r}{G}/\dgt{\mathbf{r}}]}]$ one we go to the
  % same state in a copy of the graph, hence the languages do coincide.
  %
  State
  $\cat[\gtrecur{r}{G}]$ has a transition returning to  state $\dgt G$ whilse
  state
  $\cat[{\dgt{G}[\gtrecur{r}{G}/\dgt{\mathbf{r}}]}]$ has a transition
  to the first state of an unfolding of $\dgt G$, hence the languages do coincide.
  We have to prove two inclusions.
  For the inclusion of the language of the semantics in the language of $\cat[\dgt{G}]$ the proof is by rule induction. The case for \textsc{[Choice]} is by construction. For rule \textsc{[Rec]} the thesis follows from the observation above.

  For the other inclusion the proof is by structural induction on
  $\dgt{G}$. The only difficult case is the one of recursion, which follows from the observation above.
\end{proof}
\fi

Function $\catpass$ below extends $\cat$ to deal with the semantics of global types with rule \textsc{[Pass]}. However, the computed LTS may be infinite state, hence not a c-automaton, and in this case the function cannot be used in practice. This is, e.g., the case with the global type in \cref{ex:infinite}. The LTS has $\dgt G$ as initial state, labels in $\lint$, transitions
inductively defined by the function $\catrpass$ below, and as states the ones occurring in the transitions:
\begin{eqnarray*}
  \catrpass[\gtend] = \catrpass[\gtvar r] \hspace{-2mm}& = \hspace{-2mm}&
  \emptyset
    \\
    \catrpass[\gtrecur{r}{G}] \hspace{-2mm}& = \hspace{-2mm}&
    \catrpass[\dgt{G}] \cup \set{\trans{\gtvar r}{\epsilon}{\gtrecur{r}{G}}, \trans{\gtrecur{r}{G}}{\epsilon}{\dgt G} }
    %    \\
   %\catrpass[\gtvar r] & = & \emptyset
       \\
       \catrpass[{\gtsigma{i \in I} \aint_i;\dgt{G}_i}] \hspace{-2mm}& = \hspace{-2mm}&
       %&&\hspace{-40mm}
    \bigcup_{j \in I} (\set{\trans{\gtsigma{i \in I} \aint_i;\dgt{G}_i}{\aint_j}{\dgt{G}_j}} \cup
    \catrpass[\dgt{G}_j])
    \cup %\hspace{-10mm}
    \\[-3mm]
    &&
    \hspace{-8mm}
    \bigcup_{\aint \textrm{ s.t.\ }  \cfsmtrans{\dgt{G_i}}{\aint}{\dgt{G'_i}} \land \aint_i \parallel \aint \forall i \in I }
    \hspace{-14mm}
    \set{
      \trans{\gtsigma{i \in I} \aint_i;\dgt{G}_i}{\aint}{\gtsigma{i \in I} \aint_i;\dgt{G}'_i }}  \cup \catrpass[\gtsigma{i \in I} {\aint_i} ;\dgt{G}'_i]%
\end{eqnarray*}%
\vspace{-5mm}\begin{proposition}\label{lem:pass}
  Let $\dgt{G}$ a global type. The language of $\catpass$ coincides
  with the language generated by the semantics of $\dgt{G}$.
\end{proposition}
\iftr
\begin{proof}
  As in \cref{lem:nopass} we can notice that the languages of
  $\cat[\gtrecur{r}{G}]$ and of
  $\cat[{\dgt{G}[\gtrecur{r}{G}/\dgt{\mathbf{r}}]}]$ do coincide.

  We have to prove two inclusions.
  For the inclusion of the language of the semantics in the language of $\cat[\dgt{G}]$ the proof is by rule coinduction. The case for \textsc{[Choice]} and for \textsc{[Pass]} is by construction. For rule \textsc{[Rec]} the thesis follows from the observation above.

  For the other inclusion the proof is by structural coinduction on
  $\dgt{G}$. The only difficult case is the one of recursion, which follows from the observation above.
\end{proof}
\fi

We remark that global types with infinite semantics cannot be
implemented faithfully using communicating systems with the semantics
in \cref{def:syncSem}. Indeed, a communicating system has a finite
number of configurations, which is $O(S^n)$ where $S$ is the size of
the largest CFSM and $n$ the number of participants.

%%% Local Variables:
%%% mode: latex
%%% TeX-master: "main"
%%% End:

%% \subsection{Implementation}\label{sec:impl}

\label{subsec:tool}
%% \begin{figure}
%%     \begin{center}
%%       %\includegraphics[width=\textwidth]{images/lessugly.png}
%%       \input{toolchain}
%%     \end{center}
%%   \caption{Toolchain of \tool{}}
%%   \label{fig:tool}
%% \end{figure}

\subsection{Validating Global Protocols with Choreography Automata}
\label{subsec:design}

The first component of our toolchain is part \textcolor{pine}{\bf (I)}
in \cref{fig:tool}; it allows the user to perform protocol
specification, well-formedness checks, and the generation of CFSMs for
each participant.

%% For all three participants, when we call the code generation for their
%% API, the command refers to the formal specification of the protocol:
Let us consider the OLW example: the first step for the user is to
specify the global protocol, \texttt{OnlineWallet.scr}
(\cref{fig:ow-scribble}), in the \emph{Scribble protocol description
  language}, often referred to as \quo{the practical incarnation of
  multiparty session
  types}~\cite{scribble-paper,FeatherweightScribble}.
The syntax of Scribble (\url{http://www.scribble.org},
\url{https://nuscr.dev/}) has a straightforward correspondance to the
syntax of global types, so Scribble implementations of communicating
processes will be supported by multiparty session type theory, and
inherit its semantic guarantees.  Our development for part
\textcolor{pine}{\bf (I)} of the toolchain is based on the
$\nu$\textsf{Scr} implementation \cite{nuscr}, but fundamentally
differs from this (and other Scribble versions) in two aspects:
\begin{itemize}
\item the underlying choreographic objects---normating the
  communication among multiple participants---are not global types,
  but c-automata, and
\item we allow for participants to join the communication at later stage,
only in branches where they are needed (selective participation).
\end{itemize}

\cref{fig:ow-scribble} shows the protocol \texttt{OnlineWallet.scr}
for the OLW. Noteworthy, unlike  $\nu$\textsf{Scr}, we can specify
the selective participation of the vendor.
In particular, the \ptp[vendor]{} participant is involved only in the
first branch of the choice (lines 7-12), namely on successful login.

\begin{figure}
  \begin{subfigure}[b]{.4\linewidth}
       \begin{center}
       \begin{tikzpicture}[font=\scriptsize,scale=.7,transform shape]

\node[draw,
  fill=purple!30,
  minimum width=6cm,
  minimum height=1.5cm,
  rounded corners] (block0) at (2,0.25)
     {$\begin{array}{c}\text{User input}\\ \text{}\\ \text{}\\ \text{}\\ \end{array}$};

\node[draw,
  fill=violet!30,
  minimum width=1.5cm,
  minimum height=0.75cm,
  rounded corners] (block1) at (0,0)
     {$\begin{array}{c}\text{Scribble}\\ \text{protocol}\end{array}$};

\node[draw,
  fill=violet!30,
  minimum width=1.5cm,
  minimum height=0.75cm,
  rounded corners] (block1p) at (3,0)
     {$\begin{array}{c}\text{Participant declaration:}\\ \text{server and others}\end{array}$};

\node[draw,
  fill=gray!30,
  %below=of block1,
  minimum width=1.5cm,
  minimum height=0.75cm,
  rounded corners] (block2a) at (0,-2) {CA};

\node[%behind path, draw,
  %fill=gray!30,
  %below=of block1,
  minimum width=1.5cm,
  minimum height=0.75cm,
  rounded corners] (block2b) at (0,-2.5) {$\begin{array}{c}\ \\ \text{\textcolor{red}{WF checks}}\end{array}$};

\node[draw,
  fill=blue!30,
  %below=of block2,
  minimum width=1.5cm,
  minimum height=0.75cm,
  rounded corners] (block3) at (0,-4.5) {CFSMs};

\node[draw,
   fill=LimeGreen!10,
   minimum width=1.5cm,
   minimum height=0.75cm,
   rounded corners] (block4) at (3.7,-2){\textsf{STScript}};%{Jinja};

\node[draw,
  fill=orange!20,
  minimum width=4cm,
  minimum height=4cm,
  rounded corners] (block5) at (3.7,-5){};

\node[
  minimum width=2cm,
  minimum height=0.75cm,
  rounded corners] (block5a) at (3.7,-3.5){$\begin{array}{c}\ \text{ Generated APIs for }\ \\ \ \text{ TypeScript web development }\ \end{array}$};

\node[draw,
  fill=red!20,
  minimum width=2cm,
  minimum height=0.75cm,
  rounded corners] (block5b) at (3.7,-4.75)
     {$\begin{array}{c}\text{Node.js}\\ \text{(server)}\end{array}$};

\node[draw,
  fill=red!20,
  minimum width=2cm,
  minimum height=0.75cm,
  rounded corners] (block5c) at (3.7,-6)
     {$\begin{array}{c}\text{React}\\ \text{(non-server)}\end{array}$};

     \draw [-latex](block1) edge  node[right]{$\cat[]$} (block2a);
     \draw [-latex](2.6,-0.4) |- (2.95,-1.9);
\draw [-,rounded corners] (-0.75,-2) --++ (0,-0.9) -| (0.75,-2);
\draw [-latex](block2b) edge node[right]{projection} (block3);
\draw [-](block3) edge (1.5,-4.5);
\draw [-latex](1.5,-4.5) |- %node[near end,above]{$\nu$\textsf{Scr}}
 (2.95,-2.1);
\draw [-latex](block4) edge  %node[above]{$\nu$\textsf{Scr} $+$ \textsf{STScript}}
(block5);
 %\draw [-latex](1.75,-4.5) |-  node[near end, above]{$\nu$\textsf{Scr} $+$ \textsf{STScript}} (block5);
\draw [decorate, decoration = {brace,mirror}, thick, color=pine] (-0.75,-7.5) -- node[below]{\textbf{(I)}} (0.75,-7.5);
\draw [decorate, decoration = {brace,mirror}, thick, color=pine] (1.25,-7.5) -- node[below]{\textbf{(II)}} (5.75,-7.5);

\end{tikzpicture}

%%% Local Variables:
%%% mode: latex
%%% TeX-master: "main"
%%% End:
     \end{center}
   \caption{Toolchain of \tool{}}
       \label{fig:tool}
	  \end{subfigure}
	  \qquad
\begin{subfigure}[b]{.6\linewidth}
\begin{lstlisting}[language=Scribble]
global protocol OnlineWallet
(role Wallet, role Customer, role Vendor) {
  rec AuthLoop {
    login(account: int) from Customer to Wallet;
    pin(pin: int) from Customer to Wallet;
    choice at Wallet {
      login_ok() from Wallet to Customer;
      login_ok() from Wallet to Vendor;
      request(bill: int) from Vendor to Customer;
      choice at Customer {
        authorise() from Customer to Wallet;
        pay(payment: int) from Customer to Vendor;
      } or {
        reject() from Customer to Wallet;
        reject() from Customer to Vendor;
      }
    } or {
      login_retry(msg: string) from Wallet to Customer;
      continue AuthLoop;
    } or {
      login_denied(msg: string) from Wallet to Customer;
    }}}\end{lstlisting}\vspace*{-10mm}
   \caption{Scribble Protocol for the OLW}
     \label{fig:ow-scribble}
     \end{subfigure}
     %%  \begin{subfigure}[b]{.4\linewidth}
     %%  \includegraphics[width=\textwidth]{images/owCA.png}
     %% \caption{Generated Choreography Automaton}
     %%  \label{fig:ow-CA}
     %% \end{subfigure}
     \caption{\tool{}: Toolchain and OLW Protocol}
     %% \label{fig:ow-protocol}
\end{figure}

After its specification,
the Scribble protocol is translated into a c-automaton,
with the implementation of the function $\cat[]$ from \cref{sec:global}
 (this is exactly the c-automaton from \cref{ex:OLWca}, \cref{sec:back}).
  %\textbf{TODO}{add discussion about the toggle and termination --- as soon as Neil has it}
 On this automaton, well-sequencedness and well-branchedness checks
 are performed.
 If the c-automaton passes the above well-formedness checks, it is
 then projected onto each participant (\cref{def:projection}), thus
 obtaining a collection of CFSMs, whose semantics is equivalent to the
 one of the original c-automaton. Both global c-automata and local
 CFSMs are represented using the DOT graph description language.
 \cref{ex:OLWproj} from \cref{sec:back} shows the CFSM obtained by
 projection on \ptp[vendor]{} of the c-automaton for the OLW; for the
 other participants analogous CFSMs are obtained.  The local CFSM
 representations provide the communication behaviour of each
 participant and, as such, they retain all the information for
 obtaining deadlock-free endpoint implementations.
%% The immediate next step for our toolchain is to transform such
Each CFSM is the projection of the global c-automaton onto one
of the communicating participants; from this local automaton,
the API for the implementation of the participant is generated.

% The next section provides more detail on the usage of \tool{}
% for TypeScript API generation and, consequently, for
% deadlock-free web programming.

%We propose a similar design for our toolchain. Protocols are specified in an extended version of Scribble, to accomodate for delayed joins. Each protocol is then mapped into a c-automaton, well-formedness checks are conducted on it, and the CFSMs for the API are obtained, by projecting the CA onto each participant.

%% \begin{figure}
%% \begin{subfigure}[b]{.32\linewidth}
%%       \includegraphics[width=\textwidth]{images/wallet.png}
%%      \caption{CFSM for \lstinline[language=Scribble]{Wallet}}
%%       \label{fig:wallet}
%% \end{subfigure}
%% \begin{subfigure}[b]{.32\linewidth}
%%       \includegraphics[width=\textwidth]{images/customer.png}
%%      \caption{CFSM for \lstinline[language=Scribble]{Customer}}
%%       \label{fig:customer}
%% \end{subfigure}
%% \begin{subfigure}[b]{.3\linewidth}
%%       \includegraphics[width=\textwidth]{images/vendor.png}
%%      \caption{CFSM for \lstinline[language=Scribble]{Vendor}}
%%       \label{fig:vendor}
%% \end{subfigure}
%% \caption{Projected CFSMs from the Online Wallet CA}
%% \label{fig:cfsms}
%% \end{figure}

\subsection{API Generation for Distributed Web Development}
\label{subsec:api-gen}

Our chosen domain of application is \emph{distributed web development}.
%\begin{resub}
By nature, web services are distributedly developed and feature communication
among multiple participants. In services where some courses of actions are optional,
it is likely that the participation of some role is also optional
(selective participation).
% An instance of selective participation
% is services that outsource an authentication functionality, as witnessed by
% the popular OAuth 2.0 \cite{oauth2} and Kerberos \cite{kerberos}:
% in these protocols an external resource can be
% accessed (namely, \emph{takes part to the communication}) only if
% the authentication is granted, but not otherwise.
Our OLW example is a minimal, yet representative example that
selective participation is commonplace in transactions, auctions, or
contracts.  For instance, Kickstarter~\cite{kickstarter} is a
worldwide popular crowdfunding platform where the money of
\emph{supporters} is given to a project \emph{initiator} only if the
initially set goal is met; otherwise the money is returned to
supporters.
In other words, when the deadline is passed, if the goal is met, only
the initiator is involved in the communication, if not, only the
supporters are.
%% Another example is a shipping company that offers different insurance options: while some options do not involve external insurers others could do.
%% Notably, this does not mean that previous approaches to session types were unusable, since in many cases the issue can be solved by adding additional communications. However, such communications (besides increasing the communication overhead) are not natural in practice, and involve participants when they are not needed (e.g., if you offer different authentication services, then when one is chosen without selective participation you need to involve all of them, telling the chosen one that it is needed and the others that they are not).
%\end{resub}

%\begin{resub}
More technically, our development builds on
and extends \textsf{STScript} \cite{stscript}.
%\end{resub}
%% We exploit and integrate in \tool{}
%% the codebase of \textsf{STScript} for API generation.
We target server-centric protocols
(based on the WebSocket standard \cite{websocket}),
where one role is chosen as privileged,
the \emph{server}. The generated APIs are compatible
with the Node.js runtime for server-side endpoints
and the React.js framework for browser-side endpoints.
The  \textsf{STScript} toolchain in \cite{stscript} is based on
the multiparty session type theory,
where there is no privileged role; hence which role is the server
has to be declared by the user. The same holds for our
development based on c-automata.
We have discussed in the previous section
%% \cref{subsec:design}
how the Scribble protocol in input is translated into
a c-automaton and, once well-formedness checks are performed, projected
onto a CFSM for each participant. This CFSM is passed to the
code-generation component of our
toolchain (part \textcolor{pine}{\bf (II)} of \cref{fig:tool}),
together with the role in input and the information about whether it is
the server role or not.

\begin{figure}\centering
\includegraphics[width=.8\textwidth]{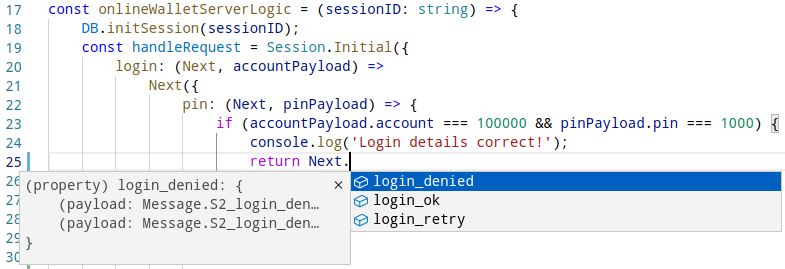}
\caption{Implementation with the API for \ptp[wallet]{} in Visual Studio Code}
\label{fig:autocomplete}
\end{figure}

\cref{fig:autocomplete} shows an example of the usage of the
generated API, when implementing the participant \ptp[wallet]{}
in Visual Studio Code (\url{https://code.visualstudio.com/}).
The autocomplete function of the editor
offers the developer appropriate options, so that
the implementation of the login choice abides
by the global discipline of the OnlineWallet protocol.

From an engineering point of view, for developing the part
\textcolor{pine}{\bf (I)} of the toolchain (\cref{fig:tool}) we have
adapted to our theory of c-automata, the codebase of
$\nu$\textsf{Scr}: a recent implementation of Scribble that offers a
toolchain
for \quo{language-independent code generation}~\cite{nuscr}.
However, $\nu$\textsf{Scr} itself does not provide direct support for TypeScript.
Hence, the development of part \textcolor{pine}{\bf (II)} in
\cref{fig:tool} integrates the $\nu$\textsf{Scr} codebase with
\textsf{STScript}.
This is is a Scribble extension---also based on multiparty session
types, but relying on the \textsf{ScribbleJava} implementation
\url{http://www.scribble.org}.  Building the API-generation of \tool{}
on top of the one of \textsf{STScript} has been a conventient choice:
\textsf{STScript} targets distributed web development directly and
offers a full implementation for generating TypeScript APIs from
$\nu$\textsf{Scr}-projected CFSMs.

The result of our development is \tool{}, of which
we list the distinctive features.
\begin{itemize}
\item \emph{Scope.} \tool{} specifically targets TypeScript
and enables safe distributed web development.
\item\emph{Input.} The user specifies the global protocol in the Scribble
language and picks one of the communicating participants as the server.
\item \emph{Correctness.} \tool{} relies on the flexible theory of
c-automata: the protocol in input is translated into a
c-automaton, which, if well-formed,
is then projected onto CFSMs.
\item \emph{APIs Generation.} From each CFSM, \tool{} generates the
TypeScript API for the respective role.
\item \emph{Safe Endpoint Implementation.} The distributed implementation
of the participants, using the generated APIs, is guaranteed to
be deadlock and lock free by the underlying theory.
\end{itemize}
%\begin{resub}
In our first implementation of \tool{} (\url{https://github.com/Tooni/CAScript-Artifact}), we provide
three simple examples: an “adder” (the client sends to the server,
in a loop, two numbers to be added), a simple contract protocol, and
the OLW, which we have used as a running example, since
it carries and shows all the core features of our novel theory,
and, in particular, selective participation (see also the discussion
at the beginning of this section).
Furthermore, we have provided
a small tutorial in the README file of \tool{},
to guide the user through the implementation of their own protocols.

It is worth mentioning that a first extension of \tool{}
is under development (see
also \iftr\cref{subsec:toolext}\else\cite{glsty22TR}\fi):
%Indeed, although
current implementation,
based on previous work
\cite{ZFHNY2020},
%right now
%\begin{resub}
%there is an implementation that
allows
%\end{resub}
the generation of APIs for Scribble protocols with assertions. % as well. %,
However, the necessary
extension of the function $\cat[]$ in \cref{sec:global}
to assertions, as well as subsequent
consistency checks, %their check for consistency has
have not been implemented
yet. While conceptually straightforward, in %from the definitions, in
practice %in order to do it
one needs to integrate the \tool{} toolchain with tools manipulating
%and using
logical formulae such as SAT solvers in order to implement
the check for the consistency property (cf. \cref{def:wa}).

To conclude, we have developed the first version of Scribble
based on \emph{choreography automata}. It improves on
the flexibility of traditional implementations of
multiparty session types, by
accomodating for \emph{selective participation}, and it
integrates previous developments with our new theory:
the $\nu$\textsf{Scr} toolchain with
the TypeScript support provided by \textsf{STScript}.
%\begin{resub}
On the one hand, our toolchain enables verified communication for
web development with selective participation, on the other hand it
paves the road to interesting extensions,
e.g., fully capturing the asynchronous semantics of websockets
(see \cref{sec:conc}), or
supporting assertions and the design-by-contract approach,
as discussed above.

%\end{resub}

%% Last part: what is missing is adapting the well-formedness
%% checks to this extended theory. Here, check how Fangyi added the
%% checks... and the theory... and finish the blabla.

%Ask Neil what interesting about this and which checks needs to be performed.

%%% Local Variables:
%%% mode: latex
%%% TeX-master: "main"
%%% End:

%%% subsec: TypeScript Programming with MPSTs
%%% subsec: Tool support

%% \iftr
%% \begin{toappendix}
%%   \input{imp-dbc}
%% \end{toappendix}
%% \fi

\section{Related Work}\label{sec:related}
% \eMcomm[{moved in intro following suggestion of \#A}]{}{\color{orange}C-automata have been proposed in~\cite{BarbaneraLT20} as a
% syntax-independent model of global views of choregraphies which could
% easily be adopted by practitioners.
% %
% Automata models of local behaviour abound in the literature (see
% e.g.,~\cite{bz83,dah01}).
% %
% We refer the reader to~\cite[Section 6]{BarbaneraLT20} for a comparison
% between c-automata and other models based on automata.
% }
%
Conditions similar to our well-branchedness and well-sequencedness arise naturally in investigations on choreographies and their
realisability.
Uniqueness of choice selector is commonly imposed syntactically (as in
\cref{sec:global}) in several multiparty session types (MPSTs) formalisms
(e.g., \cite{HYC08,bhty10,cdyp16,SD19,ZFHNY2020}) and
also adopted in global graphs \cite{dy12,gt18}, and in
choreography languages in general (cf. the notion of \emph{dominant
  role} in \cite{QiuZCY07}).
Also, notions close to well-sequencedness occur quite
naturally in \quo{well-behaved} choreographies (e.g., the
notion of {\em well-informedness} of \cite{BultanF08} in collaboration diagrams).
A distinguishing element of our notion of well-branchedness is that we admit protocols where
disjoint groups of participants may concurrently engage in a choice.
This generalises (and corrects) the notion of well-branchedness
in~\cite{BarbaneraLT20} and, to the best of our knowledge, is not supported
in any other choreographic framework.

Global graphs~\cite{dy12,gt19,gt18,lty15} are another model of global
specifications.
We refer the reader to~\cite{BarbaneraLT20} for a comparison between
c-automata and global graphs.

The first work advocating a design-by-contract
framework for MPSTs is~\cite{bhty10}.
Asserted c-automata have been strongly inspired by it.
In particular, our notion of consistency (cf. \cref{def:wa}) can be
seen as a generalisation of \emph{well-assertedness} in~\cite{bhty10}.
More recently, ideas similar to the one in~\cite{bhty10} have been
developed in~\cite{ZFHNY2020}, where refined MPSTs
have been proposed.
The results of these papers are in the vein of guaranteeing properties
of programs by a behavioural type system ensuring
communication soundness in presence of data dependencies.

Besides the added flexibility of c-automata with respect to structured
formalisms discussed in the Introduction, ac-automata do not suffer from the
constraints imposed on global types in~\cite{bhty10,ZFHNY2020}.
More precisely, interactions guarding choices in~\cite{bhty10,ZFHNY2020}
syntactically restrict to a unique partner of
the \emph{selector} (i.e., the participant choosing the branch to follow).
On the contrary, (asserted) c-automata do not have such restriction.
For instance,
\[\begin{tikzpicture}[node distance=.3cm and 4cm]
	 \tikzstyle{every state}=[cnode]
	 \tikzstyle{every edge}=[carrow]
	 \node[state, initial, initial text = {}] (q0) {$q_0$};
	 \node[state, above right = of q0,yshift=-.25cm] (q1) {$q_1$};
	 \node[state, below right = of q0,yshift=.25cm] (q2) {$q_2$};
	 \node[state, right = 8cm of q0] (q3) {$q_3$};
	 \path (q0) edge node[above] {$\gint$} (q1);
	 \path (q1) edge node[above] {$\gint[q][p][n]$} (q3);
	 \path (q0) edge node[below] {$\gint[q][p][n]$} (q2);
	 \path (q2) edge node[below] {$\gint$} (q3);
  \end{tikzpicture}
\]
is a well-branched c-automaton which would be ruled out by all the
choreography models based on global types we are aware of.
%
%XXX Moreover, c-automata allow choices where disjoint groups of participants
%engage in several choices.
%
%This feature is not present in any of the other multiparty session
%frameworks we are aware of.
%
Both~\cite{bhty10,ZFHNY2020} rely on a merge operator to guarantee
well-formedness (and projectability) of global types.
This is an obstacle for selective participation which our notion of
well-branchedness (cf. \cref{def:wb}) overcomes.
We also note that our notion of knowledge is more general than
the one in~\cite{bhty10}.
In fact, as observed in~\cite{ZFHNY2020}, the notion of history
sensitivity in~\cite{bhty10} does not allow a participant to
know variables fixed in interactions it is not involved in.
Like for  refined MPSTs, asserted c-automata
do not have this limitation and can in fact deal with protocols
like the one in Example 4.1 in~\cite{ZFHNY2020}.

Our theoretical work sees its first application
in the development of \tool{}, a toolchain
for communication-safe web development.
\tool{} takes the popular \emph{top-down approach},
following the original methodology of
MPSTs \cite{HYC08}.
The top-down approach enables \emph{correctness-by-construction}:
a developer provides a global description for the whole
protocol; by projecting the global protocol,
APIs are generated from local CFSMs,
which ensure the safe implementation of each participant.
MPSTs toolchains that take the top-down approach
have seen multiple implementations
and targeted a variety of
mainstream programming languages, such as (in no particular order)
Java \cite{HuYoshida2016,HuYoshida2017,mungo},
OCaml \cite{ImaietAl2020}, Go \cite{CastroetAl2019},
Scala \cite{ECOOP17,VieringetAl2021},
F\# \cite{NeykovaetAl2018},
F$^\star$ \cite{ZFHNY2020} and
Rust \cite{CYV2022,LNY2022}.
Like \tool{}, most of the above implementations
rely on the Scribble protocol description language
\cite{scribble-paper,FeatherweightScribble,nuscr}
(\url{http://www.scribble.org}, \url{https://nuscr.dev/}).
More relevant to this work is \cite{stscript},
in which the authors develop \textsf{STScript},
a full toolchain that applies such top-down
methodology and targets TypeScript for web development.

All the implementations above
are based on MPSTs;
they exploit the equivalence
between local types and CFSMs \cite{dy12,dy13}
to generate APIs for all the participants.
In \cite{HuYoshida2017}, \emph{explicit connections}, similar to our
selective participation, have been introduced
in Scribble, and more recently \cite{HFDG-ECOOP2021} uses
an analogous approach to implement adaptations
for an actor domain-specific language.
Both \cite{HuYoshida2017} and \cite{HFDG-ECOOP2021} need to add
explicit disconnections and connections to the
syntax of Scribble.
In \tool{} (\cref{sec:apply}), we have integrated the
theory of c-automata into the $\nu$\textsf{Scr} toolchain \cite{nuscr},
to allow for more flexible protocols, where participants
may appear only in selected branches after a choice, with no need to
change the Scribble syntax.

\section{Conclusion and Future Work}\label{sec:conc}
We have presented a flexible framework to describe protocols in a
setting of c-automata combining selective participation to branches of
choices and assertions supporting design-by-contract.
This allows us to model non trivial examples such as the OnLineWallet,
and ensures faithful realisability.
In fact, we exploited the flexibility of c-automata to generalise
well-branchedness (so to account for selective participation) and to
transfer the DbC approach~\cite{bhty10} (so to account for
\emph{data-aware} protocols).
Remarkably, the fact that c-automata are finite-state models does not
allow us to fully capture Scribble.
Nonetheless, a semi-decidable approach has been considered
(cf. \cref{sec:global}) which becomes effective when restricting to
protocols without interplay between consecutive independent
interactions and recursion.
%\begin{resub}\il{
    More precisely, it should not be possible to split a recursive protocol into groups of interactions with disjoint participants.
%    the rule \textsc{[Pass]} from the global
%  types semantics in \cref{fig:globallts}.
	 This restriction mildly affects applicability: indeed, to
	 faithfully implement such specifications one would need
	 infinite-state systems of CFSMs, while ours are finite-state.
%}\end{resub}
%
Also, a clear advantage of our approach is that we can verify more
general conditions for Scribble specifications that can be faithfully
mapped on c-automata.

We implemented our theory by allowing Scribble protocols to be
translated into c-automata, checked for well-formedness, and finally
used to derive APIs for TypeScript programming.
The flexibility of c-automata has been instrumental to capture
Scribble~\cite{scribble-paper,FeatherweightScribble,nuscr}
specifications.
Scribble notation (and semantics) may be not easy to grasp for
practitioners as it involves a non-trivial amount of technicalities.
%related to variable scoping or operational semantics.
%
Hence, defining and understanding well-formedness conditions on
Scribble could not be straightforward.

Our framework can be immediately used in practice in interesting
examples:
%\begin{resub}
the design of a variety of existing web services (e.g., for
authentication or transactions) include selective participation; with
the OLW implementation, we witness how protocols carrying this feature
can be specified in \tool{} (which from these generates APIs for
implementations).
%\end{resub}
Nonetheless, we envisage some extensions (see \cref{subsec:api-gen} and
\iftr \cref{subsec:toolext} \else \cite{glsty22TR} \fi for details).
%
%% Indeed, although
%% current implementation
%% %right now
%% %\begin{resub}
%% %there is an implementation that
%% allows
%% %\end{resub}
%% the generation of APIs for Scribble protocols with assertions as well,
%% their check for consistency has not been implemented
%% yet. While conceptually straightforward, in %from the definitions, in
%% practice %in order to do it
%% one needs to integrate the \tool{} toolchain with tools manipulating
%% %and using
%% logical formulae such as SAT solvers in order to implement
%% the check for the consistency property (cf. \cref{def:wa}).

Our focus is on selective participation and design-by-contract.
Hence, for simplicity, we consider synchronous semantics.
%\begin{resub}
  %  The discrepancy you flag (due to the fact that
\tool{} builds instead on an asynchronous implementation of Scribble
\cite{stscript}, which
makes our results applicable only to protocols in which
%(as is the case in the OLW example)
asynchronous executions do not break the causal relations imposed by
the synchronous semantics so that choices are affected.  This is the
case for the case studies in the artifact, including our running
example OLW. The discrepancy disappears if a synchronous transport
layer (e.g., http) replaces web
sockets. %% Note that we also plan to extend our theory to the asynchronous case. We leave this as future work, since this is beyond the scope of one paper.
To increase the applicability of \tool{}---and also because of its
theoretical interest, we plan
%\end{resub}
%From a theoretical point of view, one would also like
to extend the
results to cover an asynchronous communication model based on
queues.
While the general structure of the theory remains the same,
well-branchedness needs to be updated since send and receive actions
would not be symmetric anymore.
E.g., a participant that only occurs in one branch of a choice, thanks
to selective participation, needs to interact with a fully-aware
participant by performing a receive, while right now it can also
interact through a send action.
%
%\begin{resub}
  We conjecture that the extension to asynchronous semantics does not
  affect the treatment of DbC in ac-automata.
  In fact, assertions are guaranteed by the sender and relied upon by
  the receiver (hence, the nature of communication is orthogonal to
  the flow of data).
%\end{resub}

%\begin{resub}
  Our methodology follows the top-down
  software development approach of choreographies
  (cf. \cref{sec:intro} and \cref{sec:related}).  An interesting
  direction for future work is to develop an analysis of
  existing APIs; for instance, by extracting an abstract
  representation of the API, its conformance could be checked against
  a projection of the global specification.  Such design would improve
  on the applicability of our theory, for analysing and reusing
  existing developments.
%\end{resub}

%%% Local Variables:
%%% mode: latex
%%% TeX-master: "main"
%%% End:

%%
%% Bibliography
%%

%% Please use bibtex,

\bibliography{biblio}

\iftr
\appendix
%\subsection
\section{Extending the Toolchain to Design-by-Contract}\label{subsec:toolext}
%\paragraph*{Refinements}

%Comparison with Anson's work continues: implem and theory (CA vs MPST).

In \cref{sec:achor} we have shown how our c-automata theory
can be endowed with assertions, thus supporting design-by-contract.

We have started engineering an extension of
\tool{}, which combines selective participations and design-by-contract.
In \cite{ZFHNY2020}, the authors extend
the Scribble language with assertions (\emph{refinements});
we integrate a similar approach in our specification language.
As an example, we show here a description of
the OLW protocol with annotated assertions (\cref{fig:ow-assertions}).
Assertions are used, for example, (line 5) to enforce that the integer
$\amsg[account]$ is a six-digit number,
or to allow for a finite number of login attempts: the integer
$\amsg[try]$ is initiated (line 4) as $0$; then
incremented at each following attempt (lines 19 and 20); finally,
when $\amsg[try]=3$ (line 22), the login is
denied. Such prototype extension has allowed us to
combine design-by-contract and selective participation
in Scribble protocols for selected examples, and to generate TypeScript APIs
for multiple participants, with assertions to guide
the developer's implementation process. %% This part of the tool relies
%% on (an adaptation of) related work \cite{ZFHNY2020};
A future extension of the function $\cat[]$ in \cref{sec:global},
to assertions, will allow us to have a more comprehensive version of \tool{}.

\begin{figure}
\begin{lstlisting}[language=Scribble]
(*# CheckDirectedChoiceDisabled, RefinementTypes, ValidateRefinementProgress, ValidateRefinementSatisfiability #*)

global protocol OnlineWallet(role Wallet, role Customer, role Vendor) {
  rec AuthLoop [try<Customer, Wallet>: int = 0] {
    login(account: int{account >= 100000 && account < 1000000}) from Customer to Wallet;
    pin(pin: int{pin >= 1000 && pin < 10000}) from Customer to Wallet;
    choice at Wallet {
      login_ok() from Wallet to Customer;
      login_ok() from Wallet to Vendor;
      request(bill: int{bill > 0}) from Vendor to Customer;
      choice at Customer {
        authorise() from Customer to Wallet;
        pay(payment: int{payment = bill}) from Customer to Vendor;
      } or {
        reject() from Customer to Wallet;
        reject() from Customer to Vendor;
      }
    } or {
      login_retry(msg: string{try < 3}) from Wallet to Customer;
      continue AuthLoop [try + 1];
    } or {
      login_denied(msg: string{try = 3}) from Wallet to Customer;
    }
  }
}\end{lstlisting}\vspace*{-10mm}
\caption{Scribble Protocol for the OLW with Assertions}
\label{fig:ow-assertions}
\end{figure}

\fi

\end{document}

%%% Local Variables:
%%% mode: latex
%%% TeX-master: t
%%% End: